\RequirePackage{fix-cm}
\documentclass[smallextended]{svjour3}       
\usepackage[utf8]{inputenc}
\usepackage[colorlinks,
            linkcolor=red,
            anchorcolor=red,
            citecolor=blue
            ]{hyperref}
\usepackage{graphicx,amsmath,fullpage,amssymb,amsthm,amsfonts,color}
\usepackage{enumerate}

\numberwithin{equation}{section}
\numberwithin{theorem}{section}
\numberwithin{definition}{section}
\numberwithin{corollary}{section}
\numberwithin{lemma}{section}
\numberwithin{remark}{section}

\theoremstyle{definition}

\usepackage[toc,page]{appendix} 
\usepackage[utf8]{inputenc}

\renewcommand{\epsilon}{\varepsilon}

\begin{document}

\title{Asymptotic Behavior of a Sequence of
Conditional Probability Distributions and the Canonical Ensemble \thanks{H.Q. is partially supported by the Olga Jung Wan Endowed Professorship. Y.Z. is partially supported by NSF DMS-1712630.   
}
}


\author{Yu-Chen Cheng \and Hong Qian \and Yizhe Zhu}


\institute{Yu-Chen Cheng \at Department of Applied Mathematics, University of Washington, Seattle, WA 98195-3925 \\
           \email{yuchench@uw.edu}         
               \and
           Hong Qian \at
              Department of Applied Mathematics, University of Washington, Seattle, WA 98195-3925   \\ 
              \email{hqian@uw.edu} 
           \and
          Yizhe Zhu \at
              Department of Mathematics, University of California, San Diego, La Jolla, CA 92093  \\     
               \email{yiz084@ucsd.edu}     
}

\date{Received: date / Accepted: date}

\maketitle

\begin{abstract}
The probability distribution of a function of a subsystem conditioned on the value of the function of the whole,  in the limit when the ratio of their values goes to zero, has a limit law: It equals the unconditioned marginal probability distribution weighted by an exponential factor whose exponent is uniquely determined by the condition.  We apply this theorem to explain the canonical equilibrium ensemble of a system in contact with a heat reservoir. Since the theorem only requires analysis at the level of the function of the subsystem and reservoir, it is applicable even without the knowledge of the composition of the reservoir itself, which extends the applicability of the canonical ensemble. Furthermore, we generalize our theorem to a model with strong interaction that contributes an additional term to the exponent, which is beyond the typical case of approximately additive functions. This result is new in both physics and mathematics, as a theory for the Gibbs conditioning principle for strongly correlated systems. A corollary provides a precise formulation of what a temperature bath is in probabilistic terms.
\end{abstract}

\section{Introduction}
\label{ch:introduction}
The canonical ensemble with mechanical energy distribution 
in an exponential form is the centerpiece of equilibrium statistical mechanics.  It represents a {\em weight} for a microstate of a system in thermal equilibrium with its surrounding {\em heat bath} at a fixed temperature, where the bath is usually considered much larger in comparison.  The theory has wide applications from condensed matter physics to biophysical chemistry \cite{dill2012molecular,bialek2012biophysics}.  In textbooks, there are currently two heuristic justifications for the exponential factor.  One is the original derivation by L. Boltzmann in 1877 based on an ideal gas \cite{boltzmann_1877}, 
another is based on the notion of a large heat bath and a
small system within, extensively discussed by J. W. Gibbs in 
his 1902 {\it magnum opus} \cite{gibbs1902elementary}.  After an extensive discussion of the properties of an {\em invariant measure} including demonstrating it has to be a function of the mechanical energy, however, Gibbs did not attempt to derive the canonical distribution; rather he simply stated that an exponential form ``seems to represent the most simple case conceivable''.

Boltzmann's derivation was based on the idea of {\em most probable frequency} under the constraint of given total energy.  In the
process he recognized the entropy $S = -N \sum_i f_i \log f_i$
from the multinomial distribution, where $N$ is the number of gas molecules, and $i$ represents a distinct molecule state 
with kinetic energy $e_i$.  This derivation preceded both the modern
theory of large deviations \cite{dembo1998large,touchette2009large}
as well as the principles of maximum entropy (MaxEnt) championed by 
E. T. Jaynes \cite{jaynes2003probability,presse2013principles}.  In terms of the contraction principle in the former, Boltzmann computed the large-deviation rate function for a sample frequency conditioned on a given sample mean of energy instead of obtaining the rate function for the random variable.  This approach has now been 
made rigorous under the heading of the {\em Gibbs conditioning principle} \cite{kesten2012random,dembo1998large}.  MaxEnt, on the other hand, plays a pivotal role in information theory and machine learning \cite{jaynes1957information,ackley1985learning}.  In the 1980s, Boltzmann's logic was also rigorously developed into providing a connection between maximum entropy and conditional probability \cite{zabell1980rates,van1981maximum}. 

Gibbs' theory for the canonical distribution was based on the concept of {\em heat bath}. In \cite{gibbs1902elementary}, he noted that the
distribution with the exponential form had ``the property that when the system consists of parts with separate energies, the laws of the distribution in the phase of the separate parts are of the same nature". Having energy $E_A$ for the microstate $A$ of the small system and $E_B$ for the microstate $B$ of the heat bath, Gibbs assumed the phase-space distributions follow (i) additivity: $P(A, B) = P(A + B) $ (ii) independence: $P(A, B) = P(A)P(B)$. Under those two assumptions, the only possible probability distribution for $A$ is exponential: $P(A)\propto e^{\lambda E_A}$. Furthermore, all small systems in contact with the same bath share the same parameter $\lambda$, which means they are of the ``same nature".  By assuming that every small system follows  the {\em conjugate distribution laws} (a family of single parameter exponential priors), A. Ya. Khinchin \cite{aleksandr1949mathematical} rigorously proved Gibbs' assertion of the common $\lambda$ and further showed that it is determined by the given total energy.

The aim of the paper is to find a rigorous origin of the exponential weight itself for the canonical distribution from the standpoint of a heat bath. We were inspired by a very widely used derivation in standard statistical physics textbooks - based on Taylor’s expansion of the entropy function of a heat bath \cite{landau1958statistical,huang1975statistical,martin1979statistical}. The present work formulates this approach rigorously in probabilistic terms and then gives a proof. We indeed have obtained a rather general new mathematical theorem. The results can be applied back to particular scenarios in statistical physics under corresponding assumptions. Our theorems have clarified the notions of additivity, independency, and the vague “same natures of systems”. The last is actually a corollary of the existence and uniqueness of a single parameter in the exponential form of the canonical distribution, and independency is equivalent to  additivity of energy functions of two systems during the map from a phase space to its corresponding energy space. We shall emphasize that independency of two systems is a special case of our theorem; the parameter then only depends on fluctuations of the heat bath but independent of the small system.

Our results are obtained based on two mathematical ideas: 
{\em conditional probability} and {\em asymptotics}.
We use a {\em Gedankenexperiment} to illustrate the 
crucial role of the former - {\em conditional probability } - in our theorems:
Let $Z := X + Y$, where $X$ is a random variable for some function (e.g., energy) of a subsystem and $Y$ describes the same quantity in the heat bath. If one is only interested in the {\em static statistics} of $X$, there is a way to set up an experiment: Let $Z(t)$ be a fluctuating total mechanical energy as a function of time, and its distribution has a support on $D \subseteq \mathbb{R}^+$, but one selects only those measurements for $X(t)$ that simultaneously have $Z(t) \in I \subseteq D$. In the language of mathematics, this thought experiment is about the conditional probability of $X(t)$ conditioned on the event $Z(t) \in I$. Why is this thought experiment regarding conditional probability very much in line with the physicist's picture of a canonical ensemble? The answer is in the idea of time-scale separation, which involves three different time scales. The first time scale is for the subsystem $X(t)$ to reach its equilibrium, the second time scale is to restrict the total system $Z(t)$ to be fluctuating inside a finite interval $I$, and the third time scale is for $Z(t)$ to reach its equilibrium. And the first one is much shorter than the second one, which is much shorter than the third one. Based on this  framework of time-scale separation, the canonical ensemble is  the statistical ensemble that represents the possible outcomes of the system of interest on the second time scale, i.e., when the subsystem has reached its equilibrium, but the total system is still ``constrained" in a certain interval. 

In fact, having its own stationary distribution of the total system (if it evolves long enough) is very significant for the theory of conditional probability for two reasons: (1)  knowing the fluctuation of the large system is necessary to define the conditional probability mathematically and (2) to perturb the given condition of the total system to see how it has effects on the subsystem is the essence of our theory of the canonical distribution. In other words, even though the original problem is only about the behavior of $X(t)$ when $Z(t)  \in I$, if we have more information of $Z(t)$ outside of $I$, we are able to seek a deeper understanding of the original problem. Not only for the canonical ensemble, this idea of treating a given constraint (parameter) as a variable with distribution has  also been widely used in many other fields, for example, comparing the quenched and annealed invariance principles for the random conductance model \cite{barlow2016comparison}, and in studying the initial-condition naturalness in the case of statistical mechanics \cite{pittphilsci15757}.
 
Mathematically using conditional probability to understand Gibbs measure has a long history, see O. E. Lanford \cite{lanford1973entropy}, O. A. Vasicek \cite{vasicek1980conditional}, H. O. Georgii \cite{georgii2011gibbs}, and H. Touchette \cite{touchette2015equivalence}. 
In particular, on the basis of Boltzmann's logic, using asymptotic conditional probability to describe the canonical ensemble has been well-established through the Gibbs conditioning principle \cite{kesten2012random,dembo1998large}.  More discussion of 
this is provided in Section \ref{ch:preliminary} for a contradistinction with our own work.  In brief, the Gibbs conditioning principle addresses this question: Given a set $A \in \mathbb{R}$ and a constraint $Z_n \in A $, what are the limit points of the conditional probability 
\begin{align} \label{Q:gibbs.cond.prin}
    \mathbb{P}(X_1 \leq x \mid Z_n \in A) \quad \text{as} \ n\rightarrow \infty \ ?
\end{align}
In Equation \eqref{Q:gibbs.cond.prin}, $Z_n = \frac{1}{n} \sum_{i=1}^n X_i$, where $X_i$ are independent and identically distributed random variables (i.i.d. random variables). We can identify that \eqref{Q:gibbs.cond.prin} is very similar to our setup for the canonical distribution if we consider $Z_n := \frac{X_1}{n} + Y_n$, where $ Y_n = \frac{1}{n} \sum_{i=2}^{n} X_i$ is the measurable function of the heat bath in our approach. However, $Y_n$ in our setup could be defined in a much more general way: we only require that $Y_n$ converges to some random variable $Y$ in distribution (or the law of $Y_n$ satisfies a large deviation principle) rather than has a special form as the sum of  independent and identically distributed random variables. In other words, the present work is not a simple refinement of the Gibbs conditioning principle. Here we give a concrete example to which our theorems can be applied but not the Gibbs conditioning principle: Let $\tilde{\zeta}_n = \xi_1 + \eta_n $ and $\eta_n  = \sum_{i=2}^n \xi_i$, where $\{\xi_i\}_{i=1}^n$ are strongly correlated and not identically distributed, and let $\zeta_n$ be  $\tilde{\zeta}_n$ with appropriate shifting and scaling such that $\zeta_n$ has a limiting distribution (or satisfies a large deviation principle). Subject to these conditions, the Gibbs conditioning principle would not be applicable to find the limit points of the conditional probability
\begin{align} \label{Q:gibbs.cond.prin1}
    \mathbb{P}(\xi_1 \leq x \mid \zeta_n \in A) \quad \text{as} \ n\rightarrow \infty.
\end{align}
The present work will show that the canonical distribution in this non-i.i.d. example could still exist as a good approximation  (Corollary \ref{cor:corollary 1}) or the limiting distribution (Corollary \ref{cor:limit.smooth}, Corollary \ref{cor:limit.ldp}) of the conditional probability \eqref{Q:gibbs.cond.prin1}. In fact, the setup for our theorems is very general in statistical mechanics: (i) a subsystem in contact with a relatively large heat bath, which is including but not limited to the model of a sum of many independent and identical subsystems; and (ii) the subsystem and the heat bath can have weak or strong interaction.

Back to Equation \eqref{Q:gibbs.cond.prin}, it seems that either using the Gibbs conditioning principle or using our approach to derive the canonical distribution, both sides are asking a very similar question: what is the asymptotic behavior of a conditional probability? However, based on the more general setup of the conditional probability, our approach to the asymptotic behavior of this conditional probability is very different from the Gibbs conditioning principle. For the Gibbs conditioning principle, it transforms the original problem to a sampling problem: what are the limit points of
\begin{align}\label{Q:gibbs.cond.prin2}
 \mathbb{E}\left[ L_n \mid L_n \in \Gamma \right] \quad \text{as} \ n \rightarrow \infty \ ?
\end{align}
In Equation \eqref{Q:gibbs.cond.prin2}, $L_n = \frac{1}{n} \sum_{i=1}^n \delta_{X_i}$ is the corresponding empirical measure for $Z_n$ and $\Gamma = \{ \gamma : \int x \gamma(dx) \in A  \}$ is the corresponding constraint of $Z_n$.  In fact, even though this approach is called the ``Gibbs" conditioning principle, its logic exactly follows Boltzmann's derivation of the canonical ensemble. As a consequence of  the Gibbs conditioning principle, it provides a mathematical foundation of why using the maximum entropy principle with certain constraint works to find the canonical distribution \cite{zabell1980rates,van1981maximum}. 

On the other hand, our approach is direct to find the asymptotic behavior of conditional probability \eqref{Q:gibbs.cond.prin} on the basis of two things: (i) a measurable function of the subsystem is asymptotically small relative to the function of the whole and (ii) the distribution of the measurable function of the heat bath converges to a limiting distribution by appropriate shifting and scaling. Intuitively, under this framework, the distribution of the measurable function of the subsystem shall consist of its unconditional distribution and a weight from a linear approximation of the limiting distribution of the measurable function of the heat bath. As we mentioned above, our approach follows Gibbs' theory for the canonical distribution, which involved the idea of ``heat bath" that contributes a ``bias" to the system. The common point of our approach and the Gibbs conditional principle is that both sides started with a very similar question of fundamental importance in statistical mechanics and adopted the concept of conditional probability to describe that problem. However, the method of solving the problem on each side has a very different philosophy, the Gibbs conditional principle is about counting statistics by Boltzmann's logic, and ours is inspired by the idea of a heat bath from Gibbs. 

Besides the conditional probability,  we adopt a very powerful mathematical technique in our theory: asymptotics. Indeed, asymptotics is not only a mathematical technique but also the essence of statistical mechanics. The purpose of statistical mechanics is to derive equilibrium properties of a macroscopic system with enormous numbers of molecules $N$ and occupying a very large volume $V$, then that macroscopic equilibrium thermodynamics is an {\em emergent phenomenon} in the limiting case when $N \rightarrow  \infty$ and $V \rightarrow  \infty$. Following on from this concept, we shall show that the emergence of an exponential factor in the canonical ensemble is also a result of a limit law according to the probability theory. Take an analogy, our limit theorem is to the exponential form of the canonical distribution what the central limit theorem is to a normal distribution. As with every limit theorem, we have to define how our assumptions depend on $n$ carefully. In our work, as $n$ increases, a measurable function of the subsystem becomes ``relatively small" compared with the total system. Based on this main assumption, we obtain two significant results: (i) For a sufficiently large $n$, a conditional distribution can be well-approximated by its unconditional distribution weighted by an exponential factor, and (ii) a sequence of conditional distributions converges to a limit which is the unconditional distribution weighted by a unique exponential factor. 

We obtain two theorems regarding the first result in Section \ref{ch:existence}, and they provide the existence of the canonical distribution when a system is contained in a finitely large total system ($n$ is sufficiently large). Furthermore, we obtain two limit theorems regarding the second result in Section \ref{ch:uniqueness}, and they provide the existence of a unique canonical distribution when the system is contained in an infinitely large total system ($n \rightarrow \infty$). In comparison with Section \ref{ch:uniqueness}, Section \ref{ch:existence} only requires weaker conditions, but the exponential form in the canonical distribution may not be unique since there could be more than one sequence having the same asymptotic behavior. On the other hand, Section \ref{ch:uniqueness} requires stronger conditions, but it gives us a unique canonical distribution in the limit, and this distribution can be applied back to approximate the conditional probabilities for all finitely large $n$. This result can be regarded as an example that the limit theorems from probability predict the laws of nature. Here, we would like to quote from P. W. Anderson \cite{anderson1972more}  ``Starting with the fundamental laws and a computer, we would have to do two impossible things - solve a problem with infinitely many bodies, and then apply the result to a finite system - before we synthesized this behavior."  Our idea echos Anderson's view: To find the limiting behavior of a sequence of conditional probability distributions and apply it back to the distribution of a subsystem contained in a finitely large total system with some fluctuations, and this is how it is used as a scientific theory. 

\subsection{The equivalence of ensembles}

Our work is another way to consider the theory of equivalence of ensembles. As far as we know,  Khinchin's derivation of the canonical ensembles in 1949 \cite{aleksandr1949mathematical} for a subsystem of a large isolated system by a local central limit theorem was the origin of the equivalence of ensembles. Then Dobrushin and Tirozzi in 1977 \cite{dobrushin1977central}  extended Khinchin's result from a classical idea gas to a Gibbs random field. In 1979,  Martin-L\"of \cite{martin1979equivalence,martin1979statistical} further related the microcanonical, canonical, and grand canonical ensembles in the thermodynamic limit when the volume of classical lattice systems tends to infinity. In the 1990s, beyond the scale of the central limit theorem, Deuschel et al. \cite{deuschel1991microcanonical} and Georgii \cite{georgii1995equivalence} showed the equivalence of ensembles on the scale of the large deviation principle. Tasaki \cite{tasaki2018local} recently established the equivalence on the level of local states for large but finite quantum spin systems. A comprehensive introduction to infinite-volume Gibbs measures can be found in Chapter 6 in the textbook by Friedli and Velenik \cite{friedli2017statistical}, and the discussion of the equivalence of ensembles is in Section 6.14.1. 

Recently, a full survey of the equivalence of ensembles at the levels of thermodynamic, macrostates, and measures was presented by Touchette \cite{touchette2015equivalence}. We shall note that discussions on the equivalence of ensembles at the thermodynamics level can also be traced back to the textbook on statistical mechanics by Hill \cite{terrell1987statistical}. In the book, Hill showed the thermodynamic equivalence of ensembles for systems having only a single most probable energy value. In Touchette's recent work, the equivalence was extended to other macrostates, e.g., the mean magnetization of a spin system. This extension was given by the superposition of a mixture of microcanonical ensembles to represent the canonical ensemble of macrostates. Under certain conditions, the equivalence at the macrostate level and the equivalence at the measure level are equivalent. In the language of modern probability, the correspondence between the equivalence at the macrostate level and the equivalence at the measure level is by the Portmanteau theorem  \cite{billingsley2013convergence} for equivalent statements of weak convergence of measures.
We shall emphasize that the conditional distribution of the state of a small subsystem converging to the canonical distribution becomes a corollary of the equivalence of ensembles between the microcanonical ensemble and canonical ensemble at the measure level based on the assumption that the state of a small subsystem is chosen at random with a {\em uniform distribution} in the large whole system.

The essential difference between our approach  and the previous approaches for equivalence of ensembles   is that we don't assume a uniform distribution of the state of a small subsystem in the large container. This assumption is equivalent to say that the heat bath (the large container - the subsystem) has to be considered as identical copies of the subsystems, which was usually given in the previous work for classical ideal gas systems or Gibbs random fields. In contradistinction to this assumption, our theorems treat the subsystem and its heat bath as two random variables via a measurable function, i.e., we only care about the effect of the ``whole" heat bath on the subsystem with respect to that function. 

We want to indicate that applying our mathematical theory to physics is new and original since it extends the applicability of the canonical ensemble: by the pushforward measure (up to a prefactor) via a measurable function, we can derive the canonical distribution of a subsystem without assuming a uniform structure of the whole system. For example, we can apply our results to approximate the distribution of certain measurable functions of a small defect within a material.  We only require the subsystem (the defect) is small relative to its heat bath with respect to the values of the measurable functions, which is different from treating the heat bath as infinitely large $n$ copies of the subsystem, interacting or not, in order to apply the microcanonical ensemble to the canonical ensemble. In biophysics, our theory can predict the distribution of side-chain conformational variations in protein structure \cite{butterfoss2003boltzmann,xiang2007prediction,miao2016quantifying}. Proteins in general have a non-uniform structure, so the canonical distribution of side-chain conformational variations can be justified by our theory but not the other approaches based on a uniform structure of the whole system.  

We further generalize our theorem to a model when a subsystem and its heat bath have {\em strong interaction} (the function is not additive), which is beyond the weakly interacting system (the function is approximately additive). This result is new in both physics and mathematics, as a theory for the {\em Gibbs conditioning principle} for strongly correlated systems. The present work formulates our theory rigorously in probabilistic terms in Section \ref{ch:main}, and then gives a proof in Section \ref{ch:proof}. 

In Section \ref{sec:gibbs}, we apply our theory to concrete examples in statistical mechanics, under two situations when a subsystem and its heat bath are independent or strongly correlated. Since our theory also provides a sharp and precise bound of the convergence rate of conditional probabilities, we use it to approximate the conditional Poisson distribution in Section \ref{sec:poisson}. To build a connection with the equivalence of ensembles using the techniques of the large deviation principle (LDP) \cite{lewis1995entropy} and the central limit theorem (CLT) \cite{dobrushin1977central}, we applied our theory back to particular scenarios when the heat bath can be treated as a sum of identical random variables. The LDP in Section \ref{sec:em.temp.ldp} or the CLT in Section \ref{sec:em.temp.clt} gives us a convergence of a sequence of random variables for the heat bath. Nevertheless, we want to emphasize that our theory does not require the LDT or the CTL in general. By proper scaling and shifting, if there exists a convergence of the heat-bath random variable with a smooth limiting distribution, our theory is still applicable. In Section \ref{sec:heatbath}, we provide a precise formulation of what a temperature bath is in probabilistic terms.

\subsection{Organization of the paper} 
We provide some useful theorems and definitions and explain our motivation in this problem in Section \ref{ch:preliminary}. In Section \ref{ch:main} we state and explain our main results. Proofs of the main results are provided in Section \ref{ch:proof}. In Section \ref{ch:application} we present several applications of our main theorems.

\subsubsection*{Notations} Throughout the paper, we will adopt the notations $a_n=o(b_n)$ when $\lim_{n\to\infty} 
\frac{a_n}{b_n}=0$, and $a_n=O(b_n)$ when $|a_n/b_n|$ is bounded by some constant $C>0$.  

For a set  $\Omega$, we use $C(\Omega)$ to represent the set of all continuous real functions on $\Omega$,  $C_b(\Omega)$ to represent the set of all bounded continuous functions on $\Omega$, and  $C^k(\Omega)$ to represent the set of all functions with continuous derivatives of order $k$ on $\Omega$. 

We sometimes use brief notations of probabilities in our proofs, e.g.,  $P_{X_n \mid Z_n}(x; I) = P( X_n = x \mid Z_n \in I)$. We always use $X_n$, $Y_n$, $Z_n$ to denote sequences of random variables, whose definitions might change in different theorems, but we will give their exact definitions before stating the theorems.




\section{Preliminaries}

\label{ch:preliminary}

\subsection{Maximum entropy and conditional probability}

We first recall the following classical results. Here we don't specify the regularity conditions in the statements of the two theorems below. For more details, see the original references.

\begin{theorem}[\cite{zabell1980rates}]
\label{thm:zabell.1980}
Let $\{ X_n \}_{n \in \mathbb{N}} $ be a sequence of independent and identically distributed (i.i.d.) random variables with continuous density $f(x)$, then under appropriate regularity conditions, we have
\begin{align}
\lim_{n\to\infty}P \left( X_1 \leq x \mid S_n  = n \mu + c_n \right)= P \left( X_1 \leq x \right),
\end{align}
where $S_n := X_1 + X_2 + \cdots + X_n$, $\mu := \mathbb{E} [X_1]$, $s_n^2 := \textnormal{Var}\left[ S_n \right]$, and $c_n = O(s_n)$.
\end{theorem}

\begin{theorem} [\cite{van1981maximum}] 
\label{thm:von.1981}
Let $\{ X_n \}_{n \in \mathbb{N}} $ and $S_n$ follow definitions in Theorem \ref{thm:zabell.1980}. Let $\alpha \in \mathbb{R}$ and let $f(x)$ be the density function of $X_1$, then under appropriate regularity conditions,
\begin{align}
\lim_{n\to\infty}P\left( X_1 \leq x \mid S_n  = n \alpha \right) = \left( \int_{-\infty}^x  e^{\lambda s} f(s) ds \right) / c(\lambda),
\end{align}
in which
\begin{align}
c(\lambda) = \mathbb{E} \left[ e^{\lambda X_1} \right] < \infty \quad \text{and } \quad  \alpha = \left( \int x e^{\lambda x} f(x) dx   \right) / c(\lambda).
\end{align}

\end{theorem}

Note that the parameter $\lambda$ is determined by the constraint
\begin{align}
\alpha = \left( \int x e^{\lambda s} f(x) dx   \right) / c(\lambda),
\label{pre:max.contraint}
\end{align} 
and the density $g(x) = e^{\lambda x}f(x)/c(\lambda) $ maximizes the entropy relative to the density $f(x)$ of $X_1$ given by
\begin{align} \label{rel.entropy}
    H(X_1) = - \int g(x) \log \frac{g(x)}{f(x)} dx,
\end{align}
with respect to the constraint that
\begin{align}
\left( \int x g(x) dx   \right) = \alpha.
\end{align}

We see that Theorem \ref{thm:zabell.1980} implies the convergence of the conditional probability distribution of $X_1$ to its unconditional distribution.  In this case, the sum of $X_i$ is conditioned on the scale of Gaussian fluctuations: $S_n=n\mu + c_n$, where $n\mu$ is the mean of $S_n$ and $c_n$ is in the order of standard deviation of $S_n$. On the other hand, we see that Theorem \ref{thm:von.1981} implies the convergence of the conditional probability distribution of $X_1$ to the (normalized) product of its unconditional distribution and the maximal entropy distribution $e^{\lambda x}$. The parameter $\lambda$ is determined by the condition ${S_n}=n \alpha $, which is on the scale of large deviations when $\alpha \neq \mathbb{E}[X_1].$ 

Theorem \ref{thm:von.1981} is a particular case of the {\em Gibbs conditioning principle}, which is the meta-theorem \cite{dembo1996refinements} regarding the conditional probability of $X_i$ given on the {\em empirical measure} of an i.i.d. $ \{X_i\}_{i=1}^n $
\begin{align}
\label{emp.measure}
L_n = \frac{1}{n} \sum_{i=1}^n \delta_{X_i}
\end{align}
belongs to some rare event such as 
\begin{align} \label{emp.measure.cond}
\int x L_n(dx) = \frac{1}{n} \sum_{i=1}^n X_i = \alpha \quad \text{and} \quad  \alpha \neq \mathbb{E}[X_1]. 
\end{align}
Using the empirical measure defined in \eqref{emp.measure} conditioned on the rare event \eqref{emp.measure.cond} to find the limit of conditional probability in Theorem \ref{thm:von.1981} turns out to be equivalent to find the limit 
\begin{align}
\gamma^* : = \lim_{n\rightarrow \infty}\mathbb{E}\left[ L_n \mid L_n \in \Gamma \right], \quad \Gamma = \{ \gamma : \int x \gamma(dx)  = \alpha  \}.
\end{align}
By the Gibbs conditioning principle, under appropriate regularity conditions, $\gamma^*$ minimizes the relative entropy 
$$H(\gamma \mid \mu_X) :=  \int d \gamma \log\left( \frac{d \gamma }{d \mu_X} \right), $$ where $\gamma \in \Gamma $ and $\mu_X$ is the law of $X_1$. In fact, this result implies the limit law derived in Theorem \ref{thm:von.1981}.

One of the most successful approaches to
the Gibbs conditioning principle is through the theory of large deviations \cite{kesten2012random,dembo1996refinements}. This approach involves {\em Sanov's theorem} \cite{sanov1958probability} that provides the large-deviation rate function of the empirical measure induced by a sequence of i.i.d. random variables
and the {\em contraction principle} \cite{donsker1983asymptotic} that describes how continuous mappings preserve the large deviation principle from one space to another space. In short, these theorems regarding counting and transformation in the theory of large deviations yield the Gibbs conditioning principle and provide the foundation of using the maximum entropy distribution under certain constraints to find the limit of a sequence of conditional probabilities.

\subsection{Large deviation theory}
Let $\{ X_n \}_{n \in \mathbb{N}} $ be a sequence of i.i.d. absolutely integrable (i.e. 	$\mathbb{E}|X_1|  < \infty$) real random variables with mean $\mu := \mathbb{E} [X_1]$, and let
\begin{align}
\overline{X}_n := \frac{1}{n} \sum_{i=1}^n X_i 
\end{align}
 By the {\em weak law of large numbers}, \begin{align}
\overline{X}_n  \xrightarrow{P} \mu \quad \text{when} \ n \rightarrow \infty.
\end{align}  
That is, for any  $\epsilon>0$,
\begin{align}
\lim_{n \rightarrow \infty} \mathbb{P}\left( |\overline{X}_n - \mu |  > \epsilon    \right) = 0. 
\end{align}

To study the question how fast this probability tends to zero, Harald Cram\'er  obtained the following theorem in 1938:

\begin{theorem}[Cram\'er's theorem \cite{cramer1976century}] 
\label{thm:cramer}
Assume that
\begin{align*}
A(\lambda) := \log \mathbb{E}[e^{\lambda X_1}] < \infty, \quad \lambda \in \mathbb{R}.
\end{align*}
Then 
\begin{align*}
(i) \lim_{n \rightarrow \infty} \frac{1}{n} \log \mathbb{P} \left[ \overline{X}_n \geq y  \right] &= -\phi (y) \quad \text{when} \ y > \mu, \\ \notag 
(ii) \lim_{n \rightarrow \infty} \frac{1}{n} \log \mathbb{P} \left[ \overline{X}_n \leq y  \right] &= -\phi (y) \quad \text{when} \ y < \mu,
\end{align*}
where $\phi$ is defined by 
\begin{align}\label{eq:LDPphi}
\phi(y) := \sup_{\lambda \in \mathbb{R}} \left[ y\lambda - A(\lambda)   \right] \quad \text{for} \ x \in \mathbb{R}.
\end{align}
\end{theorem}

The function $A$ is called the {\em logarithmic moment generating function}. In the applications of the large deviation theory to statistical mechanics, $A$ is also called the {\em free energy function} and the function $\phi$ is called the rate function of large deviations  \cite{touchette2009large}. We can recognize that $\phi(y)$ is the {\em Legendre transform} of $A(\lambda)$ ($A$ is a convex function). Therefore, $\phi = A^*$ (the convex conjugate of $A$) and it leads to the following pair of reciprocal equations 
\begin{align}
\frac{d A(\lambda)}{d \lambda} = y	\quad \text{if and only if} \quad  \frac{d \phi(y)}{d y} = \lambda.
\label{pre:legendre}
\end{align}

Now, we can apply this equivalence \eqref{pre:legendre} to Theorem \ref{thm:von.1981}: The parameter $\lambda$ of the maximum entropy distribution $e^{\lambda s}$ is  implicitly solved by \eqref{pre:max.contraint}, which gives rise to $\lambda$ determined  by
\begin{align}
\frac{\ d \log \int e^{\lambda s} f(s) ds  }{d \lambda} = \alpha.
\label{pre:legendre1}
\end{align}
By the definition of $A(\lambda)$ and  \eqref{pre:legendre} and \eqref{pre:legendre1}, we have
\begin{align}
\frac{d A(\lambda)}{d \lambda} = \alpha	\quad \text{if and only if} \quad  \frac{d \phi(\alpha)}{d \alpha} = \lambda.
\label{pre:legendre2}
\end{align}
Therefore, this result \eqref{pre:legendre2} shows that $\lambda$ not only can be determined implicitly by the free energy function $A$ but also can be founded explicitly by the rate function $\phi$. 

One of our main theorems (Theorem \ref{thm:limit.ldp})  can be applied to a particular type of heat bath as the sum of i.i.d. random variables (Theorem \ref{thm:iidsum}), then we directly show that $\lambda$ is uniquely determined by the first derivative of the rate function $\phi$ given on the condition $\alpha$. In this case, we apply the large deviation principle directly to the distribution of the heat bath
$$  Y_n = \frac{1}{n} \sum_{i=2}^{n} X_i $$
rather than use the large deviation principle for the empirical measure
$$ L_n = \frac{1}{n} \sum_{i=1}^n \delta_{X_i}.$$
In fact, the former (our approach) actually follows Gibbs' logic of the canonical distribution through the heat bath method; The later (Gibbs conditioning principle) follows Boltzmann's logic of the canonical distribution  through counting statistics. The reason to call the ``Gibbs" conditioning principle was in order to comprehend Gibbs' prediction of the canonical distribution from a mathematical standpoint \cite{kesten2012random}, however, in our opinion, it is closer to the idea of Boltzmann's derivation of the canonical distribution.

From our perspective, choosing the maximum entropy distribution to approximate the conditional probability is a natural consequence of the emergence of $  e^{\lambda x} f(x)$ when the finite subsystem  is contained in an infinitely large system with a value far from its mean. In other words, (normalized) $  e^{\lambda x} f(x)$ is the density of the limit of a sequence of conditional probabilities and it maximizes the relative entropy \eqref{rel.entropy} as an inevitable corollary from the setup of the heat bath method. In comparison with the Gibbs conditioning principle, our logic provides a very different point of view of why the maximum entropy principle works to find the limit of conditional probabilities.  Even though these two approaches have very different philosophies, in terms of mathematics, they are connected by the reciprocal equations \eqref{pre:legendre} through the Legendre transform.

\subsection{Asymptotic behavior of probabilities}
In order to define how ``good" of an approximation of conditional probability is, we first need to decide which metric we would use in the space of measures. In what follows, let $\Omega$ denote a measurable space with $\sigma$-algebra $\mathcal{F}$ and let $\mathbb{P}$, $\mathbb{Q}$ denote two probability measures on $(\Omega, \mathcal{F})$.

\begin{definition}[KL-divergence]  	For two probability distributions of a continuous random variable, $\mathbb P$ and $\mathbb Q$, the \textit{KL-divergence} is defined by
 	\begin{align}\label{KLdivergence}
 	D_{\text{KL}}(\mathbb P ~\|~ \mathbb Q):=\int_{-\infty}^{+\infty} p(x) \log \left(\frac{p(x)}{q(x)}\right)dx,	
 	\end{align}
where $p,q$ are the density functions of $\mathbb P, \mathbb Q$, respectively. For two probability distributions of a discrete random variable, ${ \mathbb P}$ and ${ \mathbb Q} $, the Kullback-Leibler divergence between them can be written as
\begin{align} 	\label{def:QKL}
D_{\text{KL}}\left({ \mathbb P} ~\|~ { \mathbb Q} \right)&= \sum_{k\in \Omega} P(k) \log \left( \frac{P(k)}{Q(k)} \right),
\end{align} 	
where $P, Q$ are the probability mass functions of ${ \mathbb P},  { \mathbb Q}$, respectively and $\Omega$ is a countable space. By continuity arguments, the convention is assumed that $0 \log \frac{0}{q} = 0 $ for $q \in \mathbb{R}$ and $p\log \frac{p}{0} = \infty$ for $p \in \mathbb{R} \backslash \{ 0 \}$. Therefore, the KL-divergence can take values from zero to infinity.
\end{definition}
 
\begin{definition}[total variation]
 	The total variation distance between two probability measures $\mathbb P,\mathbb Q$ on a sigma-algebra $\mathcal F$ is defined by 
 	$$\delta(\mathbb P,\mathbb Q):=\sup_{A\in\mathcal F}|\mathbb P(A)-\mathbb Q(A)|.
 	$$
\end{definition}

It's well known that we have the following relation between KL-divergence and total variation by Pinsker's inequality \cite{pinsker1964information}:
 \begin{align}\label{eq:totalvariation}
 \delta(\mathbb P,\mathbb Q)\leq \sqrt{\frac{1}{2}D_{\text{KL}}(\mathbb P~\|~ \mathbb Q)}.
 \end{align}
 
\begin{definition}[convergence of measures in total variation]
Given the above definition of total variation distance, let $\{ \mathbb P_n  \}_{n \in \mathbb{N}}$ be a sequence of measures on $(\Omega, \mathcal{F} )$. The sequence is said to converge to a measure $\mathbb P$ on $(\Omega, \mathcal{F} )$ in total variation distance if 
$$  \lim_{n \rightarrow \infty} \delta(\mathbb P_n,\mathbb P)  = 0$$
and it is equivalent to 
$$ \lim_{n \rightarrow \infty}  \sup_{ \| f \|_{\infty} \leq 1  } \biggr\lvert \int f d \mathbb P_n  -  \int f d\mathbb P  \biggr\rvert = 0.  $$
\end{definition}

 \begin{definition}[weak convergence of measures] 
Let $\{ \mathbb P_n  \}_{n \in \mathbb{N}}$ be a sequence of probability measures on $(\Omega, \mathcal{F} )$. We say that $\mathbb P_n $  converges {\em weakly} to a probability measure  $\mathbb P$ on  $(\Omega, \mathcal{F} )$ if 
$$ \lim_{n \rightarrow \infty}  \int f d \mathbb P_n  = \int f d \mathbb P, $$
for all $f \in C_b(\Omega)$. 
 \end{definition} 
From the two definitions above, total variation convergence of measures always implies weak convergence of measures. 

\begin{definition}[convergence in distribution]
A sequence $\{ X_n \}_{n \in \mathbb{N}} $ of random variables is said to {\em convergence in distribution} to the random variable $X$ if 
$$ \mu_{X_n} \rightarrow \mu_X \quad \text{weakly}, $$
in which $\mu_{X_n}$ is the law of $X_n$ and $\mu$ is the law of $X$.
\end{definition}

Even though the KL-divergence is not a metric, by the inequality \eqref{eq:totalvariation}, if the KL-divergence of one sequence of measures from another sequence of measures converges to zero, then the two sequences of measures have to converge to zero in total variation. So they must converge to zero weakly. Following this line of implication, in the present work, we start with defining the KL-divergence between two sequences of measures then understand what conditions guarantee it converges to zero. Once we have that, we will attain both strong convergence and weak convergence of the two sequences of measures to zero under those conditions.

Furthermore we mention two classical theorems (see the reference \cite{shorack2000probability}) regarding the convergence of probability distributions which we will use in our proofs.

\begin{theorem}[Berry-Esseen theorem]
\label{thm:berry-esseen} 
Let $X$ have mean zero, $\mathbb{E}[X^2] = \sigma^2$, and  $\mathbb{E} | X |^3  < \infty$. Let $Z_n = \left( X_1 + \cdots + X_n \right)/ \sqrt{n} \sigma ,$ where $X_1, \cdots, X_n$ are i.i.d. copies of $X$. Then we have
\begin{align}
\left| \mathbb P \left(Z_n < z \right) - \mathbb P \left( G < z  \right) \right|     =  O\left( \frac{\mathbb{E} | X |^3 }{  \sqrt{n}} \right)
\end{align}
for all $z \in \mathbb{R},$ where $G \sim  N(0, 1).$ 	
\end{theorem}

 \begin{theorem}[Slutsky's theorem]
 \label{thm:slutsky}
Let $\{ Z_n \}_{n \in \mathbb{N}}, \{ W_n \}_{n \in \mathbb{N}}$ be sequences of random variables. If $Z_n$ converges in distribution to a random variable $X$ and $W_n $ converges in probability to a constant $c$, then
\begin{align}
Z_n + W_n \rightarrow X + c \quad \text{in distribution}. 
\end{align}

 \end{theorem}
 
\begin{corollary}
\label{cor:berry-esseen-slusky}
Let $X$ have mean zero, $\mathbb{E}[X^2] = \sigma^2$, and  $\mathbb{E} | X |^3  < \infty$. For some finite $k \in \mathbb{N}$, let $W_n = \left( X_1 + \cdots + X_{k} \right)/ \sqrt{n} \sigma $ and  $ Z_n = \left( X_{k+1} + \cdots + X_{n + k} \right)/ \sqrt{n} \sigma$, where $X_1, \cdots, X_{n+k}$ are i.i.d. copies of $X$. Let $\tilde{Z}_n= Z_n + W_n$, then we have 
\begin{align}
\tilde{Z}_n \rightarrow G \quad \text{in distribution}, \quad G \sim N(0, 1).
\label{pre:cor:slusky}
\end{align}
Furthermore,
\begin{align}
\left| \mathbb P \left( \tilde{Z}_n < z \right) - \mathbb P \left( G < z  \right) \right| =  O\left( \frac{ \mathbb{E} | X |^3 }{  \sqrt{n}} \right)
\label{pre:cor:berry}
\end{align}
for all $z \in \mathbb{R}.$  	
\end{corollary} 
This corollary follows from Theorem \ref{thm:berry-esseen} and Theorem \ref{thm:slutsky}. The proof is provided in Appendix \ref{appendix:prof.cor.berry}.




\section{Main results}
\label{ch:main}

\subsection{Setup}
\label{ch:setup}
In Section \ref{ch:introduction} of introduction, we have already provided our philosophy of adopting the
{\em conditional probability} to derive the canonical ensemble. In this section of the main results, we are going to rigorously show: when a measurable function of the subsystem is ``small" relative to the whole system,  the ``canonical distribution" is a ``good" approximation of that conditional distribution. For the sake of simplicity, we will use the terms: ``subsystem", ``heat bath", and ``whole system" to represent a measurable function of those systems, respectively. 
Within this framework, we first need to define three things rigorously:
\begin{enumerate}
    \item A relatively small subsystem.
    \label{def.small.subsystem}
    \item Canonical probability distributions.
\label{def.can.dis.}
    \item Good approximations.
\label{def.good.approx.}
\end{enumerate}
For the definition of \eqref{def.small.subsystem}: a relatively small subsystem, we consider a sequence of conditional densities
\begin{align}
f_{X \mid \tilde{Z}_n} (x; E_n), \quad E_n :=  \mu_n +  I / \beta_n,
\label{Section:existence:three.rep}
\end{align}
where $ \tilde{Z}_n := X + \tilde{Y}_n$, $X$ is a nonnegative continuous random variable and  $\tilde{Y}_n$ is a sequence of continuous random variables, $I$ is a finite interval and $\mu_n$, $\beta_n$ are positive sequences. Note that we here use $\tilde{Y}_n, \tilde{Z}_n$ instead of $Y_n, Z_n$ because we will do transformations for $\tilde{Y}_n, \tilde{Z}_n$ later, so $Y_n, Z_n$ will be used to define transformed $\tilde{Y}_n, \tilde{Z}_n$. The formula of $E_n$ is to represent two kinds of transformations that we can do for the interval $I$: $\mu_n$ is the parameter of shifting and $\beta_n$ is the parameter of scaling. Through different combinations of $\mu_n$ and $\beta_n$, the given condition of $\tilde{Z}_n$ will be on certain significant scales. For two examples,
\begin{enumerate}
	\item Assume $\mu_n := \mathbb{E} [\tilde{Z}_n] = n\mu ,$ $\mu$ is a constant and $\beta_n = 1/\sqrt{n}$, then $\tilde{Z} _n $ is conditioned to be inside the interval $E_n = n \mu + \sqrt{n}I $. The interval $E_n$ is then around  $\mathbb{E}[\tilde{Z}_n]$ with a scale of the Gaussian fluctuations in central limit theorem.
	\item Assume $\beta_n = 1/ n$, then  $\tilde{Z}_n $ is conditioned to be inside the interval  $E_n = n \mu + n I$. The interval $E_n$ is then around $\mathbb{E}[\tilde{Z}_n]$ with a scale of the large deviations.
\end{enumerate}
In our theorems, we will assume that 
\begin{align}
	\mathbb{E} [X^j] <\infty, \ \text{for some finite} \ j,  \quad \text{ and } \quad  \beta_n = o(1).
	\label{finite.small.system}
\end{align}
Therefore, the definition \eqref{Section:existence:three.rep} of conditional densities is a sequence of densities for the nonnegative continuous random variable $X$ with $\mathbb{E}[X^j] < \infty$  conditioned on the event $\tilde{Z}_n \in E_n$ with $E_n \rightarrow \infty  \ (\beta_n \rightarrow 0)$. In this way, the positive sequence $\beta_n$ characterizes that the subsystem is relatively ``small" to the given condition of the whole system. 

Then we will extend our definition of a ``small" subsystem to the case when we have discrete random variables. Consider a sequence of conditional probability functions
\begin{align} \label{def:PMF}
P \big( K =  k \mid \tilde{H}_n \in  E_n \big), \quad E_n = \mu_n + I/\beta_n,
\end{align}
where $\tilde{H}_n := K+ \tilde{L}_n$,  $K$ is a nonnegative discrete random variables and we assume that 
\begin{align}
\mathbb{E}[K^j]  < \infty, \ \text{for some finite} \ j, \quad \text{ and } \quad  \beta_n = o(1),
\label{finite.small.system2}
\end{align}
and $\tilde{L}_n$ is a sequence of discrete random variables and $\tilde{H}_n := K + \tilde{L}_n$.

For the definition of \eqref{def.can.dis.}: canonical probability distributions, we are introducing a general form of the canonical probability distribution as follows:  Let $I$ be the interval in the setup \eqref{Section:existence:three.rep}. We consider a sequence of functions $\zeta_n: \mathcal{A} \times \mathbb{R} \rightarrow \mathbb{R} $, where $\mathcal{A}$ is the set of all finite intervals on $\mathbb{R}$. For the  canonical probability distribution of a nonnegative continuous random variable $X$, its density can be represented by
\begin{align}
\frac{f_X(x) e^{-\zeta_n( I ; x) x } }{\displaystyle \int_{\mathbb{R}^+} f_X(x) e^{-\zeta_n( I ; x) x } dx }
\quad \text{and} \quad  0 \leq \zeta_n( I ; x) < \infty, \ \text{for all} \ x \in \mathbb{R^+}.
\label{def:cont.can.dis}
\end{align} 
Consider a sequence of functions $\hat{\zeta}_n: \mathcal{A} \times \mathbb{R} \rightarrow \mathbb{R} $. For the canonical probability distribution of  a nonnegative discrete random variable $K$, it can be represented by
\begin{align}
\frac{P(K=k)e^{-\hat{\zeta}_n( I; k) k } }{\sum_{k \in S} P(K=k)e^{-\hat{\zeta}_n( I; k) k }} \quad \text{and} \quad  0 \leq \hat{\zeta}_n( I; k)  < \infty, \ \text{for all} \ k \in S,
\label{def:disc.can.dis}
\end{align}
where $S$ is a set of the support of $P(K=k)$.

For the definition of \eqref{def.good.approx.}: good approximations, ``good" is defined by a sufficiently small distance of two distributions in total variation \eqref{eq:totalvariation}. In most of our results, we prove that two sequences of distributions converge to zero in KL-divergence, by Pinsker's inequality, it implies those two sequences converge to zero in total variation, i.e., one sequence is a good approximation of the other one.

\subsection{Approximation of conditional probabilities}
\label{ch:existence}

Based on the definitions of \eqref{def.small.subsystem}, \eqref{def.can.dis.}, and \eqref{def.good.approx.} in the setup, we provide two approximation theorems to show the existence of the canonical distributions as good approximations of conditional distributions when the subsystem is sufficiently small relative to the whole systems.

Based on the setup \eqref{Section:existence:three.rep}, let $X_n := \beta_n X$ and take $j = 2$ for the assumption \eqref{finite.small.system}, i.e.,
\begin{align}
\mathbb{E}[X^2] < \infty, \quad \text{ and } \quad  \beta_n = o(1).
\label{finite.small.system3}
\end{align}
Let $a_n := \beta_n^2 \mathbb{E}[X^2] $, hence we have that
\begin{align}
	\mathbb{E}[X_n^2] = a_n, \ a_n=o(1).
\end{align}
Let $Y_n := \beta_n \left(  \tilde{Y}_n - \mu_n  \right)$ and $Z_n := X_n + Y_n$. Note that $Y_n, Z_n$ are the linear transformations of $\tilde{Y}_n, \tilde{Z}_n$, respectively; and recall that $\tilde{Z}_n = X +\tilde{Y}_n$ and the parameters of the transformation, $\beta_n, \mu_n,$ are from $E_n = \mu_n + I/\beta_n$ in the conditional density \eqref{Section:existence:three.rep}. Since we assume $I$ is a finite interval in \eqref{Section:existence:three.rep}, we can define it explicitly as $ I = [h, h + \delta], \ h,\delta \in \mathbb{R}$ and $\delta>0$. 

Based on the definitions given above, let ${\mathbb P}^{(n)}_I$ be a sequence of probability measures with density functions
\begin{align}\label{eq:PIQI}
\frac{f_{X}(x) e^{-\beta_n \psi_n (  I ;  \beta_n x )x}}{\displaystyle \int_{\mathbb{R}^+} f_{X}(x)  e^{-\beta_n \psi_n (I ; \beta_n x )x}dx}, \quad \psi_n(I  ; \beta_n x):= \frac{\partial \log P\big( Y_n  \in [y, y+\delta] \mid X_n = \beta_n x \big)}{\partial y}\biggr\rvert_{y= h }.
\end{align}
And let ${\mathbb Q}^{(n)}_I$ be a sequence of probability measures with density functions 
$f_{X \mid \tilde{Z}_n } \big(x  ;  E_n \big).$

Our first theorem for continuous random variables is as follows:

\begin{theorem}

\label{thm:theorem 1}
  Assume there exist positive constants $C_1, C_2$, a positive sequence $b_n = o(1)$, and an open interval $D$ such that the following holds:

\begin{enumerate}
	\item For all  $x \in \mathbb{R}^+$,  $y \in \mathbb{R}$,  
	\begin{align} \label{eq:logsquare}  \left|\frac{\partial^2  P\big( Y_n  \in [y, y + \delta ] \mid X_n = x \big)}{\partial y^2} \right|\leq C_1, \quad \left|\frac{\partial^2 \log P\big( Y_n  \in [y, y + \delta ] \mid X_n = x \big)}{\partial y^2}\right|\leq C_2. \end{align} 
	\item  For all  $x \in \mathbb{R}^+$ and every $ [y,  y+ \delta] \subset   D$, there exist positive constants $\delta_1, C_3$ depending on $y$ such that
	 \begin{align} \label{eq:fyx}
	 	&   P\big(Y_n \in [y,  y+ \delta] \mid X_n= x \big) \geq \delta_1 ,\quad 0 \leq \frac{\partial \log P\big(Y_n \in [y,  y+ \delta] \mid X_n= x \big)}{ \partial y} \leq C_3,   \\
	 	 \label{eq:fyx2}
	 	& \big| P\big(Y_n \in [y,  y+ \delta] \mid X_n= x \big)  -  P\big( Y_n \in [y,  y+ \delta] \big) \big| \leq b_n  P\big( Y_n \in [y,  y+ \delta]  \big) .
	 \end{align}
	\item  For every $ [z, z+\delta] \subset D$, there exists a positive constant $\delta_2$ depending on $z$ such that
	\begin{align} \label{eq:condforz}
	 P\big( Z_n \in [z, z+\delta] \big) \geq \delta_2.	
	\end{align}

\end{enumerate}
Given an interval $I   \subset D$, then 
\begin{align}
D_{\textnormal{KL}}\left({\mathbb P}^{(n)}_I \|~ {\mathbb Q}^{(n)}_I\right) =O(a_n+b_n),
\end{align} 
and ${\mathbb P}^{(n)}_I$  satisfies  the definition of the canonical probability distributions in \eqref{def:cont.can.dis}.
 
\end{theorem}

\begin{remark}
By Pinsker's inequality, Theorem \ref{thm:theorem 1} implies that
\begin{align*}
\delta \left({\mathbb P}^{(n)}_I, {\mathbb Q}^{(n)}_I\right) =O\left(\sqrt{ a_n+b_n}\right).
\end{align*}
\end{remark}

\begin{remark} \label{rmk:smallandweak}
Interpretations of Theorem \ref{thm:theorem 1} for statistical mechanics: the sequence  $a_n = o(1)$ represents that the second moment of the function of the subsystem $X$ scaled by the size of the given condition of the whole system asymptotically goes to zero. And the sequence $b_n = o(1)$ represents that $X_n$ and $Y_n$ are asymptotically independent. By our approximation theorem, 
using the canonical distribution to approximate the conditional distribution results in a very small error $O(\sqrt{ a_n+b_n})$ when $n $ is sufficiently large, i.e., 
\begin{enumerate}
\item \label{cond:stat.mech.small} The subsystem is small relative to the whole system.
\item \label{cond:stat.mech.indep} The subsystem has weak interaction with its surrounding. 
\end{enumerate}
Note that these conditions \eqref{cond:stat.mech.small} and \eqref{cond:stat.mech.indep} echo the physicist's setup of the canonical ensemble in statistical mechanics. 
\end{remark}

For Theorem \ref{thm:theorem 1}, we require the condition \ref{eq:fyx2} and the sequence $b_n$ in that condition is asymptotic to zero. As Remark \ref{rmk:smallandweak}, it means that the subsystem and the heat bath are asymptotically independent. In the following corollary, we are going to extend Theorem \ref{thm:theorem 1} to the case when the subsystem $X_n$ and its surrounding (the heat bath) $Y_n$ are not asymptotically independent.

Recall that  $ I = [h, h + \delta], \ h,\delta \in \mathbb{R}$ and $\delta>0$ and ${\mathbb Q}^{(n)}_I$ is a sequence of probability measures with density functions 
$f_{X \mid \tilde{Z}_n } \big(x  ;  E_n \big).$ Let $\hat{{\mathbb P}}^{(n)}_I$ be a sequence of probability measures with density functions
\begin{align*}
\frac{f_{X}(x) e^{-\beta_n \phi_n (  I )x}}{\displaystyle \int_{\mathbb{R}^+} f_{X}(x)  e^{-\beta_n \phi_n (I )x}dx},
\end{align*}
where
\begin{align} \label{cor:parm.exp}
   \phi_n(I):= \frac{\partial \log P\big( Y_n  \in [y, y+\delta] \mid X_n = 0 \big)}{\partial y}\biggr\rvert_{y= h } - \frac{\partial \log P\big( Y_n  \in [y, y+\delta] \mid X_n = 0 \big)}{\partial x}\biggr\rvert_{y= h }.  
\end{align}

\begin{corollary}

\label{cor:corollary 1}
  Assume there exist positive constants $C_1, C_2$, and an open interval $D$ such that the following holds:

\begin{enumerate}
	\item For all  $x \in \mathbb{R}^+$,  $y \in \mathbb{R}$,  
	\begin{align} \label{cor:eq:logsquare}  \left| \partial^{(2)}  P\big( Y_n  \in [y, y + \delta ] \mid X_n = x \big) \right|\leq C_1, \quad \left|\partial^{(2)} \log P\big( Y_n  \in [y, y + \delta ] \mid X_n = x \big)\right|\leq C_2, \end{align} 
	where $\partial^{(2)}$ denotes all the second order partial derivatives. 
	\item  For all  $x \in \mathbb{R}^+$ and every $ [y,  y+ \delta] \subset   D$, there exist positive constants $\delta_1, C_3$ depending on $y$ such that
	 \begin{align} \label{cor:eq:fyx}
	 	&    P\big(Y_n \in [y,  y+ \delta] \mid X_n= x \big) \geq  \delta_1 ,\quad 0 \leq \partial^{(1)}  \log P\big(Y_n \in [y,  y+ \delta] \mid X_n= x \big) \leq C_3,
	 \end{align}
	 where $\partial^{(1)}$ denotes all the first order partial derivatives. 
	\item  For every $ [z, z+\delta] \subset D$, there exists a positive constant $\delta_2$ depending on $z$ such that
	\begin{align} \label{cor:eq:condforz}
	 P\big( Z_n \in [z, z+\delta] \big) \geq \delta_2.	
	\end{align}
\end{enumerate}
Given an interval $I   \subset D$, then 
\begin{align}
D_{\textnormal{KL}}\left(\hat{{\mathbb P}}^{(n)}_I \|~ {\mathbb Q}^{(n)}_I\right) =O(a_n),
\end{align} 
and $ \hat{{\mathbb P}}^{(n)}_I$  satisfies  the definition of the canonical probability distributions in \eqref{def:cont.can.dis}.
 
\end{corollary}

The proof of Corollary \ref{cor:corollary 1} basically follows from the proof of Theorem \ref{thm:theorem 1}, and we provide the details of the proof in Appendix \ref{proof:cor:corollary 1}. 

\begin{remark}
Here we want to emphasize the difference between  Theorem \ref{thm:theorem 1} and
Corollary \ref{cor:corollary 1}: On the one hand, Corollary \ref{cor:corollary 1} requires a stronger condition that those partial derivatives in the conditions \eqref{cor:eq:logsquare} and \eqref{cor:eq:fyx} have to be bounded both in the $x$ and $y$ directions; however,
Theorem \ref{thm:theorem 1} only requires that the partial derivatives in the conditions \eqref{eq:logsquare} and \eqref{eq:fyx} are bounded in the $y$ direction. On the other hand, Corollary \ref{cor:corollary 1} does not require the condition \eqref{eq:fyx2} in Theorem \ref{thm:theorem 1}, which is to define the asymptotic independence between $X_n$ and $Y_n$. Based on the difference of those conditions, Theorem \ref{thm:theorem 1} and Corollary \ref{cor:corollary 1} give rise to distinct parameters of the exponential factors. The parameter of the exponential factor \eqref{cor:parm.exp} in Corollary \ref{cor:corollary 1} includes one additional term which involves the partial derivative with respect to $x$.

Note that the parameter of the exponential factor \eqref{cor:parm.exp} can be rewritten as
\begin{align} \label{cor:parm.exp2}
   \phi_n(I) = \frac{\partial \log P\big( Y_n  \in [y, y+\delta] \big)}{\partial y}\biggr\rvert_{y= h } + \left( \frac{ \partial \log C(x,y)}{\partial y} -  \frac{ \partial \log C(x,y)}{\partial x}
    \right) \biggr\rvert_{(x=0, y= h)},
\end{align}
where 
\begin{align} \label{cor:parm.exp3}
    C(x,y) = \frac{P\big( Y_n  \in [y, y+\delta] \mid X_n = x \big)}{P\big( Y_n  \in [y, y+\delta] \big)}.
\end{align}
Corollary \ref{cor:corollary 1} with the parameter represented by $\eqref{cor:parm.exp2}$ has a critical interpretation in statistical mechanics: For a system in contact with a heat bath, if the interaction are not weak (i.e. the correlation in mathematical terms does not approach zero), then the effect of this interaction will appear in the parameter of the exponential factor as the function of $C(x, y)$ in \eqref{cor:parm.exp2} for the canonical distribution. This result is different from the standard example in statistical mechanics: in the limit where the interaction goes to zero, the parameter only includes the effect of the fluctuations of heat bath (the first term on the right side of \eqref{cor:parm.exp2}) without any effect from the correlation  (the second term on the right side of \eqref{cor:parm.exp2}). 

\end{remark}

Now we extend our approximation theorem to discrete random variables based on the setup \eqref{def:PMF}. Recall that  $\tilde{H}_n = K+ \tilde{L}_n$ and $E_n = \mu_n + I/\beta_n$ defined in the conditional probability mass function  \eqref{def:PMF}. Take $j = 2$ for the assumption \eqref{finite.small.system2}, i.e.,
\begin{align}
\mathbb{E}[K^2] < \infty,
\label{finite.small.system4}
\end{align}
and by the definition \eqref{def:disc.can.dis}, we have a set $S$ such that 
\begin{align}\label{def:setS}
S:=\{k \in \mathbb R: P(K = k)>0 \}.
\end{align}

Let $K_n := \beta_n K$ be a sequence of nonnegative discrete random variables and let $a_n := \beta_n^2 \mathbb{E}[K^2] $. By \eqref{finite.small.system4} and \eqref{def:setS}, we have that
\begin{align}
\mathbb{E}[K_n^2] = a_n, \ a_n=o(1),
\end{align} 
and a sequence of sets $S_n$ such that
\begin{align}\label{eq:setS}
S_n :=\{  \beta_n k \in \mathbb R: P(K_n =  \beta_n k )>0 \}.
\end{align}

By shifting with $\mu_n$ and scaling with $\beta_n$, we can define a linear transformation of $\tilde{L}_n$, $L_n := \beta_n \left( \tilde{L}_n - \mu_n \right)$, and let $H_n := K_n + L_n$. Furthermore, let  $Y_n$ be a sequence of continuous random variables and $Z_n := K_n + Y_n$. Based on the given definitions, our second theorem for discrete random variables is as follows:

\begin{theorem}
\label{thm:d+d} 

Assume the following conditions hold:
\begin{enumerate}
    \item \label{cond:sum.dis}
    All conditions in Theorem \ref{thm:theorem 1} hold for $K_n, Y_n, Z_n$ on an open interval $D$.
    \item There exists a set $D' \subset D$ and a positive sequence $c_n = o(1)$ such that for every interval $ I' \subset D' $,
\begin{align}\label{eq:domain}
\sup_{ \beta_n k \in S_n}\biggr\lvert P\big( K_n = \beta_n k \mid H_n \in I'  \big) - P \big( K_n =   \beta_n k \mid Z_n \in  I' \big) \biggr\rvert= O(c_n).
\end{align}
\end{enumerate}
Given an interval $ I \subset D' $, then
\begin{align}\label{eq:cor55triangle}
\sup_{ k\in S}\biggr\lvert  P \big( K =  k \mid \tilde{H}_n \in  E_n \big)  -  B_n P( K = k) e^{-\beta_n\hat{\psi}_n (I; \beta_n k) k} \biggr\rvert = O\left(c_n + \sqrt{a_n + b_n} \right),
\end{align}   
where 
\begin{align*} 
\frac{1}{B_n} : = {\sum_{k\in S} P(K=k)  e^{-\beta_n \psi_n(I ; \beta_n k)k}}    \quad \text{and} \quad  \hat{\psi}_n(I ; \beta_n k):= \frac{\partial \log P\big( Y_n  \in [y, y + \delta ] \mid  K_n = \beta_n k \big)}{\partial y}\biggr\rvert_{y=h}, 
\end{align*}
and $b_n$ is defined in Condition \eqref{eq:fyx2} of Theorem \ref{thm:theorem 1}. Furthermore,
\begin{align}
B_n P( K = k) e^{-\beta_n\hat{\psi}_n (I; \beta_n k) k} 
\end{align} 
satisfies the definition of the canonical probability distribution in \eqref{def:disc.can.dis}.
\end{theorem}
Note that the given assumption \eqref{cond:sum.dis} in Theorem \ref{thm:d+d}: all conditions in Theorem \ref{thm:theorem 1} hold for $K_n, Y_n, Z_n$ on an open interval $D$, in which $K_n$ is corresponding to $X_n$ in Theorem \ref{thm:theorem 1}; and all conditions defined for ``all $x \in \mathbb{R}^+$" in Theorem \ref{thm:theorem 1} become defined for ``all $x \in S_n$" for Theorem \ref{thm:d+d}. In this way, even $K_n$ is a sequence of discrete random variables, all conditions in Theorem \ref{thm:theorem 1} are well-defined.  

\begin{remark}\label{rmk:approx.thm}
In Theorem \ref{thm:theorem 1} and Theorem \ref{thm:d+d}, $X$ and $K$ are defined as a nonnegative random variable. In the following two points, we extend our approximation theorem to the case when $X$ (or $K$) is bounded from below (shifting property) and the case when $X$ (or $K$) is a nonpositive  random variable (reflection property):

\begin{enumerate}
    \item \label{rmk:shift.prop.}
    (Shifting property)
Let $X$ be a continuous random variable bounded from below. By change of variables, let $\hat{X_n} := \beta_n (X - C)$, where $C$ is the finite lower bound, since $\beta_n = o(1)$, we still have
\begin{align}
\mathbb{E}[\hat{X}_n^2] = o(1).
\end{align}
In addition, assume that the conditional probability 
$$P\big( Y_n  \in [y, y + \delta ] \mid \hat{X}_n = x \big) $$
satisfies all of the conditions in Theorem \ref{thm:theorem 1}, then we can apply Theorem \ref{thm:theorem 1} to obtain the canonical distribution for $X$. We call this the {\em shifting property} of the canonical distributions. For the discrete random variable $K$, its canonical probability distribution has this property as well. This {\em shifting property} can be interpreted as the extension of the cases restricted to nonnegative quantities (e.g., energy and number of molecules) for the canonical ensemble and the grand canonical ensemble in statistical mechanics: the canonical distribution can be generalized to represent the possible values of a function which is {\em bounded from below} of the subsystem in thermal equilibrium with the heat bath at a {\em positive temperature }(In Theorem \ref{thm:theorem 1}, we choose the condition $I$ such that $ 0 \leq \psi_n( I; \beta_n x) < \infty$ ).
    \item \label{rmk:ref.prop.}
    (Reflection property) Let $X$ be a nonpositive continuous random variable.  Assume the condition \eqref{eq:fyx} in Theorem \ref{thm:theorem 1} becomes
\begin{align}
	& P\big(Y_n \in [y,  y+ \delta] \mid X_n= x \big) \geq \delta_1,\quad 
-C_3 < \frac{\partial \log P\big(Y_n \in [y,  y+ \delta] \mid X_n=  x \big)}{ \partial y} \leq 0,
\end{align}
for all $x \in  \mathbb{R^-}$. And assume all of the other conditions in Theorem \ref{thm:theorem 1} are satisfied, then Theorem \ref{thm:theorem 1} can be applied to an interval $I = [h, h+ \delta] \subset D$ such that $ -\infty < \psi_n( I; \beta_n x) \leq 0,$ for all $x \in  \mathbb{R^-}$. We call this {\em reflection property} of the canonical distributions. For the discrete random variable $K$, its canonical probability distribution has this property as well. Here is our interpretation of this {\em reflection property} for statistical mechanics: When a given condition $I$ of the whole system gives rise to a negative parameter ($ -\infty < \psi_n( I; \beta_n x) \leq 0$) in the exponential weight of the canonical distribution, our approximation theorem can be applied to the case of a nonpositive function of the subsystem. In combination with this property with the shifting property, the canonical distribution can represent the possible values of a function which is {\em bounded from above} of the subsystem in thermal equilibrium with the heat bath at a {\em negative temperature} (Here we choose the condition $I$ such that $ -\infty < \psi_n( I; \beta_n x) \leq 0$).
\end{enumerate} 

\end{remark}

\subsection{Limit theorems for conditional probabilities}

\label{ch:uniqueness}

In this section, we provide two limit theorems to show that a sequence of conditional distributions converges to a unique canonical distribution  by appropriate shifting and scaling, where the convergence is also in a corresponding scaling of the KL-divergence of this sequence of conditional distributions from its limit distribution. In contrast to the section \ref{ch:existence}, here we obtain a unique canonical distribution at the appropriate scale when a system is conditioned on an infinitely large total system ($n \rightarrow \infty$). It is different from the section \ref{ch:existence} in which we derive the canonical distribution for each finitely large $n$ directly.

Recall that from the section \ref{ch:existence}, for a sufficiently large $n$, we know that $\mathbb{Q}^{(n)}_I$ with density function $$f_{X \mid \tilde{Z}_n}(x; E_n)$$ can be well-approximated by $\mathbb{P}^{(n)}_I$ with density function
\begin{align}
\frac{ f_{X}(x) e^{-   \beta_n \psi_n ( I ; \beta_n x)x}}{\displaystyle \int_{\mathbb{R}^+} f_{X}(x)  e^{-\beta_n \psi_n (I ; \beta_n x)x}dx} \quad \text{and} \quad \psi_n(I  ; \beta_n x):= \frac{\partial \log P\big( Y_n  \in [y, y+\delta] \mid X_n = \beta_n x \big)}{\partial y}\biggr\rvert_{y= h }.
\label{consequence.of.existence}
\end{align}
Note that the parameter of the exponential function $\psi_n ( I ; \beta_n x)$ in \eqref{consequence.of.existence} depends on $n$ and $x$. 

Through our limit theorems in this section, we show that the sequence of measures $\mathbb{Q}^{(n)}_I$ can be well-approximated by a unique (sequence of) canonical distribution(s) with density function(s) 
\begin{align}
\frac{ f_X(x)e^{- \lambda_n(I) x}}{ \displaystyle\int_{\mathbb{R}^+}  f_X(x)e^{- \lambda_n(I) x} dx}
\label{uniq.can.dis}
\end{align}
in one of the cases: 
\begin{enumerate}
\item $\lambda_n(I)  = \beta_n \psi(I)$, where $\beta_n = o(1)$, and $\psi: \mathcal{A} \rightarrow \mathbb{R}$ is a function such that   $\mathcal{A}$ is the set of all finite intervals on $\mathbb{R}$, and $0 < \psi(I) < \infty$. \label{uniq.case1}
\item $\lambda_n(I) = \varphi(I)$, where $\varphi : \mathcal{A} \rightarrow \mathbb{R}$ is a function satisfying $0 < \varphi(I) < \infty$.  \label{uniq.case2}
\end{enumerate}
Note that $ \psi(I) $ and $\varphi(I)$ are independent of $x$ and $n$ in comparison with $\psi_n(I; x)$ in \eqref{consequence.of.existence}. One of the main ideas behind the proof of our limit theorems is as follows: Let $\tilde{\mathbb P}^{(n)}_I$ be a sequence of probability measures with density functions (normalized) 
$f_X(x) e^{-\beta_n\psi(I)x},
$ and let $\mathbb{P}_I$ be a probability measure with density function (normalized) 
$f_X(x) e^{-\varphi(I)x}.$
With $D_{\textnormal{KL}}$ defined as KL-divergence, Case \eqref{uniq.case1} can be considered as 
\begin{align}
D_{\textnormal{KL}}\left(\tilde{\mathbb P}^{(n)}_I~\|~ \mathbb{Q}^{(n)}_I \right) \rightarrow  0\quad \text{as} \ n\rightarrow \infty;
\end{align}
Case \eqref{uniq.case2}  can be considered as 
\begin{align}
D_{\textnormal{KL}}\left( \mathbb{P}_I ~\|~  \mathbb{Q}^{(n)}_I \right) \rightarrow  0\quad \text{as} \ n\rightarrow \infty.
\end{align}
Note that in Case \eqref{uniq.case1}, since $\beta_n =o(1)$, the sequence $\lambda_n(I) \rightarrow 0$ for any bounded $\psi(I)$. Therefore, we have to scale the distance $D_{\textnormal{KL}}$ by some function of $\beta_n$ to guarantee the uniqueness of $\psi(I)$. More details  are provided in Theorem \ref{thm:limit.smooth}.

Furthermore, we require stronger conditions than the conditions for \eqref{consequence.of.existence} in order to apply Lemma \ref{thm:unique1} and Lemma \ref{thm:unique2} to the proof of our limit theorems. Here is the essence of those two lemmas: under appropriate regularity conditions, the sequence $\lambda_n(I) $ in \eqref{uniq.can.dis} is uniquely determined by a linear approximation of the following sequence	 
\begin{align}
\log\left( \frac{f_{X \mid \tilde{Z}_n}(x ; E_n) }{ f_X(x)} \right).
\label{uni:lin.approx}
\end{align} 
Therefore, most of the conditions in our limit theorems are required to guarantee that \eqref{uni:lin.approx} is well-approximated by a linear function and the remainder term converges to zero fast enough.

Recall that $X_n := \beta_n X$, $Y_n := \beta_n \left( \tilde{Y}_n - \mu_n \right)$, and $Z_n := X_n + Y_n$, where $\beta_n, \mu_n$ are positive sequences and $\beta_n = o(1)$, and $E_n = \mu_n + I/\beta_n$, $ I = [h, h + \delta], \ h,\delta \in \mathbb{R}$ and $\delta>0$. 

Our first limit theorem for Case \eqref{uniq.case1}: $\lambda_n(I)  = \beta_n \psi(I)$ is as follows 
\begin{theorem}
\label{thm:limit.smooth}
Consider a function $\psi: \mathcal{B}\left( \mathbb{R}\right) \rightarrow \mathbb{R}$ such that $ 0 < \psi(I) < \infty$ for the given interval $I$.   Let $\tilde{\mathbb P}^{(n)}_I$ be a sequence of probability measures with density functions
\begin{align}
\frac{f_X(x) e^{-\beta_n\psi(I)x}}{\displaystyle \int_{\mathbb{R}^+}  f_X(x) e^{-\beta_n \psi(I)x } dx }.
\end{align}
 Assume the following conditions hold:
\begin{enumerate}
\item  \label{cond:limit.smooth1} $X$ is a nonconstant random variable with $\mathbb{E}[X^3] < \infty$ and $\displaystyle \frac{f_{X \mid \tilde{Z}_n}(x ; E_n) }{ f_X(x)} $ is uniformly bounded on $\mathbb{R}^+$.
\item \label{cond:limit.smooth2} $Y_n \rightarrow Y$ in distribution. The distribution function of $\ Y$ is bounded on $\mathbb{R}^+$ and satisfies  
\begin{align}
\log P \left( Y \in \left[ y,  y+\delta \right] \right) \in C^2(D) \quad \text{and} \quad 0 < \frac{\partial \log P(Y \in [y, y+\delta])}{\partial y}\biggr\rvert_h < \infty,
\end{align}
where $D$ is an open interval containing $h$.
\item  \label{cond:limit.smooth3} There exists a sequence of functions $g_n : \mathbb{R} \rightarrow \mathbb{R}$ with
$ \left| g_n(x)e^{-\beta_n \xi x} \right|$ uniformly bounded on $\mathbb{R}^+$ for any $\xi>0$ and $\mathbb{E}\left[ g_n(X)^2 \right] \rightarrow 0$
such that on $I_n = [0, d_n]$ with $d_n = O\left(\frac{1}{\beta_n}\right)$,
\begin{align} \label{cond:lin.approx.smooth}
\log \left( \frac{P\left( Y_n \in I - \beta_n x  \mid X_n = \beta_n x  \right)}{P \left( Z_n \in I \right)} \right) = \log \left(    \frac{ P\left( Y \in I - \beta_n x  \right)}{P\left( Y\in I\right)} \right) +  \beta_n g_n(x).
\end{align}

\end{enumerate}
Then
$$ \displaystyle \lim_{n \rightarrow \infty}\frac{{D_{\textnormal{KL}}\left( \tilde{\mathbb P}^{(n)}_I ~ \|~ {\mathbb Q}^{(n)}_I \right)}}{\beta_n^2}  
 = 0  \quad \text{if and only if} \quad \psi(I) =  \frac{\partial \log P(Y \in [y, y+\delta])}{\partial y}\biggr\rvert_h. $$
And $ \tilde{\mathbb P}^{(n)}_I $ satisfies the definition of the canonical probability distributions in \eqref{def:cont.can.dis}.

\end{theorem}

Our second limit theorem for Case \eqref{uniq.case2}: $\lambda_n(I)  =  \varphi(I)$ is as follows

\begin{theorem}
\label{thm:limit.ldp}
Let  $\varphi: \mathcal{B}\left( \mathbb{R}\right) \rightarrow \mathbb{R}$ be a function such that $ 0 < \varphi(I) < \infty$ for the given interval $I$.  Let $\mathbb{P}_I$ be a probability measure with density function
\begin{align*}
\frac{f_X(x) e^{-\varphi(I)x}}{\displaystyle \int_{\mathbb{R}^+}  f_X(x) e^{-\varphi(I)x } dx }.
\end{align*}
Assume the following conditions hold:
\begin{enumerate}
\item \label{cond:limit.ldp1} $X$ is a nonconstant random variable with $\mathbb{E}[X] < \infty$ and $\displaystyle \frac{f_{X \mid \tilde{Z}_n}(x ; E_n) }{ f_X(x)} $ is uniformly bounded on $\mathbb{R}^+$.
\item \label{cond:limit.ldp2} $Y_n \rightarrow \mu$ in probability, for some constant $\mu \notin I$. The sequence of laws of $Y_n$  satisfies a large deviation principle with speed $1/\beta_n$ and rate function $\phi \in C^2(D),$ where $D$ is an open interval containing $I$, and $ -\infty < \phi'(y) < 0  \text{ for all} \ y \in I$.
\item \label{cond:limit.ldp3} There exists a sequence of functions $r_n : \mathbb{R} \rightarrow \mathbb{R}$ with $  \left| r_n(x)e^{- \xi x} \right|$  uniformly bounded on  $\mathbb{R}^+$   for any   $\xi>0$ and $\mathbb{E}\left[ r_n(X)^2 \right]  \rightarrow 0$
such that on $I_n = [0, d_n]$ with $ d_n = O\left(\frac{1}{\beta_n}\right)$,
\begin{align} \label{cond:lin.approx.ldp}
\log \left( \frac{P\left( Y_n \in I - \beta_n x  \mid X_n = \beta_n x  \right)}{P \left( Z_n \in I \right)} \right) &= \log \left(    \frac{  \exp{ \left[ - \frac{1}{\beta_n}\phi \left(y^* - \beta_n x \right) \right] } }{  \exp{ \left[-  \frac{1}{\beta_n}\phi\left(y^* \right)\right] } } \right) + r_n(x),  \\
y^* &= \{y: \inf_{y\in I} \phi(y) \}.
\end{align}
\end{enumerate}
Then
$$ \displaystyle \lim_{n \rightarrow \infty}{D_{\textnormal{KL}}\left( {\mathbb P}_I ~ \|~ {\mathbb Q}^{(n)}_I \right)}  
 = 0  \quad \text{if and only if} \quad \varphi(I) = -\phi'(y^*).  $$
And ${\mathbb P}_I $ satisfies  the definition of the canonical probability distributions in \eqref{def:cont.can.dis}.

\end{theorem}

The following is the discussion about the circumstances when the condition \eqref{cond:lin.approx.smooth} (or the condition \eqref{cond:lin.approx.ldp}) for Theorem \ref{thm:limit.smooth} (or Theorem \ref{thm:limit.ldp}) could hold. Here we only discuss the condition \eqref{cond:lin.approx.smooth} but it is applied to the condition \eqref{cond:lin.approx.ldp} as well. 
We can consider three circumstances
\begin{enumerate}
    \item \label{dis:condition1} When $Y_n \rightarrow Y$ in distribution.  
    \item \label{dis:condition2} When $ X_n \rightarrow 0 $ in probability. 
    \item \label{dis:condition3} When $X_n$ and $Y_n$ are asymptotically independent.  
\end{enumerate}
Even thought the condition  \eqref{cond:lin.approx.smooth} and the combination of circumstances \eqref{dis:condition1} - \eqref{dis:condition3} are not the exact same, they are very close; therefore, these three circumstances provide us an insight regarding three elements of the condition \eqref{cond:lin.approx.smooth}: Circumstance \eqref{dis:condition1} means the heat bath has a limiting distribution; Circumstance \eqref{dis:condition2} means the subsystem is relatively small in comparison with the whole system; Circumstance  \eqref{dis:condition3} means those two systems have weak interaction. 

As Corollary \ref{cor:corollary 1} for the approximation theorem \ref{thm:theorem 1}, we are going to extend our limit theorems to the case when $X_n$ and $Y_n$ are not asymptotically independent.

\begin{corollary}
\label{cor:limit.smooth}
Let $G: \mathcal{A} \times \mathbb{R} \rightarrow \mathbb{R}$ be a function such that $ \mathcal{A}$ is the set of all finite intervals. For the given interval $I$, $G( I; 0)=1$ and $\log G(I, \xi) \in C^2(\mathbb{R}^+)$ with respect to $\xi$.
Assume the condition \eqref{cond:lin.approx.smooth} in Theorem \ref{thm:limit.smooth} becomes
\begin{align} \label{cor:cond:lin.approx.smooth}
\log \left( \frac{P\left( Y_n \in I - \beta_n x  \mid X_n = \beta_n x  \right)}{P \left( Z_n \in I \right)} \right) = \log \left(    \frac{ P\left( Y \in I - \beta_n x  \right) \cdot G(I; \beta_n x )}{P\left( Y\in I\right)} \right)  +  \beta_n g_n(x),
\end{align}
and the other conditions in Theorem \ref{thm:limit.smooth} hold. 
Then
$$ \displaystyle \lim_{n \rightarrow \infty}\frac{{D_{\textnormal{KL}}\left( \tilde{\mathbb P}^{(n)}_I ~ \|~ {\mathbb Q}^{(n)}_I \right)}}{\beta_n^2}  
 = 0  \quad \text{if and only if} \quad \psi(I) =  \frac{\partial \log P(Y \in [y, y+\delta])}{\partial y}\biggr\rvert_h - \frac{\partial \log G( I ; \xi)}{\partial \xi}\biggr\rvert_0. $$
And $ \tilde{\mathbb P}^{(n)}_I $ satisfies the definition of the canonical probability distributions in \eqref{def:cont.can.dis}.
\end{corollary}
\begin{remark}
The function $G$ in \eqref{cor:cond:lin.approx.smooth} can be considered as an approximation: 
\begin{align*}
    \frac{P\left( Y_n \in I - \xi  \mid X_n = \xi  \right)}{P\left( Y_n \in I - \xi  \right)} \approx G( I; \xi),
\end{align*}
in which the left side is equivalent to the joint probability of $X_n, Y_n$ divided by the products of their marginal probabilities. Therefore, $G$ could represent an estimation of the correlation of $X_n$ and $Y_n$; in information theory, the function $G$ is closely related to the mutual information between $X_n$ and $Y_n$.
\end{remark}

\begin{corollary}
\label{cor:limit.ldp}
Let $R: \mathcal{A} \times \mathbb{R} \rightarrow \mathbb{R}$ be a function, for the given interval $I$, $R( I; 0)=1$ and $\log R(I, \xi) \in C^2(\mathbb{R}^+)$ with respect to $\xi$.
Assume the condition \eqref{cond:lin.approx.ldp} in Theorem \ref{thm:limit.ldp} becomes
\begin{align} \label{cor:cond:lin.approx.ldp}
\log \left( \frac{P\left( Y_n \in I - \beta_n x  \mid X_n = \beta_n x  \right)}{P \left( Z_n \in I \right)} \right) &= \log \left(    \frac{  \exp{ \left[ - \frac{1}{\beta_n}\phi \left(y^* - \beta_n x \right) \right] } \cdot \left( R( I; \beta_n x)\right)^{\frac{1}{\beta_n}}}{  \exp{ \left[-  \frac{1}{\beta_n}\phi\left(y^* \right)\right] } } \right) +  r_n(x),  \\
y^* &= \{y: \inf_{y\in I} \phi(y) \},
\end{align}
and the other conditions in Theorem \ref{thm:limit.ldp} hold.
Then
$$ \displaystyle \lim_{n \rightarrow \infty}{D_{\textnormal{KL}}\left( {\mathbb P}_I ~ \|~ {\mathbb Q}^{(n)}_I \right)}  
 = 0  \quad \text{if and only if} \quad \varphi(I) = -\phi'(y^*)  - \frac{\partial \log R(I; \xi)}{\partial \xi}\biggr\rvert_0.  $$
And ${\mathbb P}_I $ satisfies  the definition of the canonical probability distributions in \eqref{def:cont.can.dis}.
\end{corollary}

\begin{remark}
The function $R$ in \eqref{cor:cond:lin.approx.ldp} can be considered as an approximation:
\begin{align*}
    \frac{P\left( Y_n \in I - \xi  \mid X_n = \xi  \right)}{P\left( Y_n \in I - \xi  \right)} \approx \left( R( I, \xi)\right)^{\frac{1}{\beta_n}}.
\end{align*}
In comparison with Corollary \ref{cor:limit.smooth}, when the sequence of laws of $Y_n$  satisfies a large deviation principle with speed $1/\beta_n$, the correlation of the subsystem and its heat bath has to be in $O( R ^{\frac{1}{\beta_n}})$ to contribute an additional term in the parameter of the exponential weight. Otherwise, if the correlation is just in $O(R)$ as the order in Corollary \ref{cor:limit.smooth}, then it has no influence on the canonical distribution. 
\end{remark}

The proof of Corollary \ref{cor:limit.smooth} (Corollary \ref{cor:limit.ldp}) basically follows from the proof of Theorem \ref{thm:limit.smooth} (Theorem \ref{thm:limit.ldp}).  We provide the details of proof in Appendix \ref{proof:cor:corollary 1}.

As our approximation theorems in Section \ref{ch:existence}, we can extend our limit theorems to discrete random variables, random variables bounded below, and random variables bounded above as follows:

\begin{enumerate}
\item Discrete random variables:
Theorem \ref{thm:limit.smooth} and Theorem \ref{thm:limit.ldp} can also be applied to the case when we have a nonnegative discrete random variable $K$,  a sequence of discrete random variables $\tilde{L}_n$, and $\tilde{H}_n := K+\tilde{L}_n $. It is said that the sequence of conditional probabilities $P(K=k \mid \tilde{H}_n \in E_n)$ has a limit ( by appropriate scaling) 
\begin{align}
\frac{P(K = k) e^{ -  \lambda_n(I)k} }{\sum_{k \in S} P(K = k) e^{-  \lambda_n(I)k}}.
\label{limit:dis.rvs}
\end{align}
The case of  $\lambda_n(I) = \beta_n \psi(I)$ follows from Theorem \ref{thm:limit.smooth}; The case of  $\lambda_n(I) = \varphi(I)$  follows from Theorem \ref{thm:limit.ldp}.  Furthermore, the probability function \eqref{limit:dis.rvs}  satisfies  the definition of the canonical probability distribution in \eqref{def:disc.can.dis}.

\item Random variables bounded below: As Remark \ref{rmk:approx.thm}, we can extend those limit theorems to the case when $X$ is bounded below. By change of variable, let $\hat{X_n} := \beta_n (X - C)$, where $C$ is the finite lower bound, we still have
\begin{align}
\mathbb{E}[ (X - C)^j ] < \infty, \ j=1 \ \text{or} \ 3.
\end{align}
Note that $j=3$ is for Theorem \ref{thm:limit.smooth} and $j=1$ is for Theorem \ref{thm:limit.ldp}. In addition, assume
$$ \log \left( \frac{P\left( Y_n \in I - \beta_n x  \mid X_n = \beta_n x  \right)}{P \left( Z_n \in I \right)} \right) $$
satisfies the condition of linear approximation in \eqref{cond:lin.approx.smooth} and \eqref{cond:lin.approx.ldp}, for Theorem \ref{thm:limit.smooth} and  Theorem \ref{thm:limit.ldp}, respectively. Then we can apply those limit theorems to obtain a unique canonical distribution of $X$.  Therefore, as the point \eqref{rmk:shift.prop.} in Remark \ref{rmk:approx.thm}, a unique canonical distribution derived by the limit of a sequence of conditional distributions has the ``shifting property". For the discrete random variable $K$, its unique canonical distribution has this property as well. 

\item Random variables bounded above:  Let $X$ be a nonpositive continuous random variable and the corresponding canonical distribution be a sequence of distributions with density functions
\begin{align}
 \frac{f_{X}(x) e^{-\lambda_n(I) x}}{\displaystyle\int_{\mathbb{R}^- } f_{X}(x) e^{-\lambda_n(I) x} dx }, \quad -\infty < \lambda_n(I) < 0.
\end{align}
When $\lambda_n(I) = \beta_n \psi(I) $, Theorem \ref{thm:limit.smooth} can be applied to an interval $I$ such that  $  -\infty <\psi(I) < 0$; When $\lambda_n(I) =  \varphi(I) $, Theorem \ref{thm:limit.ldp} can be applied to an interval $I$ such that  $   -\infty < \varphi(I) < 0$.  Therefore, as the point \eqref{rmk:ref.prop.} in Remark \ref{rmk:approx.thm}, a unique canonical distribution derived by the limit of a sequence of  conditional distributions has the “reflection  property”. For the discrete random variable $K$, its unique canonical distribution has this property as well. This reflection property provides us an explanation of the possibility of {\em negative temperature}: For some given condition of the whole system which arises a {\em negative} parameter $( -\infty < \lambda_n(I) < 0)$ in the exponential weight, a unique canonical distribution for a function {\em bounded from above} of the subsystem emerges as the limit of a sequence of conditional distributions.
\end{enumerate}




\section{Proofs of main results}\label{ch:proof}

\subsection{Proofs of Theorem \ref{thm:theorem 1} and Theorem \ref{thm:d+d}}

\subsubsection{Proof of Theorem \ref{thm:theorem 1}}

\begin{proof}
\label{proof:theorem 1}
We first prove for the case: $\{x: f_{X_n}(x) > 0 \} = \mathbb{R^+}$. 
In this case, $P(Z_n \in I \mid X_n =x)$ is well-defined for all $x \in \mathbb{R}^+.$ 
Let 
\[I = [h, h + \delta] \subseteq D, \quad  I - x := \left\{y- x: \ y \in I  \right\}, \]
with Condition \eqref{eq:condforz}: for $ I \subseteq D$, $ P\big( Z_n \in [h, h+\delta] \big) \geq \delta_2,$ we can derive the following conditional density by Bayes' theorem
\begin{align} \label{condprob}
f_{ X_n | Z_n}(x; I) &= 	\frac{  f_{X_n} (x) P(Z_n \in I \mid X_n =x)}{   P(Z_n \in I) } = \frac{  f_{X_n} (x)  P(Y_n \in I - x \mid X_n = x)}{  P(Z_n \in I) } 
, \quad \text{for} \ x \in \mathbb{R}^+.
\end{align} 
Note that $P(Y_n \in I - x \mid X_n = x) =P(Y_n \in [h - x , h + \delta - x]   \mid X_n = x)  $. Define \[G_\delta(y, x) :=  P\big( Y_n \in [y, y+\delta]  \mid X_n = x \big).\] 
By Taylor expansion and Condition \eqref{eq:logsquare}, we can expand $ G_\delta(h - x, x)$ at $(h, x)$  to get \[G_\delta(h - x, x)  = G_\delta(h , x) - \frac{\partial G_\delta(h , x) }{\partial y} x  + \frac{\partial^2 G_\delta(h - \alpha_n x , x) }{2 \partial y^2} x^2,  \quad \text{for some } \alpha_n \in (0,1).\] 
It implies that
\begin{align} \label{eq:taylorexpansion}
&P \big( Y_n \in [h -x, h +\delta-x] \mid X_n = x \big) \notag \\
= & P \big( Y_n\in [h , h +\delta] \mid X_n= x \big) -\frac{\partial P \big( Y_n \in [y, y+\delta] \mid X_n = x \big)}{\partial y}\biggr\rvert_{y=h} \cdot x  + r_n(x)x^2 \notag \\
= & P \big( Y_n \in [h, h+\delta] \mid X_n= x \big) \left[1 - \psi_n(I; x) \cdot x +   \frac{r_n(x)x^2}{ P \big( Y_n \in [h, h+\delta] \mid X_n =  x \big)}  \right] \notag \\
= & P \big( Y_n \in [h, h+\delta] \mid X_n= x \big) \left[ e^{-\psi_n(I ; x) x } -  \frac{(\psi_n(I ; x)x)^2e^{- \gamma_n \cdot \psi_n(I; x) x }}{2}  +   \frac{r_n(x)x^2}{ P \big( Y_n \in [h, h+\delta] \mid X_n =  x \big)}  \right] \notag \\
=& P \big( Y_n \in [h, h+\delta] \mid X_n= x \big) \left[ e^{-\psi_n(I ; x) x } + k_n(x)x^2 \right],
\end{align}
where 
\begin{align} \label{ansequence}
	\psi_n(I; x) &=  \frac{\partial \log  P \big( Y_n \in [y, y+\delta] \mid  X_n = x \big)}{\partial y}\biggr\rvert_{y=h}, \\
	 \label{ansequence2}
r_n(x) &= \frac{1}{2}  \frac{\partial^2  P \big( Y_n \in [y, y+\delta] \mid X_n= x \big)}{\partial y^2} \biggr\rvert_{y=h -\alpha_n x},  \\
 \label{ansequence3}
 k_n(x) &=  \frac{r_n(x)}{ P \big( Y_n \in [h, h+\delta] \mid X_n =  x \big)}   -  \frac{\psi_n(I ; x)^2e^{- \gamma_n \cdot \psi_n(I; x) x }}{2},
\end{align}
and we have applied Taylor's expansion 
$$ e^{y_n} = 1 + y_n + \frac{(y_n)^2 e^{\gamma_n y_n}}{2}, \quad \text{for some} \ \gamma_n \in (0,1) \quad \text{and} \quad y_n:=\psi_n(I; x) x$$
to the third equation in \eqref{eq:taylorexpansion}.
Note that by Condition \eqref{eq:fyx}, 
\begin{align}\label{eq:nonnegativephi}
 0\leq \psi_n(I ; x)\leq C_3,    
\end{align} and by Conditions \eqref{eq:logsquare} and \eqref{eq:fyx}, for all $ x \in \mathbb{R}^+$, $k_n(x)$ is uniformly bounded. Therefore, by the results of \eqref{condprob} and \eqref{eq:taylorexpansion}, for all $x\in \mathbb{R}^+$, we obtain that
\begin{align} \label{condprob:approx}
f_{X_n \mid Z_n } \big( x ; I \big) &= \frac{f_{X_n}(x)  P(Y_n \in  I \mid X_n = x) (e^{-\psi_n(I ; x) x } + k_n(x)x^2 ) }{P(Z_n \in I) }. \end{align}
In the following, we will use  brief notations 
\[P_{Y_n \mid X_n} \big( I ; x \big) :=  P(Y_n \in I \mid X_n = x),\quad P_{Z_n} \big(I \big ) := P(Z_n \in I).\] First, we let 
\begin{align} \label{ABh2}
 A_n:=\frac{1}{ \displaystyle \int_{ \mathbb{R}^+} f_{X_n}( x) e^{-\psi_n(I ; x) x} dx}.
\end{align}
Since $  \int_{\mathbb{R}^+} f_{X_n}(x) dx = 1$, from \eqref{eq:nonnegativephi}, we have $\displaystyle  \int_{\mathbb{R}^+} f_{X_n}(x) e^{-\psi_n(I ; x) x}dx \leq 1$, hence $A_n\geq 1$ for all $n\geq 1$.
By definition $X_n = \beta_n X$, $\beta_n \rightarrow 0$, we also have 
$$ \lim_{n\to\infty} \frac{1}{A_n}=\lim\limits_{n \rightarrow \infty} \int_{\mathbb{R}^+} f_{X_n}(x) e^{-\psi_n(I ; x) x}dx = \lim\limits_{n \rightarrow \infty} \int_{\mathbb{R}^+} f_{X}(x) e^{-\psi_n(I ; \beta_n x) \beta_n x}dx = 1 $$
by the dominated convergence theorem, so $A_n$ is uniformly bounded from above and from below.

Recall the definition of KL-divergence from \eqref{KLdivergence}, and the definitions of $\mathbb P^{(n)}_I$ and $\mathbb Q^{(n)}_I$ from \eqref{eq:PIQI}, we have 
\begin{align}
D_{\text{KL}}\left(\mathbb P^{(n)}_I ~\|~ \mathbb Q^{(n)}_I \right) &= \int_{\mathbb{R}^+} A_n f_{X}(x) e^{-\psi_n(I ; \beta_n x) \beta_n x} \log \left( \frac{ A_n f_{X}(x) e^{-\psi_n(I ; \beta_n x) \beta_n x} }{f_{X \mid \tilde{Z}_n}(x, E_n )}\right) dx \notag\\
&= \int_{\mathbb{R}^+} A_n f_{X_n}(x) e^{-\psi_n(I ; x) x} \log \left( \frac{ A_n f_{X_n}(x) e^{-\psi_n(I ; x) x} }{f_{X_n|Z_n}(x, I )}\right) dx \label{eq:KL1}\\
&= \left| \int_{\mathbb{R}^+} A_n f_{X_n}(x) e^{-\psi_n(I ; x) x} \log \left( \frac{f_{X_n|Z_n}(x, I )}{ A_n f_{X_n}(x) e^{-\psi_n(I ; x) x}}\right) dx\right|.\label{KL1}
\end{align}
\eqref{eq:KL1} is obtained by the change of variables $X_n = \beta_n X$ and the scale invariant property of the KL-divergence.   \eqref{KL1} is true because the KL-divergence is nonnegative. With \eqref{condprob}, the right hand side in \eqref{KL1} can be written as 
\begin{align} 
 & \left| \int_{\mathbb{R}^+}  A_n f_{X_n}(x)e^{-\psi_n(I ; x)}  \log \left( \frac{f_{X_n}(x) P_{Y_n | X_n} \big( I - x; x \big)  }{P_{Z_n}\big(I \big) }\cdot \frac{1}{ A_n f_{X_n}(x) e^{-\psi_n(I ; x) x}}\right)dx \right| \notag\\
= & \left| \int_{\mathbb{R}^+} A_n f_{X_n}(x)e^{-\psi_n(I ; x)} \log \left( \frac{P_{Y_n | X_n} \big(I - x ; x \big)  }{ P_{Y_n | X_n} \big( I ; x \big)   e^{-\psi_n(I ; x) x}} \cdot  \frac{ P_{Y_n | X_n} \big( I ; x \big)  }{ P_{Z_n}\big(I \big) A_n}\right) dx \right|  \notag\\
\leq &  \left| \int_{\mathbb{R}^+} A_n f_{X_n}(x)e^{-\psi_n(I ; x)} \log \left(  \frac{P_{Y_n | X_n} \big( I ; x \big) }{P_{Z_n}\left(I \right) A_n }\right) dx \right|   \notag\\
 &+\left| \int_{\mathbb{R}^+}A_n f_{X_n}(x)e^{-\psi_n(I ; x)} \log \left( \frac{ P_{Y_n | X_n} \big( I -x ; x \big)}{  P_{Y_n | X_n} \big( I ; x \big) e^{-\psi_n(I ; x) x}  }\right) dx \right|. \label{small2}
\end{align}
From the expression of $f_{X_n|Z_n}\big(x; I  \big)$ in  (\ref{condprob:approx}), we have the following identity
\begin{align}
\label{sum2.2}
 1&= \int_{\mathbb{R}^+} f_{X_n|Z_n} \big(x; I \big) dx \notag\\
  & = \frac{\displaystyle \int_{\mathbb{R}^+} f_{X_n}( x)e^{-\psi_n(I ; x)x } P_{Y_n | X_n} \big( I  ; x \big) dx  }{ P_{Z_n}\big(I  \big) } +  \frac{ \displaystyle \int_{\mathbb{R}^+} f_{X_n}( x) k_n(x)x^2 P_{Y_n | X_n} \big( I ; x \big) dx }{ P_{Z_n}\big(I \big)}.
\end{align}
For the second term in \eqref{sum2.2},  Conditions \eqref{eq:logsquare} and \eqref{eq:fyx} imply that $P_{Y_n | X_n} \big(I  ; x \big)$ and $k_n(x)$ are uniformly bounded and Condition \eqref{eq:condforz} implies that $P_{Z_n}\big(I \big)$ is uniformly bounded from below. Then by the assumption $\mathbb{E}[X_n^2] = a_n$, the first term in (\ref{sum2.2}) satisfies 
\begin{align}
\left|  \frac{\displaystyle \int_{\mathbb{R}^+} f_{X_n}( x)e^{-\psi_n(I ; x)x } P_{Y_n | X_n} \big( I ; x \big)   dx }{ P_{Z_n}\big(I \big)}  \right|=\left|  \frac{\displaystyle \int_{\mathbb{R}^+} f_{X_n}( x) k_n(x)x^2 P_{Y_n | X_n} \big(I  ; x \big) dx}{P_{Z_n}\big(I \big)}-1\right|=1+O(a_n).
\label{bofsec}
\end{align}
With Condition \eqref{eq:fyx}, \eqref{bofsec} implies 
\begin{align}
  \frac{ 1}{P_{Z_n}\big(I  \big)A_n}= \frac{\displaystyle \int_{\mathbb{R}^+} f_{X_n}( x)e^{-\psi_n(I ; x)x }    dx }{ P_{Z_n}\big(I \big)}   \leq \frac{1}{\delta_1}+O(a_n).  \label{bd_ZnAn}
\end{align}
By Conditions \eqref{eq:fyx} and \eqref{eq:fyx2}: $\big|P_{Y_n | X_n} \big(I  ; x \big) -  P_{Y_n}\big(I \big) \big| \leq b_n P_{Y_n}\big(I \big) $ with $b_n \rightarrow 0$, therefore
\begin{align}
 P_{Y_n}\big(I  \big) \leq \frac{P_{Y_n | X_n} \big(I ; x \big)}{1 - b_n} \leq K_1
 \label{bd_fYn}
\end{align}
for some constant $K_1>0$. With \eqref{bd_ZnAn} and \eqref{bd_fYn} and recalling the definition of $A_n$ in \eqref{ABh2}, we have 
\begin{align}
& \left|\frac{\displaystyle \int_{\mathbb{R}^+} f_{X_n}(x)e^{-\psi_n(I ; x)x } P_{Y_n | X_n} \big( I ; x \big) dx  }{P_{Z_n}\big( I  \big)}   -  \frac{P_{Y_n}\big(I  \big) }{ P_{Z_n}\big(I  \big) A_n} \right| \notag \\
= & \left| \frac{\displaystyle \int_{\mathbb{R}^+} f_{X_n}(x)e^{-\psi_n(I ; x)x } \left( P_{Y_n | X_n} \big( I ; x \big) -P_{Y_n}\big(I  \big) \right)  dx }{ P_{Z_n}\big(I \big)}    \right| \notag\\
 \leq & \frac{\displaystyle \int_{\mathbb{R}^+} f_{X_n}(x)e^{-\psi_n(I ; x)x } \left| P_{Y_n | X_n} \big( I ; x \big) -P_{Y_n}\big(I  \big) \right| dx }{P_{Z_n}\big(I \big)} \notag \\
 \leq & \frac{b_n P_{Y_n}\big(I  \big)  \displaystyle \int_{\mathbb{R}^+} f_{X_n}(x)e^{-\psi_n(I ; x)x }  dx }{ P_{Z_n}\big(I \big) } = \frac{b_n P_{Y_n}\big(I \big)  }{ P_{Z_n}\big(I \big) A_n}  
\leq  K_1 b_n \left(  \frac{1}{\delta_1} +O(a_n) \right) = O(b_n).
 \label{bofc2}
\end{align}
And similarly,
\begin{align}
\bigg\lvert \frac{P_{Y_n | X_n} \big(I ; x \big)}{P_{Z_n}\big(I \big) A_n  }  - \frac{P_{Y_n}\big(I \big) }{P_{Z_n}\big(I \big) A_n  }\bigg\rvert 
 & \leq  \frac{ b_n P_{Y_n}\big(I  \big) }{P_{Z_n}\big(I \big) A_n  }  = O(b_n).
  \label{bofc3}
\end{align}
By the triangle inequality, from  \eqref{bofsec}, \eqref{bofc2} and \eqref{bofc3}, we have
\begin{align*}
 \frac{ P_{Y_n | X_n} \big(I ; x \big)}{P_{Z_n}\big(I \big) A_n  }  = 1+O(a_n+b_n).
\end{align*}
Since $\log(1 + x) \leq x$ for all $x > -1$, for sufficiently large $n$, we have
\begin{align}
\log  \left( \frac{ P_{Y_n | X_n} \big( I ; x \big)}{P_{Z_n}\big(I \big) A_n }   \right) =O(a_n+b_n).
\label{bd_log}
\end{align}
Note that the term $O(a_n+b_n)$ in \eqref{bd_log} is independent of $x$. Therefore, for the first term in (\ref{small2}) we have
\begin{align}
 &\left| \displaystyle \int_{\mathbb{R}^+} A_n f_{X_n}(x)e^{-\psi_n(I ; x)x}  \log \left( \frac{  P_{Y_n | X_n} \big(I  ; x \big)  }{P_{Z_n}\big(I \big) A_n  }\right) dx \right|  \notag\\
 \leq &\sup_x \left| \log \left( \frac{  P_{Y_n | X_n} \big(I  ; x \big)  }{P_{Z_n}\big(I \big) A_n  }\right) \right|    \displaystyle \int_{\mathbb{R}^+} A_n f_{X_n}(x)e^{-\psi_n(I ; x)x}   dx \notag \\
 =& \sup_x \left| \log \left( \frac{  P_{Y_n | X_n} \big(I   ; x \big)  }{P_{Z_n}\big(I \big) A_n  }\right) \right|   \cdot 1 =  O(a_n+b_n).
\label{small2.5}
\end{align}
Define \[\hat{G}_\delta(y, x) := \log P\big( Y_n \in [y, y+\delta]  \mid X_n = x \big).\] 
Then by Taylor expansion and the conditions \eqref{eq:logsquare}, \eqref{eq:fyx}, we can expand $ \hat{G}_\delta(h - x, x)$ at $(h, x)$  to get
\begin{align} \label{taylor.log}
\hat{G}_\delta(h - x, x)  &= \hat{G}_\delta(h , x) - \frac{\partial \hat{G}_\delta(h , x) }{\partial y} x  + \frac{\partial^2 \hat{G}_\delta(h - \hat{\alpha}_n x , x) }{2 \partial y^2} x^2,  \quad \text{for some } \hat{\alpha}_n \in (0,1), \notag \\
& =  \hat{G}_\delta(h , x) - \psi_n(I  ; x) x + q_n(x)x^2,
\end{align}
where $\displaystyle
q_n(x):= \frac{1}{2}\frac{\partial^2 \log P \big(Y_n \in [y, y + \delta] \mid X_n =  x\big)}{\partial y^2} \biggr\rvert_{y = h -\hat{\alpha}_n x} $. Therefore, for the second term in (\ref{small2}),  by   \eqref{taylor.log},  we can get
\begin{align*}
 \log \left( \frac{ P \big( Y_n \in [ h - x, h +\delta - x] \mid X_n =  x \big) }{  P \big( Y_n \in [ h, h +\delta] \mid X_n = x \big) e^{-\psi_n(I ; x) x}}\right)  
=& \hat{G}_\delta(h - x, x)- \hat{G}_\delta(h , x)  + \psi_n(I  ; x) x =q_n(x)x^2.
\end{align*}
 And by Condition \eqref{eq:logsquare}, for all $x\in \mathbb{R}^+$, there is a constant $K_2>0$ such that 
\begin{align}
|e^{-\psi_n(I ; x)x} q_n(x)| \leq K_2.
\label{taylor-qn2}
\end{align}
In the following, we use a brief notation $P_{Y_n \mid X_n} \big(E_n -x; x   \big) =  P \big(Y_n \in [y, y + \delta] \mid X_n =  x\big)$. By \eqref{taylor-qn2}, and the uniform boundedness of $A_n$, and the assumption: $\mathbb E \big[X^2_n \big] = a_n$,  the second term in \eqref{small2} satisfies
\begin{align}
\label{small3}
&\Bigg\rvert \displaystyle \int_{\mathbb{R}^+} A_n f_{X_n}(x)e^{-\psi_n(I ; x)x}\log \left( \frac{ P_{Y_n | X_n} \big( I -x  ; x \big)}{  P_{Y_n | X_n} \big( I ; x \big) e^{-\psi_n(I ; x) x}  }\right) dx \Bigg\rvert \notag\\
= &  \Bigg\rvert \displaystyle \int_{\mathbb{R}^+} A_n e^{-\psi_n(I ; x)x} f_{X_n}(x)  q_n(x)x^2   dx  \Bigg\rvert       \leq  M \Bigg\rvert \displaystyle \int_{\mathbb{R}^+}  f_{X_n}(x) e^{-\psi_n(I ; x)x} q_n(x)x^2  dx  \Bigg\rvert 
\leq  M  K_2 \mathbb E \big[X^2_n \big]=O( a_n).   
\end{align}
Combining \eqref{KL1}, \eqref{small2}, \eqref{small2.5} and \eqref{small3},
\begin{align}
D_{\text{KL}}\left ({\mathbb P}^{(n)}_I ~\|~ {\mathbb Q}^{(n)}_I\right)=O(a_n+b_n).
\end{align}

For the case $S_n := \{x: f_{X_n}(x)>0 \} \subset \mathbb{R}^+$, we can only define $P(Z_n \in I \mid X_n= x)$ on $S_n$. But we can still  define the KL-divergence on $\mathbb{R}^+$ since the part of KL-divergence on $\mathbb{R}^+ \backslash S_n $ is $0$.
Therefore, in the same way as \eqref{eq:KL1},
\begin{align}\label{KL2}
D_{\text{KL}}\left(\mathbb P^{(n)}_I ~\|~ \mathbb Q^{(n)}_I \right) 
&= \int_{\mathbb{R}^+} A_n f_{X_n}(x) e^{-\psi_n(I ; x) x} \log \left(  \frac{A_n f_{X_n}(x) e^{-\psi_n(I ; x) x}}{ f_{X_n|Z_n}(x, I ) }\right) dx \notag \\
 &= \int_{S_n} A_n f_{X_n}(x) e^{-\psi_n(I ; x) x} \log \left(  \frac{A_n f_{X_n}(x) e^{-\psi_n(I ; x) x}}{ f_{X_n|Z_n}(x, I ) }\right) dx \notag \\
&= \left| \int_{S_n} A_n f_{X_n}(x) e^{-\psi_n(I ; x) x} \log \left( \frac{f_{X_n|Z_n}(x, I )}{ A_n f_{X_n}(x) e^{-\psi_n(I ; x) x}}\right) dx\right|.
\end{align}
Then we can follow every step from the step \eqref{KL1} in our proof for the case $\{x: f_{X_n}(x)>0 \} = \mathbb{R}^+$ to get 
\begin{align}
D_{\text{KL}}\left ({\mathbb P}^{(n)}_I ~\|~ {\mathbb Q}^{(n)}_I\right)=O(a_n+b_n).
\end{align}
Furthermore, let ${\zeta}_n(I; x) := \beta_n {\psi}_n(I; \beta_n x)$, by the condition \eqref{eq:fyx}, there is a constant $C>0$ such that for all $x \in \mathbb{R}^+,$
$ 0 \leq {\zeta}_n(I; x) < C.$
Therefore,   ${\mathbb P}^{(n)}_I$ with the density function $A_n f_X(x) e^{- \beta_n {\psi}_n(I; \beta_n x) x} $ satisfies the definition of the canonical probability distributions in \eqref{def:cont.can.dis}.
\end{proof}

\subsubsection{Proof of Theorem \ref{thm:d+d}}

For a finite interval $ I = [h, h + \delta], \ h,\delta \in \mathbb{R}$ and $\delta>0$, let 
\begin{align*}
\mathbb{\hat{P}}^{(n)}_I =P(K_n =  \beta_n k \mid Z_n \in I) \quad \text{and} \quad \mathbb{\hat{Q}}^{(n)}_I = B_n P(K_n =  \beta_n k ) e^{-\hat{\psi}(I;   \beta_n k )  \beta_n k },
\end{align*} 
where 
\begin{align*}
\frac{1}{B_n} := {\sum_{ \beta_n k \in S_n} P(K_n=  \beta_n k )  e^{-\hat{\psi}(I ;  \beta_n k) \beta_n k }} 
\end{align*}
and
\begin{align*}
    \hat{\psi}_n(I ; \beta_n k):= \frac{\partial \log P\big( Y_n  \in [y, y + \delta ] \mid  K_n =  \beta_n k \big)}{\partial y}\biggr\rvert_{y=h}. 
\end{align*}
We first state the following lemma. The  proof   follows from the proof of Theorem \ref{thm:theorem 1} with the Definition of KL-divergence for discrete probability distributions in \eqref{def:QKL}.
\begin{lemma}\label{lm:d+c}
 Assume there exist positive constants $\delta_1, \delta_2,\{C_i, 1\leq i\leq 3\}$, a sequence $b_n = o(1)$, and an open interval $D$ such that the conditions \eqref{eq:logsquare} -- \eqref{eq:condforz} in Theorem \ref{thm:theorem 1} hold for $K_n$, $Y_n$, and $Z_n$. Then
 \begin{align}
D_{\textnormal{KL}}\left(\hat{\mathbb P}^{(n)}_I \|~ \hat{\mathbb Q}^{(n)}_I\right) =O(a_n+b_n), \quad \text{for every} \ I \subseteq D.
\end{align}
\end{lemma}

Now we are ready to prove Theorem \ref{thm:d+d}.
\begin{proof}[Proof of Theorem \ref{thm:d+d}]
All of the conditions in Theorem \ref{thm:theorem 1} hold for $K_n, Y_n, Z_n$ by the assumptions, hence  Lemma \ref{lm:d+c} can be applied. Therefore, we obtain the following relation between total variation and KL-divergence from \eqref{eq:totalvariation}:  for every $I \subseteq D$,  
\begin{align} 
\label{eq:TV}
& \sup_{\beta_n k \in S_n}\biggr\lvert P\big( K_n =  \beta_n k \mid Z_n \in  I \big) -  B_n P(K_n =   \beta_n k) e^{-\hat{\psi}_n(I;  \beta_n k)  \beta_n k} \biggr\rvert \notag \\
\leq & \delta\left(\hat{\mathbb P}^{(n)}_I, \hat{\mathbb Q}^{(n)}_I \right)\leq \frac{1}{2} \sqrt{D_{\textnormal{KL}}\left( 
\hat{\mathbb P}^{(n)}_I \|~\hat{\mathbb Q}^{(n)}_I\right)}=O\left(\sqrt{a_n+b_n}\right).
\end{align}
With \eqref{eq:domain} and \eqref{eq:TV}, the conclusion \eqref{eq:cor55triangle} follows from the change of variable $K_n =\beta_n K $ and the triangle inequality: 
\begin{align*}
& \sup_{k \in S}\biggr\lvert  P\big( K = k \mid \tilde{H}_n \in  E_n  \big) -  B_n P(K =  k) e^{- \beta_n \hat{\psi}_n(I; \beta_n k) k} \biggr\rvert \\
=&  \sup_{ \beta_n k \in S_n }\biggr\lvert P\big( K_n =  \beta_n k \mid {H}_n \in  I  \big) -  B_n P(K_n  =   \beta_n k) e^{-  \hat{\psi}_n(I;  \beta_n k)  \beta_n k} \biggr\rvert \notag \\
\leq& \sup_{ \beta_n k \in S_n}\biggr\lvert P\big( K_n =  \beta_n k \mid Z_n \in  I \big) -  B_n P(K_n =   \beta_n k) e^{-\hat{\psi}_n(I;  \beta_n k)  \beta_n k} \biggr\rvert \notag \\
& + \sup_{ \beta_n k \in S_n}\biggr\lvert  P\big( K_n =  \beta_n k \mid H_n \in  I \big) -  P\big( K_n =  \beta_n k \mid Z_n \in  I \big)  \biggr\rvert \\
=&  O(c_n + \sqrt{a_n + b_n}).
\end{align*}

Furthermore, let $\hat{\zeta}_n(I; k) := \beta_n \hat{\psi}_n(I; \beta_n k)$.  We can check that 
$ 0 \leq \hat{\zeta}_n(I; k) < C$  for all $\ k\in S$ and a constant $C>0$.
Therefore,  $B_n P(K =  k) e^{- \beta_n \hat{\psi}_n(I; \beta_n k) k} $ satisfies the definition of the canonical probability distributions in \eqref{def:disc.can.dis}.
\end{proof}

\subsection{Proofs of  Theorem \ref{thm:limit.smooth} and Theorem \ref{thm:limit.ldp}}

Let $X$ be a nonnegative continuous random variable and  with $\mathbb{E}[X] < \infty$ and let $Z_n$ be a sequence of real-valued continuous random variables. Given a Borel measurable set $ E \in  \mathcal{B}\left( \mathbb{R}\right)$ and a function $\psi: \mathcal{B}\left( \mathbb{R}\right) \rightarrow \mathbb{R}$ with $ 0 < \psi(E) < \infty$, let $\mathbb P_E$ be a probability measure with density function 
$$\displaystyle A f_{X}(x) e^{-\psi(E)x}, \quad \displaystyle \frac{1}{A} : = {\displaystyle \int_{\mathbb{R}^+} f_{X}(x) e^{-\psi(E)x}dx}.  $$  
And let $\mathbb Q^{(n)}_E$ be a probability measure with density function $f_{X|Z_n}(x; E)$.  We obtain the following lemma for the case \eqref{uniq.case2} of the canonical distribution \eqref{uniq.can.dis}:

 \begin{lemma}
\label{thm:unique1}
Assume the following conditions hold: 
\begin{enumerate}
\item (Boundedness) \label{cond:temp state 1} 
$\displaystyle  \left|  \frac{f_{X \mid Z_n} (x;E)}{f_{X}(x)} \right|$ and $\displaystyle  \left|  e^{- \xi x} \log \left( \frac{f_{X \mid Z_n} (x; E)   }{f_{X}(x) }  \right) \right|$, for any $\xi>0$, are uniformly bounded on $ \mathbb{R}^+$.
\item (Linear approximation)  \label{cond:temp state 2} 
There exist constants $b, c \in \mathbb{R}, \ 0 < c < \infty$, and a sequence of functions $q_n : \mathbb{R} \rightarrow \mathbb{R} $ with  \[\mathbb{E}\left[ q_n(X)^2 \right] = \gamma_n \rightarrow 0\] such that on an interval $I_n = [0, d_n]$ with $d_n \rightarrow \infty$,
\begin{align} 
&   \log \left(   \frac{f_{X \mid Z_n} (x; E)}{f_{X}(x)} \right) = b - cx + q_n(x). 
\label{cond:linear.approx}
\end{align}

\end{enumerate}
Then 
$$ \displaystyle  \lim_{n \rightarrow \infty}{D_{\textnormal{KL}}\left(\mathbb P_E ~ \|~ \mathbb Q^{(n)}_E \right)} = 0  \quad \text{if and only if} \quad   c = \psi(E). $$

\end{lemma}

Furthermore, assume $\mathbb{E}[X^3] < \infty$ and $X$ is not a constant random variable, let $\tilde{\mathbb{P}}^{(n)}_E $ be a probability measure with density function 
\begin{align}
\displaystyle \tilde A_n f_X(x) e^{ -  \beta_n  \psi(E)x  }, \quad \frac{1}{\tilde{A}_n} := \displaystyle \int_{\mathbb{R}^+}  f_X(x) e^{-   \beta_n \psi(E)  x  },
\end{align}
in which $\beta_n >0, \ \beta_n = o(1)$. We obtain the following lemma for the case \eqref{uniq.case1} of the canonical distribution \eqref{uniq.can.dis}:

\begin{lemma}

\label{thm:unique2}
Assume the following conditions hold : 
\begin{enumerate}
\item (Boundedness) \label{cond:cor:temp state 1} 
$\displaystyle \left| \frac{f_{X \mid Z_n} (x;E)}{f_{X}(x)} \right| $ and $\displaystyle \left| e^{-\beta_n \xi x} \log \left( \frac{f_{X \mid Z_n} (x; E)   }{f_{X}(x) }  \right) \right|$, for any $\xi>0$, are uniformly bounded on $ \mathbb{R}^+$.
\item  (Linear approximation) \label{cond:cor:temp state 2} 
There exist constants $b,c \in \mathbb{R}$, $ 0 < c < \infty$,  and a sequence of functions $q_n : \mathbb{R} \rightarrow \mathbb{R} $ with $\mathbb{E}\left[ q_n(X)^2 \right] \rightarrow 0$ such that  on $I_n = [0, d_n]$ with $d_n = O\left(\frac{1}{\beta_n}\right)$,
\begin{align} 
&  \frac{1}{\beta_n} \log \left(   \frac{f_{X \mid Z_n} (x; E)}{f_{X}(x)} \right) = b - c x + q_n(x). 
\label{cond:linear.approx2}
\end{align}

\end{enumerate}
Then 
$$ \displaystyle \lim_{n \rightarrow \infty}\frac{{D_{\textnormal{KL}}\left( \tilde{\mathbb P}^{(n)}_E ~ \|~ {\mathbb Q}^{(n)}_E \right)}}{\beta_n^2}  
 = 0  \quad \text{if and only if} \quad c=  \psi(E). $$

\end{lemma}

\begin{remark}
In particular, if we choose $Z_n = \beta_n X + \beta_n ( \tilde{Y}_n  - \mu_n),$ where $\tilde{Y}_n, \beta_n, \mu_n$ are given in the definitions in Section \ref{ch:existence}, and choose the Borel set $E$ to be a finite interval $I$, by Equation \eqref{Section:existence:three.rep}, those general results of Lemma \ref{thm:unique1} and Lemma \ref{thm:unique2} for $f_{X \mid Z_n}(x, E) $ can be applied to $f_{X \mid \tilde{Z}_n} (x, E_n )$, which is the conditional density defined in Section \ref{ch:existence}. 
\end{remark}

\subsubsection{Proof of Lemma \ref{thm:unique1} }

\begin{proof}

Note that for any uniformly bounded function $\left| b_n(x) \right|$ on $\mathbb{R}^+$:
\begin{align}
\left| \int_{\mathbb{R}^+ \backslash I_n} f_{X}(x) b_n(x) dx \right| &\leq  \| b_n(x)  \|_{\infty}\int_{\mathbb{R}^+ \backslash I_n} f_{X}(x) dx 
= \|  b_n(x) \|_{\infty} \mathbb{P}\left( X \geq d_n \right)  \notag \\ 
&\leq  \|  b_n(x) \|_{\infty}  \left( \frac{ \mathbb{E} \left[  X \right]}{d_n} \right) = O(\epsilon_n),
\label{temp:markov} 
\end{align}
for a sequence $\epsilon_n \rightarrow 0$ since $d_n \rightarrow \infty$ by Condition \eqref{cond:temp state 2} and $\mathbb{E}[X]$ is bounded by the assumption. 

We first prove $ \displaystyle   c = \psi(E) \ \Rightarrow \ {D_{\textnormal{KL}}\left(\mathbb P_E~ \|~ \mathbb Q^{(n)}_E\right)}
 \rightarrow 0$.
 
By Condition \eqref{cond:temp state 2}, 
\begin{align}
 \log \left( \frac{f_{X \mid Z_n} (x; E) }{ f_{X}(x) }  \right) =b -\psi(E)x  + q_n(x) \quad \text{on} \ I_n,
\label{temp exap 1}
\end{align}
therefore, we have 
\begin{align}
 \log \left( \frac{A f_{X}(x) e^{-\psi(E) x}}{f_{X \mid Z_n} (x; E)}  \right) = \log A  - b - q_n(x)  \quad \text{on} \ I_n.
\label{temp exap 2}
\end{align}
Since $\int_{\mathbb{R}^+} f_X(x) dx =1,$
there exists a bounded closed set $D \subset \mathbb{R}^+$ such that  $ \int_{D} f_X(x) dx > 0$. 
 Hence,
\begin{align}
  \displaystyle  A =  \frac{1}{  \displaystyle  \int_{\mathbb{R}^+} f_{X}(x) e^{-\psi(E)x} dx} \leq \frac{1}{  \displaystyle  \int_{D} f_{X}(x) e^{-\psi(E)x} dx} \leq \frac{1}{  \displaystyle \inf_{x \in D } \left|  e^{-\psi(E)x} \right|  \int_{D} f_{X}(x) dx}  < \infty.
\label{temp A_n}
\end{align}

Furthermore, we can derive
\begin{align}
1 = \int_{\mathbb{R}^+} f_{X\mid Z_n} (x; E) dx &=\int_{I_n} f_{X}(x) \frac{ f_{X \mid Z_n} (x; E)}{f_{X}(x)} dx +  \int_{\mathbb{R}^+ \backslash I_n} f_{X}(x) \frac{ f_{X \mid Z_n} (x; E)}{f_{X}(x)} dx \notag \\
&= \int_{I_n} f_{X}(x) e^ {  b -\psi(E)x  + q_n(x)} dx + O(\epsilon_n),
\label{temp A_n 1}
\end{align}
in which the last equality is from Equation \eqref{temp exap 1}, and the result of \eqref{temp:markov} applied to the uniformly bounded function $\displaystyle \left| b_n(x) \right| = \left| \frac{ f_{X \mid Z_n} (x; E)}{f_{X}(x)}\right| $ on $\mathbb{R}^+$ (Condition \eqref{cond:temp state 1}). Multiplying by $e^{-b}$ on both sides in \eqref{temp A_n 1}, we have 
\begin{align}
e^{-b}=	\int_{I_n} f_{X}(x) e^ { -\psi(E)x  + q_n(x)} dx + O(\epsilon_n).
\end{align}
Then we can apply Taylor's expansion to $e^{ q_n(x)}$ to get
\begin{align}
e^{-b} 
&= \int_{I_n} f_{X}(x) e^{-\psi(E)x} dx  +  \int_{I_n} f_{X}(x) e^{-\psi(E)x} \left( q_n(x)   + \frac{ q_n(x)  ^2}{2} e^{\alpha_n \cdot q_n(x)} \right)dx  + O(\epsilon_n),
 \label{temp A_n 2}
\end{align}
for some sequence $\alpha_n \in (0,1)$. Note that we use the formula \[\displaystyle e^y = 1 + y + \frac{e^{\alpha(y) \cdot y}}{2} y^2, \quad \alpha(y) \in (0,1)\] and let $y = q_n(x), \ \alpha_n = \alpha(q_n(x))$.  Combined with Equation \eqref{temp exap 1} and Condition \eqref{cond:temp state 1}, it implies there exists a constant $M>1$   independent of $n$ such that $e^{-\psi(E)x + q_n(x)} \leq M$ for all $x \in I_n$. Since $\alpha_n \in (0,1)$ and $\psi(E)>0$ in the assumption,
\begin{align}
   e^{- \psi(E)x + \alpha_n \cdot q_n(x) } \leq M^{\alpha_n}   e^{- (1-\alpha_n)\psi(E)x}\leq M  \quad \text{for all} \ x \in I_n.  
\label{temp upper bound}
\end{align}
Hence $e^{- \psi(E)x + \alpha_n  q_n(x)   } $ is uniformly bounded on $ I_n $. Then  
\begin{align}
\int_{I_n} f_{X}(x) e^{-\psi(E)x} \left(  \frac{ q_n(x)    ^2}{2} e^{\alpha_n \cdot q_n(x)} \right)dx   
&\leq  \left\| \frac{e^{- \psi(E)x + \alpha_n \cdot q_n(x) }}{2}\right\| _{\infty} \int_{I_n} f_{X}(x) q_n(x)^2 dx \notag  \\
& \leq  M \mathbb{E}\left[ q_n(X)^2 \right]  = O(\gamma_n),
\label{temp A_n 3}
\end{align}
where $ O(\gamma_n) \rightarrow 0$  by Condition \eqref{cond:temp state 2}. By Equations \eqref{temp A_n 2} and \eqref{temp A_n 3}, 
\begin{align}
e^{-b} &\leq   \int_{I_n} f_{X}(x) e^{-\psi(E)x} dx +  \int_{I_n} f_{X}(x) e^{-\psi(E)x} q_n(x) dx +  O(\gamma_n)  + O(\epsilon_n) \notag \\
&= \int_{\mathbb{R}^+} f_{X}(x) e^{-\psi(E)x} dx  - \int_{\mathbb{R}^+ \backslash I_n} f_{X}(x) e^{-\psi(E)x} dx +  \int_{I_n} f_{X}(x) e^{-\psi(E)x} q_n(x) dx +  O(\gamma_n)  + O(\epsilon_n) \notag \\
&= \frac{1}{A}   +  \int_{I_n} f_{X}(x) e^{-\psi(E)x}  q_n(x)   dx + O(\gamma_n)  + O(\epsilon_n),
\label{temp A_n 3.5}
\end{align}
where the last equation is from the result of \eqref{temp:markov}. And since $A$ is bounded by the result \eqref{temp A_n}, we have
\begin{align}
A e^{-b} \leq 1 +   \int_{I_n} A f_{X}(x) e^{-\psi(E)x}  q_n(x)   dx + O(\gamma_n)  + O(\epsilon_n).
\label{temp A_n 4}
\end{align}
Using the inequality $\log(1+x ) \leq x$ for all $x > -1$, we find a bound
\begin{align}
\log A - b \leq \int_{I_n} A f_{X}(x) e^{-\psi(E)x}  q_n(x)   dx  + O(\gamma_n)  + O(\epsilon_n).
\label{temp A_n 5}
\end{align}
Furthermore, by Condition \eqref{cond:cor:temp state 1}, $  \left|  e^{-\psi(E) x} \log \left( \frac{f_{X \mid Z_n} (x; E)   }{f_{X}(x) }  \right) \right| $ is uniformly bounded on $ \mathbb{R}^+$, so we can check that 
\begin{align}
 \left| e^{-\psi(E) x} \log \left( \frac{A f_{X}(x) e^{-\psi(E) x}}{f_{X \mid Z_n} (x; E)}  \right)\right| \ \text{is uniformly bounded on} \ \mathbb{R}^+  \text{as well.}
\label{temp A_n 5.5}
\end{align}

Recall that 
$$ \log \left( \frac{A f_{X}(x) e^{-\psi(E) x}}{f_{X \mid Z_n} (x; E)}  \right) = \log A  - b - q_n(x) \ \text{on} \ I_n $$ 
by \eqref{temp exap 2}. With the result of \eqref{temp A_n 5}, we can get  
\begin{align}
& D_{\textnormal{KL}}\left(\mathbb P_E ~ \|~ \mathbb Q^{(n)}_E \right) \notag \\
 &= \int_{I_n} A f_{X}(x) e^{-\psi(E) x} \log \left( \frac{A f_{X}(x) e^{-\psi(E) x}}{f_{X \mid Z_n} (x; E)}  \right) dx + \int_{\mathbb{R}^+ \backslash  I_n} A f_{X}(x) e^{-\psi(E) x} \log \left( \frac{A f_{X}(x) e^{-\psi(E) x}}{f_{X \mid Z_n} (x; E)}  \right)dx \notag \\
 &= \int_{I_n} A f_{X}(x) e^{-\psi(E) x} \left( \log A - b\right)dx -  \int_{I_n} A f_{X}(x)e^{-\psi(E)x}    q_n(x)   dx + O(\epsilon_n)  \notag \\
  &=\left( \log A - b\right) - \int_{\mathbb{R}^+ \backslash I_n} A f_{X}(x) e^{-\psi(E) x} \left( \log A - b\right)dx -  \int_{I_n} A f_{X}(x)e^{-\psi(E)x}    q_n(x)   dx + O(\epsilon_n)  \notag \\
 &= \log A - b + O(\epsilon_n) -  \int_{I_n} A f_{X}(x)e^{-\psi(E)x}  q_n(x)  dx + O(\epsilon_n) =  O(\gamma_n)  + O(\epsilon_n),
 \label{temp KL}
\end{align}
where the $O(\epsilon_n)$ terms are from the result of \eqref{temp:markov} applied to the bounded function \eqref{temp A_n 5.5}. 
Therefore, by \eqref{temp A_n} and \eqref{temp KL}, we get\begin{align*}
D_{\textnormal{KL}}\left(\mathbb P_E ~ \|~ \mathbb Q^{(n)}_E\right) \rightarrow 0.
\end{align*}

\bigskip


Next we  prove 
\begin{align} \label{assump:KL}
    {D_{\textnormal{KL}}\left(\mathbb P_E ~ \|~ \mathbb Q^{(n)}_E \right)}
 \rightarrow 0  \displaystyle  \  \Rightarrow  \  c = \psi(E).
\end{align}
By Condition \eqref{cond:temp state 2}, there exists a constant $\hat {b}$ and a sequence of functions $\hat{q}_n(x)$ such that 
\begin{align}
\log \left( \frac{f_{X \mid Z_n} (x; E) }{ f_{X}(x) }  \right) = b - cx   +  q_n(x) \quad \text{on} \ I_n.
\label{temp c exap 1}
\end{align}
Similar to the derivation of \eqref{temp exap 2}, we have
\begin{align}
\log \left( \frac{\hat{A} f_{X}(x) e^{-c x}}{f_{X \mid Z_n} (x; E)}  \right) =  \log \hat{A}  - b - q_n(x)   \quad \text{on} \ I_n,
\label{temp c exap 2}
\end{align}
where 
\begin{align}
\hat{A} = \frac{1}{\displaystyle \int_{\mathbb{R}^+} f_{X}(x) e^{-cx} dx}  < \infty,
\label{temp c A_n}
\end{align}
which  can be proved by a similar approach as in \eqref{temp A_n}. Then following the previous proof from \eqref{temp A_n 1} to \eqref{temp KL}, we can get
\begin{align}
D_{\textnormal{KL}} \left( \hat{\mathbb{P}}_E ~ \|~ \mathbb Q^{(n)}_E\right)  \rightarrow  0,
\end{align}
where $\hat{\mathbb{P}}_E$ is a probability measure with density function $\hat{A}f_X(x) e^{-cx }$. By the assumption \eqref{assump:KL}, we also know
\begin{align}
D_{\textnormal{KL}} \left( \mathbb{P}_E ~ \|~ \mathbb Q^{(n)}_E \right)  \rightarrow  0.
\end{align}
By Pinsker's inequality \eqref{eq:totalvariation}, we have that the total variation distance denoted by $\delta(\cdot,\cdot)$ satisfies
\begin{align}
\delta(\hat{\mathbb{P}}_E, \mathbb Q^{(n)}_E ) \rightarrow 0 \quad \text{and} \quad  \delta(\mathbb{P}_E, \mathbb Q^{(n)}_E) \rightarrow 0.
\end{align} 
Then by the triangle inequality,
$
\delta(\hat{\mathbb{P}}_E, \mathbb{P}_E)=  0. $ It implies
\begin{align*}
\int_0^x \left( \hat{A} f_X(s)  e^{-cs}ds -  A   f_X(s) e^{-\psi(E)s} \right)ds = 0, \quad \text{ for all }  x\in\ \mathbb{R}^+.
\end{align*}
Hence 
\begin{align*}
\hat{A}  f_X(x)  e^{-cx}  = A   f_X(x) e^{-\psi(E)x} \quad \text{holds almost everywhere on} \ \mathbb{R}^+.
\end{align*}
Since $\hat{A}$ and $A$ are both  independent  of $x$ and there exists an interval such that $f_X(x)>0 $, we  obtain  $c = \psi(E)$. 
\end{proof}

\subsubsection{Proof of Lemma \ref{thm:unique2} }

\begin{proof}
\label{proof:lemma:unique2}
Note that  for any uniform bounded function $\left| b_n(x) \right|$ on $\mathbb{R}^+$:
\begin{align}
\left| \int_{\mathbb{R}^+ \backslash I_n} f_{X}(x) b_n(x) dx \right| &\leq  \| b_n(x)  \|_{\infty}\int_{\mathbb{R}^+ \backslash I_n} f_{X}(x) dx 
= \|  b_n(x) \|_{\infty} \mathbb{P}\left( X \geq d_n \right)  \notag \\ 
&\leq  \|  b_n(x) \|_{\infty}  \left( \frac{ \mathbb{E} \left[  X^3 \right]}{d^3_n} \right) = O(\beta_n^3),
\label{cor:temp:markov} 
\end{align}
where the  existence of $O(\beta_n^3)$ is due to $d_n = O(\frac{1}{\beta_n})$  by Condition \eqref{cond:cor:temp state 2} and $\mathbb{E}[X^3] < \infty$ by the assumption.

We first prove $  c=  \psi(E) \ \Rightarrow \ \displaystyle \frac{{D_{\textnormal{KL}}\left(\mathbb {\tilde{P}}^{(n)}_E ~ \|~ \mathbb Q^{(n)}_E \right)}}{\beta_n^2}  
 \rightarrow 0 $. By Condition \eqref{cond:cor:temp state 2},
\begin{align} 
&   \log \left(   \frac{f_{X \mid Z_n} (x; E)}{f_{X}(x)} \right) = \beta_n \left( b - \psi(E)  x + q_n(x) \right) \quad \text{ on} \ I_n,
\label{temp.n exap 1}
\end{align}
Therefore we have 
\begin{align}
 \log \left( \frac{A_n f_{X}(x) e^{- \beta_n \psi(E) x}}{f_{X \mid Z_n} (x; E)}  \right) = \log A_n  - \beta_n b -  \beta_n q_n(x)  \quad \text{on} \ I_n.
\label{temp.n exap 2}
\end{align}
Following the proof in \eqref{temp A_n}, for each $n$, we have 
\begin{align}
\displaystyle  A_n =  \frac{1}{  \displaystyle  \int_{\mathbb{R}^+} f_{X}(x) e^{-  \beta_n  \psi(E) x} dx} < \infty,
\end{align}
and we can check that $\displaystyle  \lim\limits_{n \rightarrow \infty}  \int_{\mathbb{R}^+}  f_X(x) e^{ - \beta_n\psi(E) x  } dx \rightarrow 1 $, hence, $A_n$ is uniformly bounded.

We can apply a similar proof as for Lemma \ref{thm:unique1} to Equation \eqref{temp.n exap 2}. Substituting $b$ by $\beta_n b$, $\psi(E)x$ by  $\beta_n \psi(E)x$, $q_n(x)$ by $\beta_n q_n(x)$ and $A$ by $A_n$, then every step from Equation \eqref{temp A_n 1} to Equation \eqref{temp KL} follows. Therefore, we can get 
\[{D_{\textnormal{KL}}\left(\mathbb {\tilde{P}}^{(n)}_E ~ \|~ \mathbb Q^{(n)}_E \right)} = O(\beta_n^2 \gamma_n)+ O(\beta_n^3 ),\]
where the $O(\beta_n^2 \gamma_n)$ term follows from the derivation of the $O( \gamma_n)$ term in Lemma \ref{thm:unique1}, the $O(\beta_n^3)$ term follows from Equation \eqref{cor:temp:markov} and the derivation of the $O( \epsilon_n)$ term in Lemma \ref{thm:unique1}. It implies 
\begin{align*}
 \frac{ {D_{\textnormal{KL}}\left(\mathbb {\tilde{P}}^{(n)}_E ~ \|~ \mathbb Q^{(n)}_E \right)} }{\beta_n^2} = O(\gamma_n)+O(\beta_n ) \rightarrow 0.
\end{align*}
 
 \bigskip
 
Next we prove $ \displaystyle \frac{{D_{\textnormal{KL}}\left(\mathbb {\tilde{P}}^{(n)}_E ~ \|~ \mathbb Q^{(n)}_E \right)}}{\beta_n^2}  
 \rightarrow 0  \ \Rightarrow \  c=  \psi(E)$. Similar to the proof for Lemma \ref{thm:unique1}, we can show
\begin{align}
\frac{ D_{\textnormal{KL}} \left( \hat{\mathbb{P}}^{(n)}_E ~ \|~ \mathbb Q^{(n)}_E \right) }{\beta_n^2} \rightarrow  0,
\end{align}
where $\hat{\mathbb{P}}^{(n)}_E$ is a probability measure with density function $\hat{A}_{n} f_X e^{-\beta_n  cx }$. By the assumption, we also know
\begin{align}
\frac{ D_{\textnormal{KL}} \left( \mathbb{{\tilde{P}}}^{(n)}_E ~ \|~ \mathbb Q^{(n)}_E \right) }{\beta_n^2} \rightarrow  0.
\end{align}
Therefore, by Pinsker's inequality, we have that the  total variation distance denoted by $\delta(\cdot,\cdot)$ satisfies
\begin{align}
\frac{1}{\beta_n}\delta(\hat{\mathbb{P}}_E, \mathbb Q^{(n)}_E ) \rightarrow 0 \quad \text{and} \quad  \frac{1}{\beta_n} \delta(\mathbb{{\tilde{P}}}_E, \mathbb Q^{(n)}_E ) \rightarrow 0.
\end{align} 
Then by the triangle inequality,
$ \frac{1}{\beta_n} \delta(\hat{\mathbb{P}}_E, \mathbb{{\tilde{P}}}_E) \rightarrow  0. $ It implies
\begin{align}
\lim\limits_{n \rightarrow \infty} \frac{1}{\beta_n}\left( \int_0^x \hat{A}_n f_X(s) e^{-\beta_n cs}ds - \int_{0}^x  A_n f_X(s) e^{-\beta_n\psi(E)s}ds \right) = 0, \quad \text{ for all }  x\in\ \mathbb{R}^+.
\label{cor:temp:limintegral}
\end{align}
We can apply Taylor's expansion to $e^{-\beta_n cs}$ and $e^{-\beta_n\psi(E)s}$ to get
\begin{align}
e^{-\beta_n cs} = 1 -\beta_n cs + O(\beta_n^2 s^2) \quad \text{and} \quad e^{-\beta_n \psi(E)s} = 1 -\beta_n \psi(E)s + O(\beta_n^2 s^2).
\label{cor:temp:taylor}
\end{align}
By the results of  \eqref{cor:temp:taylor}, Equation \eqref{cor:temp:limintegral} can be written as
\begin{align}
\lim\limits_{n \rightarrow \infty}  \int_0^x  \frac{1}{\beta_n}\left( \tilde{A}_n - A_n - \left(  \tilde{A}_n c   -  A_n   \psi(E) \right)\beta_n s + O(\beta_n^2 s^2) \right) f_X(s) ds = 0, \quad \text{ for all }  x \in  \mathbb{R}^+.
\label{cor:temp:limintegral 2}
\end{align}
Since $\mathbb{E}[X^2] < \infty$ from the fact $\mathbb{E}[X^3] < \infty$, we know $\displaystyle \int_0^x s^2 f_X(s)ds  < \infty $ on $ \mathbb{R}^+$. 
Therefore, the $ O(\beta_n^2 s^2)$ in Equation \eqref{cor:temp:limintegral 2} can be dropped and we obtain
\begin{align}
\lim\limits_{n \rightarrow \infty}  \int_0^x  \frac{1}{\beta_n}\left( \tilde{A}_n - A_n - \left(  \tilde{A}_n c   -  A_n   \psi(E) \right)\beta_n s  \right) f_X(s) ds = 0, \quad \text{ for all }  x\in\ \mathbb{R}^+.
\label{cor:temp:limintegral4}
\end{align}
By the Dominated Convergence Theorem, 
\begin{align}
 \hat{A}_n = \frac{1}{    \int_{\mathbb{R}^+} f_{X}(x) e^{-  \beta_n c x} dx}  \rightarrow 1 \quad \text{and} \quad    A_n = \frac{1}{    \int_{\mathbb{R}^+} f_{X}(x) e^{- \beta_n \psi_n(E)x} dx} \rightarrow 1.
 \label{cor:temp:A}
\end{align} 
Also we have
\begin{align}
&\lim\limits_{n \rightarrow \infty} \frac{1}{\beta_n}\left(   \int_{\mathbb{R}^+} f_{X}(x) e^{- \beta_n \psi_n(E)x} dx - \int_{\mathbb{R}^+} f_{X}(x) e^{-  \beta_n c x} dx \right) \notag \\
= & \lim\limits_{n \rightarrow \infty} \frac{1}{\beta_n} \left(   \int_{\mathbb{R}^+}  \left( \left( c - \psi(E)     \right) \beta_n x  + O(\beta_n^2 x^2) \right) f_X(x) dx \right) \notag\\
=&\left( c-  \psi(E)\right)\mathbb{E}[X] +  \lim\limits_{n \rightarrow \infty} O(\beta_n \mathbb{E}[X^2]  ) = \left( c - \psi(E)\right)\mathbb{E}[X],
 \label{cor:temp:A2}
\end{align}
where in the first equality we apply  Taylor's expansion \eqref{cor:temp:taylor} again. By   \eqref{cor:temp:A} and \eqref{cor:temp:A2}, we have
\begin{align}
\lim\limits_{n \rightarrow \infty} \frac{\tilde{A}_n - A_n }{ \beta_n } = \lim\limits_{n \rightarrow \infty} \frac{1}{\beta_n} \left( \frac{ \displaystyle \int_{\mathbb{R}^+} f_{X}(x) e^{- \beta_n \psi_n(E)x} dx - \int_{\mathbb{R}^+} f_{X}(x) e^{-  \beta_n c x} dx }{ \displaystyle \int_{\mathbb{R}^+} f_{X}(x) e^{-  \beta_n c x} dx  \int_{\mathbb{R}^+} f_{X}(x) e^{- \beta_n \psi_n(E)x} dx } \right) =  \left(  c - \psi(E)\right)\mathbb{E}[X].
\label{cor:temp:LHS}
\end{align}
Therefore from \eqref{cor:temp:limintegral4} and \eqref{cor:temp:LHS},
\begin{align}
\lim\limits_{n \rightarrow \infty}  \int_0^x \left(  \tilde{A}_n c   -  A_n   \psi(E) \right)  s  f_X(s) ds =  \left( c - \psi(E)\right) \int_0^x s  f_X(s) ds, \quad \text{ for all }  x \in \mathbb{R}^+.
\label{cor:temp:RHS}
\end{align}

Therefore, we can apply the results of \eqref{cor:temp:LHS} and \eqref{cor:temp:RHS} to Equation \eqref{cor:temp:limintegral4} to get
\begin{align}
 \left(  c - \psi(E)\right) \int_0^x  \mathbb{E}[X]  f_X(s) ds =   \left( c - \psi(E)\right) \int_0^x s  f_X(s) ds, \quad \text{ for all }  x \in \mathbb{R}^+.
 \label{cor:temp:end}
\end{align}
Since $X$ is is not a constant random variable by our assumption, \eqref{cor:temp:end} is only true when $c= \psi(E)$. 
\end{proof}

\subsubsection{Proof of Theorem \ref{thm:limit.smooth} } 
\label{proof:limit.smooth}
\begin{proof}
The proof follows from Lemma \ref{thm:unique2}. By the condition \eqref{cond:limit.smooth2} in Theorem \ref{thm:limit.smooth}, we have 
\begin{align} \label{prof:lemma:unique2}
\log \left( \frac{P\left( Y_n \in I - \beta_n x  \mid X_n = \beta_n x  \right)}{P \left( Z_n \in I \right)} \right) &= \log \left(    \frac{ P\left( Y \in I - \beta_n x  \right)}{P\left( Y\in I\right)} \right) +  g_n(x) \notag \\
& = - \frac{\partial \log P(Y \in [y, y+\delta])}{\partial y}\biggr\rvert_h \beta_n x + O(\beta_n^2 x^2 ) +  \beta_n g_n(x).
\end{align}

We now check whether all conditions in Lemma \ref{thm:unique2} are satisfied:
\begin{enumerate}
\item (Boundedness): 
$
\left|   \frac{f_{X \mid Z_n} (x; I)   }{f_{X}(x) }   \right| = \left| \frac{f_{X \mid \tilde{Z}_n}(x ; E_n) }{ f_X(x)} \right|,
$
which is uniformly bounded on $\mathbb{R}^+$ by the condition \eqref{cond:limit.smooth2} in Theorem \ref{thm:limit.smooth}. And  from \eqref{prof:lemma:unique2}, for any $\xi > 0$,
\begin{align*}
& \left| e^{-\beta_n \xi x} \log \left( \frac{f_{X \mid Z_n} (x; I)   }{f_{X}(x) }  \right) \right| = \left| e^{-\beta_n \xi x} \log \left( \frac{P\left( Y_n \in I - \beta_n x  \mid X_n = \beta_n x  \right)}{P \left( Z_n \in I \right)} \right) \right| \notag \\
& \leq   \left|  e^{-\beta_n \xi x} O(\beta_n x +  \beta_n^2 x^2 ) \right|   +\left|  e^{-\beta_n \xi x}  g_n(x)  \right|,
\end{align*}
where the first term is uniformly bouneded on $\mathbb{R}^+$, and the second term is uniformly bouneded on $\mathbb{R}^+$ by the condition \eqref{cond:limit.smooth3} in Theorem \ref{thm:limit.smooth}.
\item (Linear approximation): Following  \eqref{prof:lemma:unique2}, we have
\begin{align*} 
 \frac{1}{\beta_n} \log \left(   \frac{f_{X \mid Z_n} (x; E)}{f_{X}(x)} \right) &=  \frac{1}{\beta_n} \log \left( \frac{P\left( Y_n \in I - \beta_n x  \mid X_n = \beta_n x  \right)}{P \left( Z_n \in I \right)} \right) \notag \\
 &= - \frac{\partial \log P(Y \in [y, y+\delta])}{\partial y}\biggr\rvert_h  x + O(\beta_n x^2 ) + g_n(x) 
\end{align*}
on $I_n = [0, d_n]$ with $ d_n = O\left(\frac{1}{\beta_n}\right)$. Therefore, we obtain 
$$ c =  \frac{\partial \log P(Y \in [y, y+\delta])}{\partial y}\biggr\rvert_h \quad \text{and} \quad  q_n(x) =  O(\beta_n x^2 ) + g_n(x)$$ 
and we can check that $\mathbb{E}[q_n(X)^2] \rightarrow 0$ since $\mathbb{E}[g_n(X)^2] \rightarrow 0$ by the condition \eqref{cond:limit.smooth3}.
\end{enumerate}
Therefore, applying Lemma \ref{thm:unique2}, we have $$ \displaystyle \lim_{n \rightarrow \infty}\frac{{D_{\textnormal{KL}}\left( \tilde{\mathbb P}^{(n)}_I ~ \|~ {\mathbb Q}^{(n)}_I \right)}}{\beta_n^2}  
 = 0  \quad \text{if and only if} \quad \psi(I) =  \frac{\partial \log P(Y \in [y, y+\delta])}{\partial y}\biggr\rvert_h. $$
Furthermore, since $ 0 < \frac{\partial \log P(Y \in [y, y+\delta])}{\partial y}\biggr\rvert_h < C$ for a constant $C>0$, $ \tilde{\mathbb P}^{(n)}_I $ satisfies the definition of the canonical probability distributions in \eqref{def:cont.can.dis}.

\end{proof}

\subsubsection{Proof of Theorem \ref{thm:limit.ldp} }
\label{proof:limit.ldp}
\begin{proof}
The proof follows from Lemma \ref{thm:unique1}. By the condition \eqref{cond:limit.ldp2} in Theorem \ref{thm:limit.ldp}, we have 
\begin{align} \label{prof:lemma:unique1}
\log \left( \frac{P\left( Y_n \in I - \beta_n x  \mid X_n = \beta_n x  \right)}{P \left( Z_n \in I \right)} \right) &= \log \left(    \frac{  \exp{ \left[ - \frac{1}{\beta_n}\phi \left(y^* - \beta_n x \right) \right] } }{  \exp{ \left[-  \frac{1}{\beta_n}\phi\left(y^* \right)\right] } } \right) + r_n(x) \notag \\
&= \phi'(y^*) x + O(\beta_n x^2) + r_n(x).  
\end{align}
To check that all conditions are satisfied:
\begin{enumerate}
\item (Boundedness): 
$
\left|   \frac{f_{X \mid Z_n} (x; I)   }{f_{X}(x) }   \right| = \left| \frac{f_{X \mid \tilde{Z}_n}(x ; E_n) }{ f_X(x)} \right|,
$
which is uniformly bounded on $\mathbb{R}^+$ by the condition \eqref{cond:limit.ldp1} in Theorem \ref{thm:limit.ldp}. And by \eqref{prof:lemma:unique1}, for any $\xi > 0$,
\begin{align*}
& \left| e^{- \xi x} \log \left( \frac{f_{X \mid Z_n} (x; I)   }{f_{X}(x) }  \right) \right| = \left| e^{- \xi x} \log \left( \frac{P\left( Y_n \in I - \beta_n x  \mid X_n = \beta_n x  \right)}{P \left( Z_n \in I \right)} \right) \right| \notag \\
& \leq   \left|  e^{- \xi x} O(x +  \beta_n x^2 ) \right|   +\left|  e^{- \xi x}  r_n(x)  \right|,
\end{align*}
where the first term is uniformly bounded on $\mathbb{R}^+$, and the second term is uniformly bounded on $\mathbb{R}^+$ by the condition \eqref{cond:limit.ldp3} in Theorem \ref{thm:limit.ldp}.
\item (Linear approximation): As follows from \eqref{prof:lemma:unique1}, we have
\begin{align*} 
  \log \left(   \frac{f_{X \mid Z_n} (x; E)}{f_{X}(x)} \right) &=   \log \left( \frac{P\left( Y_n \in I - \beta_n x  \mid X_n = \beta_n x  \right)}{P \left( Z_n \in I \right)} \right) \notag \\
 &= \phi'(y^*) x + O(\beta_n x^2 ) + r_n(x)
\end{align*}
on $I_n = [0, d_n]$ with $ d_n = O\left(\frac{1}{\beta_n}\right)$. Hence we obtain 
$$ c =  -\phi'(y^*) \quad \text{and} \quad  q_n(x) =  O(\beta_n x^2 ) + r_n(x), $$ 
and we can check that $\mathbb{E}[q_n(X)^2] \rightarrow 0$ since $\mathbb{E}[r_n(X)^2] \rightarrow 0$ by the condition \eqref{cond:limit.ldp3}.
\end{enumerate}
Therefore, by applying Lemma \ref{thm:unique1}, we have 
$$ \displaystyle \lim_{n \rightarrow \infty}{D_{\textnormal{KL}}\left( {\mathbb P}_I ~ \|~ {\mathbb Q}^{(n)}_I \right)}  
 = 0  \quad \text{if and only if} \quad \varphi(I) = -\phi'(y^*).  $$
Since $ 0 < -\phi'(y^*) < \infty$, ${\mathbb P}_I $ satisfies the definition of the canonical probability distributions in \eqref{def:cont.can.dis}.
\end{proof}



\section{Applications}\label{ch:application}

\subsection{Gibbs measure on the phase space} 
\label{sec:gibbs}
\begin{definition}
\label{def:phase.space.map}
Consider a  probability space $(\Omega, \mathcal{F}, \mathbb{P})$, let  $\mathbf{V} = (V_1, V_2, ..., V_n): \Omega \rightarrow \mathbb{R}^n$ be a measurable function   and let $\pi_1, \pi_2$, and $\pi$ be three projection maps defined on $\mathbb R^n $ such that
\begin{align}
\pi_1(\mathbf{V}) &= \mathbf{U} = (V_1, V_2, ..., V_k),\quad 
\pi_2(\mathbf{V}) = \mathbf{W} = (V_{k+1}, V_{k+2} ..., V_n),\quad 
\pi(\mathbf{V}) = \mathbf{V}.
\end{align}
Assume there exist measurable functions $e_1: \mathbb{R}^{k} \rightarrow \mathbb{R}^+$,  $e_2: \mathbb{R}^{n-k} \rightarrow \mathbb{R}^+$, and $e: \mathbb{R}^{n} \rightarrow \mathbb{R}^+$ such that
\begin{align*}
(e_1 \circ \pi_1) (\mathbf{V}) = e_1(\mathbf{U}), \quad 
(e_2 \circ \pi_2) (\mathbf{V}) = e_2(\mathbf{W}), \quad 
(e \circ \pi) (\mathbf{V}) &= e(\mathbf{V}). 
\end{align*}
Therefore, random variables and induced measures can be defined through the following maps:
\begin{align*}
\begin{split}
(\Omega, \mathcal{F}, \mathbb{P}) &\xrightarrow{\mathbf{V}} (\mathbb{R}^n, \mathcal{B}(\mathbb{R}^n), \mu)\xrightarrow{\pi_1} (\mathbb{R}^k, \mathcal{B}(\mathbb{R}^k), \nu_1) \xrightarrow{e_1}  (\mathbb{R}^+,  \mathcal{B}(\mathbb{R}^+), \lambda_1) \\
(\Omega, \mathcal{F}, \mathbb{P}) &\xrightarrow{\mathbf{V}} (\mathbb{R}^n, \mathcal{B}(\mathbb{R}^n), \mu)\xrightarrow{\pi_2} (\mathbb{R}^{n-k}, \mathcal{B}(\mathbb{R}^{n-k}), \nu_2) \xrightarrow{e_2}  (\mathbb{R}^+, \mathcal{B}(\mathbb{R}^+), \lambda_2).
\end{split} 
\end{align*}
\end{definition}

\begin{definition}
\label{def:additive.function}
Let $e_1 \circ \pi_1, e_2 \circ \pi_2 ,$ and $ e  \circ \pi$ be the functions given in Definition \ref{def:phase.space.map}. Define $e_1 \circ \pi_1$ and $e_2 \circ \pi_2$ to be {\em additive} on $\mathbf{V}$ if 
\begin{align}
e_1 \circ \pi_1(\mathbf{V} )+ \  e_2  \circ \pi_2 (\mathbf{V})= e  \circ \pi(\mathbf{V}).
\end{align}
\end{definition}

\begin{theorem}
\label{thm:gibbs}
Suppose $e_1 \circ \pi_1$ and $e_2 \circ \pi_2$ are additive on $\mathbf{V}$ and suppose   $e_1(\mathbf{U}), e_2(\mathbf{W})$ are continuous nonnegative independent random variables.
Denote $ X := e_1(\mathbf{U}), \ Y  := e_2(\mathbf{W}), \ Z :=e(\mathbf{V})$ and let $I = [h, h+\delta]$ be a finite interval. Assume the following conditions hold:
\begin{enumerate}
\item  $\mathbb{E} [X^2] = \epsilon_n^2,$ where $\epsilon_n\to 0 $. 
\item For all $ y\in \mathbb{R}^+,$ there exists a nonnegative integrable function $\Gamma  \in C^2(\mathbb{R}^+)$ such that 
\begin{align}
P \left(Y \in [y, y+ \delta] \right) = \frac{\displaystyle \int_y^{y+\delta} \Gamma(s) ds}{\displaystyle \int_{\mathbb{R}^+} \Gamma(s) ds   } \quad \text{and} \quad  \left|\frac{\partial^2 \log P\big( Y  \in [y, y + \delta ] \big)}{\partial y^2}\right| < \infty.
\label{app:condition}
\end{align}
\item 
$ I \subset supp(\Gamma) $ and $\Gamma'(y) \geq 0$, for $y \in I$. 
\end{enumerate}
Then we have
\begin{align}
\sup_{S \in \mathcal{B}(\mathbb{R}^+)} \left| \mathbb{P} \left(  e_1(\mathbf{U})  \in S \mid Z \in I \right) -  \int_{e_1(\mathbf{u}) \in S  }  A e^{-\psi (I) e_1(\mathbf{u})} \nu_1(d\mathbf{u}) \right| = O(\epsilon_n),  
\label{ex:boltzman.law}
\end{align}
where 
$\psi (I) = \displaystyle \frac{  \displaystyle \partial \log \int_y^{y+\delta} \Gamma(s) ds }{\partial y} \Biggr\rvert_{y=h}.$

\end{theorem}

\begin{proof}
Since the functions $e_1 \circ \pi_1, e_2 \circ \pi_2$ are {\em additive} on $\mathbf{V}$, we have
\begin{align*}
 X + Y &= e_1(\mathbf{U}) + e_2(\mathbf{W})= (e_1 \circ \pi_1) (\mathbf{V}) + (e_2 \circ \pi_2) (\mathbf{V}) = (e \circ \pi) (\mathbf{V}) = e(\mathbf{V})= Z.
\end{align*}
Since $X+Y = Z$ and they are corresponding to $X_n, Y_n, Z_n$ in Theorem \ref{thm:theorem 1}, respectively, it suffices to show that all the conditions in Theorem \ref{thm:theorem 1} are satisfied for $X, Y$, and $Z$.
\begin{enumerate}
\item For all $y \in \mathbb{R}^+$, since $\Gamma(y) \in C^2(\mathbb{R}^+),$ $ \displaystyle \left| \frac{\partial^2  P\left( Y \in [ y, y+\delta]  \right)}{\partial^2 y } \right| $ exists and is bounded on $\mathbb{R}^+$. \\ And $\left| \displaystyle \frac{\partial^2 \log P\left( Y \in [ y, y+\delta]  \right)}{\partial^2 y } \right|$  is bounded on $\mathbb{R}^+$ by \eqref{app:condition}.  Therefore, \eqref{eq:logsquare} holds. 
\item Since $I \subset supp(\Gamma)$, there exists $\delta_1>0$ such that $P\left( Y\in I \right) \geq \delta_1$. And we can derive
\begin{align}
\frac{\partial \log P(Y \in [y, y+\delta])}{\partial y} \biggr\rvert_{y=h} = \frac{\Gamma(h+\delta) - \Gamma(h)}{ \displaystyle \int_h^{h+\delta} \Gamma(s) ds}.
\end{align}
Again, since $I \subset supp(\Gamma)$, and the nonnegative function $\Gamma(y) \in C^2(\mathbb{R}+)$, $\Gamma'(y) \geq 0$, for $y \in I$, we can check that there exists a positive constant $C$ such that
 \begin{align}
 0 \leq \frac{\partial \log P\big(Y \in [y,  y+ \delta] \big)}{ \partial y} \leq C \quad \text{for} \ [y, y+\delta] \subset I,
 \end{align}
hence \eqref{eq:fyx} holds for $D=I$. Furthermore, since $X$ and $Y$ are independent, $b_n = 0$. Therefore, \eqref{eq:fyx2} holds.
\item Since $X$ and $Y$ are supported on $\mathbb{R}^+,$ there exists $\delta_2 > 0$ such that 
$$ P\left( Z \in [z, z+\delta]\right) \geq \delta_2 \quad \text{for} \ [y, y+\delta] \subset \mathbb{R^+}, $$ hence \eqref{eq:condforz} holds.  

\end{enumerate}

Therefore, all of the conditions hold for $D=I$ in Theorem \ref{thm:theorem 1}, we can apply it with $a_n = \epsilon_n^2, \  b_n=0$, and Pinsker's inequality \eqref{eq:totalvariation} to get
\begin{align}
\sup_{S \in \mathcal{B}(\mathbb{R}^+)} \left| \mathbb{P} \left( X  \in S \mid Z \in I \right) - \int_{x \in S }  A e^{-\psi (I) x} f_X(x) dx \right| = O(\epsilon_n),  
\label{ex:boltzmann.law2}
\end{align}
where 
$\psi (I) = \displaystyle \frac{  \displaystyle \partial \log \int_y^{y+\delta} \Gamma(s) ds }{\partial y} \Biggr\rvert_{y=h}.$ Then applying a change of variables
\begin{align}
 \int_{x \in S } A e^{-\psi (I) x}f_{X}(x)dx &=  \int_{x \in S } A e^{-\psi (I) x} \lambda_1( dx)  = \int_{ e_1(\mathbf{u})\in S} A e^{-\psi (I) e_1(\mathbf{u})} \nu_1(d\mathbf{u})
\label{change_of_variable}
\end{align}
to \eqref{ex:boltzmann.law2}, we obtain Equation \eqref{ex:boltzman.law}. It completes the proof.
\end{proof}

In statistical mechanics, the induced measure $\nu_1(d\mathbf{u})$ in phase space is often considered as the Lebesgue measure $d\mathbf{u}$ normalized by the total volume of the phase space $\Lambda$ (Here we assume it is finite). Therefore, for the random vector $\mathbf U$, we have its density 
\begin{align}
 \hat{A} e^{-\psi (I) e_1(\mathbf{u})} \quad \text{with respect to} \ d\mathbf{u},
\label{Boltzmann.law}
\end{align}
where $\hat{A} = A / \Lambda$ is the corresponding normalization factor. 

The assumption $\nu_1(d\mathbf{u}) =  d\mathbf{u} / \Lambda$ for the phase space has already implied that all microstates are {\em equally probable} when the system is unconstrained. It is a reasonable {\em prior} probability for $\mathbf U$ by a symmetry of a physical system when we do not have any previous information about it. For the random variable $X$ (e.g. energy), its density $f_X(x)$ is referred to {\em prior} probability for $X$ when it is unconstrained. Based on the principle of equal a priori probabilities of microstates in the phase space, we can show that $f_X(x) = \gamma(x) / \Lambda$, where $\gamma(x)$ is the Lebseque measure of the surface area of microstates when the energy is fixed on $x$ (i.e. $e_1(\mathbf U) =x$). This can be verified by
\begin{align} 
\int_{ x \in S} f_X(x) dx = \int_{ e_1(\mathbf{u})\in S} \nu_1(d\mathbf{u})= \frac{1}{\Lambda}\int_{ e_1(\mathbf{u})\in S} d\mathbf{u} = \frac{1}{\Lambda} \int_{ x \in S} \gamma(x) dx \quad \text{for all} \ S \in \mathcal{B}(\mathbb{R}^+).
\end{align}
Note that $\gamma(x)$ is also known as the {\em structure function of $X$}. In Theorem \ref{thm:gibbs}, we also make the same assumption for $Y$: $f_Y(y) = \Gamma(y) / \Lambda$, where $\Gamma(y)$ is the {\em structure function of $Y$}.

Therefore, the density of $X$ can be written as
\begin{align}
 \hat{A} e^{-\psi (I) x} \gamma(x) \quad \text{with respect to} \ dx,
 \label{ex:boltzmann.law3}
\end{align}
which can be interpreted as a uniform prior biased by an exponential weight $e^{-\psi (I) x}$ when the system is conditioned on some extra information.  Note that Equation \eqref{Boltzmann.law} is known as the density of  Gibbs measure on the phase space and Equation \eqref{ex:boltzmann.law3} is known as the density of Gibbs measure on the energy of the system \cite{georgii2011gibbs}.

In the work of A. Ya. Khinchin \cite{aleksandr1949mathematical}, he assumed 
{\em conjugate distribution laws} for all systems. It is said that 
\begin{align}
f_X(x) = \frac{e^{- \alpha x} \gamma(x)  }{\int e^{- \alpha s} \gamma(s)  ds } \quad \text{and} \quad f_Y(y) = \frac{e^{- \alpha y} \Gamma(y)  }{\int e^{- \alpha s} \Gamma(s)  ds }
\label{ex:khinchin}
\end{align}
for some constant $\alpha$. Those priors are more general than the uniform prior and they have some nice properties, e.g., for a proper $\alpha$, it may guarantee integrability of $ e^{- \alpha s} \gamma(s)$ when $ \gamma(s)$ itself is not integrable. However, we can show that the choice of $e^{- \alpha x}$ term does not have influence on our results. Here is the proof sketch: Suppose $\delta = o(1)$,
\begin{align}
\hat{\psi}(I) : = \displaystyle \frac{  \displaystyle \partial \log P(Y  \in [y, y+\delta]) }{\partial y} \Biggr\rvert_{y=h}  &= \displaystyle \frac{  \displaystyle \partial \log \int_y^{y+\delta} \Gamma(s) e^{-\alpha s} ds }{\partial y} \Biggr\rvert_{y=h} \notag \\
&\approx  \displaystyle \frac{  \displaystyle \partial \log \int_y^{y+\delta} \Gamma(s) ds }{\partial y} \Biggr\rvert_{y=h} 
-  \alpha  = \psi(I) - \alpha.
\label{ex:khinchin2}
\end{align}
By \eqref{ex:khinchin} and \eqref{ex:khinchin2}, we have 
\begin{align}
A f_X(x) e^ {-\hat{\psi}(I) x} = \hat{A} \gamma(x)  e^ {-
\alpha  x} e^ {-(\psi(I) - \alpha)x} =  \hat{A} \gamma(x)   e^ {-{\psi(I)}x}.
\end{align}

Therefore, to choose priors as the structure functions multiplied by the exponential functions $e^{- \alpha x}$ for integrability is irrelevant to Gibbs measure. Indeed, it is the extra information (condition) giving rise to the exponential weight in Gibbs measure and the parameter of the exponential function is determined by
$$\psi (I) = \displaystyle \frac{  \displaystyle \partial \log \int_y^{y+\delta} \Gamma(s) ds }{\partial y} \Biggr\rvert_{y=h},$$
in which $\int_y^{y+\delta} \Gamma(s) ds $ represents the volume of microstates in the shell between $y$ and $y +\delta$. The logarithm of it is known as the entropy of $Y$, denoted by $S_Y(y)$, hence we have
\begin{align}
\psi (I)  = \frac{ \partial S_Y(y)}{ \partial y}\Biggr\rvert_{y=h}.
\label{ex:gibbs.temperature}
\end{align}
By Equation \eqref{ex:gibbs.temperature}, we can identify $  \frac{1}{\psi(I)}$ as the temperature defined in statistical mechanics \cite{huang1975statistical,landau1958statistical}.

\begin{remark}
We can extend Theorem \ref{thm:gibbs} to the model that the subsystem and its heat bath have strong interaction defined by  non-additivity of energy functions in statistical mechanics. 
Assume there exists a measurable function $e_3: \mathbb R^n \rightarrow \mathbb R^+$ such that
\begin{align*}
    (e_3 \circ \pi)(\mathbf{V}) = e_3(\mathbf{V}),
\end{align*}
which means that this energy function $e_3$ could depend on the whole vector $\mathbf{V} = (V_1, V_2, ..., V_n)$ in the phase space.
And suppose 
$$e_1 \circ \pi_1(\mathbf{V}) + e_2 \circ \pi_2(\mathbf{V}) + e_3 \circ \pi(\mathbf{V}) = e \circ \pi(\mathbf{V}),$$ 
in which the existence of the extra term $e_3 \circ \pi(\mathbf{V})$ means that $e_1 \circ \pi_1$ and $e_2 \circ \pi_2$ are not {\em additive} on $\mathbf{V}$ by Definition \ref{def:phase.space.map}. 
Denote that $ R := e_3(\mathbf{V})$. Recall that $\mathbf{V} = (\mathbf{U}, \mathbf{W})$ and
$ X = e_1(\mathbf{U}) = (e_1 \circ \pi_1)(\mathbf{V}), \ Y  = e_2(\mathbf{W}) = (e_1 \circ \pi_2)(\mathbf{V}), \ Z =e(\mathbf{V}) =(e \circ \pi)(\mathbf{V})$. Then we have 
\begin{align}
X + Y + R = Z.
\end{align}
In statistical mechanics, $R$ is known as the {\em interaction energy} caused by interaction between the subsystem and its heat bath. Based on this setup, we can define a new random variable $\hat{Y} := Y + R$, but $X,\hat{Y} $ are no longer independent since the random variable $R$ may depend on both $\mathbf{U}, \mathbf{W}$ in the phase space. If we modify the condition \eqref{app:condition} in Theorem \ref{thm:gibbs} to guarantee the existence and boundedness of
	\begin{align} 
	\left| \partial^{(k)}  P\big( \hat{Y}  \in [y, y + \delta ] \mid X = x \big) \right| \ \text{and} \  \left|\partial^{(k)} \log P\big( \hat{Y}  \in [y, y + \delta ] \mid X = x \big)\right|, \ \text{for} \ k=0,1,2,
	\end{align} 
in which the partial derivatives are with respect to both $x$ and $y$,
then we are able to apply Corollary \ref{cor:corollary 1} to this model. As the result \eqref{cor:parm.exp2} in Corollary \ref{cor:corollary 1}, this model with strong interaction will give rise to a new parameter $\phi(I)$ of the exponential weight which involves two terms: one is from  fluctuations of the energy of the ``new" heat bath $\hat{Y}$ (it combines  the energy of the heat bath $Y$ without interaction and the interaction energy  $R$);
and the other one is from the correlation of $X$ and $\hat{Y}$.
\end{remark}

\subsection{Integer-valued random variables and conditional Poisson distributions}
\label{sec:poisson}
In the following Theorem \ref{thm:d+d:iid}, we will show a limiting behavior of a sequence of conditional probabilities for a nonnegative integer-valued random variables $K$, which is conditioned on $K + \tilde{L}_n$, $\tilde{L}_n$ is a sequence of sums of i.i.d random variables $\xi_i$. This sequence of conditional probabilities has the same limiting behavior as its unconditional probability $P(K=k)$ weighted by an exponential factor. The most important result of this theorem is that the parameter of this exponential factor determined by a normal distribution rather than the distribution of $\xi_i$. By this result, we provide a very simple formula with an approximation error to approximate an intractable problem in calculating the conditional probability of an integer-valued random variable. And we give an example \ref{ex:poisson} to show an approximation formula for calculating the conditional probability of a Poisson random variable conditioned on the sum of that Poisson random variable with another Poisson random variable.

\begin{theorem}
\label{thm:d+d:iid}
Let $K$ be a nonnegative integer-valued random variable with $\mathbb{E} [K] < \infty $. Let $\tilde L_n = \sum_{i=1}^n \xi_i,$ where $\{\xi_i\}_{i=1}^n$ are nonnegative  i.i.d. random variables. $K$ and $\tilde L_n$ are independent and denote $\tilde H_n := K+ \tilde L_n $. 
Let $\mu = \mathbb{E} [\xi_i]$,  $\sigma^2=\textnormal{Var}(\xi_i)$ and assume $\mathbb{E}[(\xi_i - \mu)^3] < \infty$. And let 
$$ B_n = \sum_{k=0}^{\infty} \frac{1}{P(K= k) \exp{ \left( \frac{ -\psi(I) k }{\sqrt{n}} \right) }}, \quad \psi(I) = \frac{\partial \log P \big( Y \in [y, y + \delta] \big)}{\partial y} \biggr\rvert_{y=-h}, \quad \text{and} \  Y \sim N(0, \sigma^2).$$
For every fixed finite interval $I = [-h, -h + \delta], \ h, \delta \in \mathbb{R}^+$, $-h + \delta \leq 0$, and $2 \delta / \sigma^2 < \psi(I)$,
\begin{align} \label{integer.rate.of.conv}
\sup_k \bigg\lvert  P(K = k \mid  \tilde H_n \in n\mu + \sqrt{n} I  ) - B_n P(K = k) \exp{ \left( \frac{ -\psi(I) k }{\sqrt{n}} \right) }  \bigg\rvert = O(\frac{1}{\sqrt{n}}).
\end{align}

\end{theorem}

\begin{proof}
Let  $K_n := \frac{K}{\sqrt{n}}$, $L_n :=\frac{\tilde L_n - n\mu}{  \sqrt{n}}$ and $ H_n: =\frac{\tilde H_n - n\mu}{\sqrt{n}}$. We have $
K_n + L_n = H_n.
$
By the Central Limit Theorem,  $L_n$ converges in distribution to $Y$. Furthermore, since $(\xi_i-\mu)$ has finite second and third moments, by Berry-Esseen Theorem \ref{thm:berry-esseen}, 
\begin{align}\label{eq:BEbound}
    \sup_{k}\left|P_{L_n}\left( I - \frac{k}{\sqrt{n}} \right)-P_{Y}\left( I - \frac{k}{\sqrt{n}}\right)\right|=O\left( \frac{1}{\sqrt n}\right).
\end{align}
Since $\mathbb E \left[K_n \right]\to 0$, we have $K_n$ converges to $0$ in probability.   By Slutsky's Theorem \ref{thm:slutsky}, $H_n$ converges to $Y$ in distribution. By Corollary \ref{cor:berry-esseen-slusky}, we can also get

\begin{align} \label{eq:Slutsky}
   P_{H_n}\left( I \right)=  P_{Y}\left(I \right) + O\left(\frac{1}{\sqrt{n}}\right).
\end{align}
By  \eqref{eq:BEbound} and \eqref{eq:Slutsky},
\begin{align}
P_{K \mid \tilde H_n} \left(k ; n\mu + \sqrt{n}I \right)  &= P_{K \mid H_n } \left( k ; I \right) 
= P_K(k) \frac{P_{L_n}\left( I - \frac{k}{\sqrt{n}}\right) }{    P_{H_n}\left( I \right)  } = P_K(k) \frac{P_{Y}\left( I - \frac{k}{\sqrt{n}}\right)+ O\left(\frac{1}{\sqrt{n}}\right)    }{   P_{Y}\left( I\right)  + O\left(\frac{1}{\sqrt{n}}\right)  } \notag\\
&= P_K(k) \frac{P_{Y}\left( I - \frac{k}{\sqrt{n}}\right) }{  P_{Y}\left(I \right)  } + O\left(\frac{1}{\sqrt{n}}\right),
\label{dis.ex.approx.1}
\end{align}
in which we use the fact $Y \sim N(0, \sigma^2)$ and $ P(-  h  \leq  Y \leq  -h + \delta   )$ is bounded from below.  Moreover, since $P_K(k)\leq 1$, the  term $O\left(\frac{1}{\sqrt{n}}\right)$ in \eqref{dis.ex.approx.1} is independent of $k$. Let $\tilde{Y}_n \sim N( n \mu , n \sigma^2)$ and $ \tilde Z_n := K + \tilde Y_n$. Then we have  
\begin{align}\label{eq:YW}
K_n + Y_n  = Z_n, \quad\text{where} \quad Y_n:=\frac{\tilde{Y}_n - n\mu} { \sqrt{n}} \  \text{and} \  Z_n:= \frac{\tilde{Z}_n - n\mu}{\sqrt{n}}.
\end{align}
Note that ${Y}_n = Y \sim N(0, \sigma^2)$ and $Z_n$ converges in distribution to $Y$. Similar to \eqref{dis.ex.approx.1},
\begin{align}
& P_{K \mid \tilde{Z}_n} \big( k ; n\mu + \sqrt{n}I \big) 
=  P_K(k) \frac{P_{Y}\left( I - \frac{k}{\sqrt{n}}\right) }{  P_{Y}\left( I\right)  } + O\left(\frac{1}{\sqrt{n}}\right).
\label{dis.ex.approx.2}
\end{align}
Applying the triangle inequality to \eqref{dis.ex.approx.1} and \eqref{dis.ex.approx.2}, we finally obtain
\begin{align} \label{integer.tria.1}
\sup_{k}\biggr\lvert  P_{K \mid  \tilde H_n} \left( k ;  n\mu + \sqrt{n}I \right) - P_{K \mid  \tilde{Z}_n} \left(k ; n\mu + \sqrt{n} I \right) \biggr\rvert = O(\frac{1}{\sqrt{n}}).
\end{align}
Now, it remains to show that the convergence rate of
\begin{align}\label{eq:547}
 \sup_k \biggr\lvert P_{K \mid  \tilde{Z}_n} \left(k ; n\mu + \sqrt{n} I \right) -  B_n P_K(k) \exp{ \left( \frac{ -\psi(I) k }{\sqrt{n}} \right) }  \biggr\rvert.
\end{align} 
Then it suffices to show that all the conditions in Theorem \ref{thm:theorem 1}   are satisfied for $K_n, Y_n, Z_n$, then we can apply Theorem \ref{thm:d+d}.

First, we can check that $\mathbb{E}[K^2_n] = a_n, \ a_n = o(1)$:
\begin{align}
	 \mathbb{E}[K_n^2]& = \frac{1}{n} \mathbb{E}[K^2] = O\left(\frac{1}{n} \right).
\end{align}	 

Second, by change of variables, 
\begin{align}
 P_{K \mid  \tilde{H}_n} \left(k ; n\mu + \sqrt{n}I \right)  = P_{K_n \mid H_n} \left(  \frac{k}{\sqrt n} ; I\right).
\end{align}
And we can define the set $S$ in terms of the value for $K$ as below:
\begin{align*}
    S =\left\{k : k \in\mathbb N, \ \mathbb P(K=k)> 0 \right\}
\end{align*}
such that for all $k \in S$, $P(K_n = \frac{k}{\sqrt n}  ) >0$.
Choose $d> 0 $ such that  $I =  [-h, -h+\delta] \subseteq D = (-d,0) $. Below we follow every steps in Theorem \ref{thm:theorem 1} with slight modifications:
\begin{enumerate}
\item   For all $y\in  \mathbb{R} $,
	 ${Y}_n = Y \sim N(0, \sigma^2)$, by the formula of the density of normal distribution, we have
	\begin{align} \label{normal.dis.bound1}
	    \frac{\partial^{2} P\big({{Y}} \in [y, y+\delta]\big)}{ \partial y^2} = f'_Y(y+\delta) - f'_Y(y)  
	\end{align} 
	and 
	\begin{align} \label{normal.dis.bound2}
	     \frac{\partial^{2} \log   P\big( {{Y}} \in [y, y+\delta]\big) }{ \partial y^2} = \frac{f'_Y(y+\delta) - f'_Y(y) }{P\big( {{Y}} \in [y, y+\delta]\big)} - \left( \frac{f_Y(y+\delta) - f_Y(y) }{P\big( {{Y}} \in [y, y+\delta]\big)}\right)^2,
	\end{align}
	so we can check \eqref{normal.dis.bound1} exist and are uniformly bounded. For \eqref{normal.dis.bound2}, we modify the boundedness slightly and the details of proof are provided in Appendix \ref{appendix:prof.normal}.
	Therefore, \eqref{eq:logsquare} with a slight modification holds. 
	
	\item  Since  ${Y}_n = Y \sim N(0, \sigma^2)$, there exist positive constants $\delta_1$ and $C$ depending on $y$ such that $P\big( {{Y}} \in [y, y+\delta]\big)\geq \delta_1$ and  $0 \leq \displaystyle\frac{\partial \log P\big( {{Y}} \in [y, y+\delta]\big)}{ \partial y} \leq C $ for every $[y, y+\delta] \subset D$. Therefore \eqref{eq:fyx} holds. Since $K_n $ and ${Y}_n $ are independent, we have $b_n=0$. Therefore \eqref{eq:fyx2} holds.
    
    	\item  Since $Z_n \rightarrow Y$ in distribution where $Y \sim N(0, \sigma^2)$, there exists $\epsilon_n(z) \rightarrow 0 $ such that \[P \big( {Z_n} \in [z, z+\delta]\big) = P \big( Y \in [z, z+\delta]\big) + \epsilon_n(z).\] 
    	Since $P \big( Y \in  [z, z+\delta]\big)$ is bounded from below for $[z, z+\delta] \subset D$, there exists a positive constant $\delta_2(z)$ such that  $P\big(Z_n \in [z, z+\delta]\big) \geq \delta_2>0$ for all $[z, z+\delta] \subset D$. Then the second inequality in \eqref{eq:condforz} holds. 

 \label{condition5}
\end{enumerate}
To apply Theorem \ref{thm:d+d}, we then obtain
\begin{align}
\sup_{k \in S} \biggr\lvert P_{K_n \mid {Z_n} } \left( \frac{k}{\sqrt n} ; I \right) - B_n P_{K_n}\left(\frac{k}{\sqrt n} \right) \exp{\left(-\psi(I)\frac{k}{\sqrt n}\right)}    \biggr \rvert = O(\frac{1}{n}),
\end{align}
where 
\begin{align}
\psi(I) = \frac{\partial \log P_Y \big( [y, y+\delta]\big)}{\partial y} \biggr\lvert_{y=-h} \quad \text{and} \quad Y \sim N(0, \sigma^2).
\end{align}
By change of variable, we then obtain
\begin{align} \label{integer.tria.2}
\sup_{k} \biggr\lvert P_{K \mid \tilde{Z}_n }\big(k ; n\mu + \sqrt{n}I \big) - B_nP_{K}( k ) \exp{\left( \frac{-\psi(I) k}{\sqrt{n}} \right)}    \biggr \rvert = O(\frac{1}{n}),
\end{align}
where \[B_n = \frac{1}{\sum_{k\in S} P_{K_n}( k/\sqrt{n} ) \exp{(-\psi(I)k/\sqrt{n})}} =\frac{1}{ \sum_k P_K( k ) \exp{(-\psi(I)k/\sqrt{n})}}.\]
By applying triangle inequality to \eqref{integer.tria.1} and \eqref{integer.tria.2}, we can obtain \eqref{integer.rate.of.conv} in the theorem.

\end{proof}

Finally we apply Theorem \ref{thm:d+d:iid} to a concrete example.

\begin{example}
\label{ex:poisson}
Let $\lambda,\mu>0$ be two constants. Consider two independent random variables $K \sim \text{Pois} (\lambda)$ and $\tilde L_n\sim \text{Pois} (n\mu)$.   Let $ \tilde H_n := K + \tilde L_n$. For every fixed finite interval $I$ which follows from Theorem \ref{thm:d+d:iid}, we can show that
 \begin{align*}
\sup_k \bigg\lvert  P \left( K = k \mid  \tilde H_n \in n\mu + \sqrt{n} I \right) - B_n P(K= k) \exp{ \left( \frac{ -\psi(I) k }{\sqrt{n}} \right) }  \bigg\rvert = O(\frac{1}{\sqrt{n}}),
\end{align*}
where  $\displaystyle
B_n = \sum_{k=0}^{\infty} \frac{1}{P(K= k) \exp{ \left( \frac{ -\psi(I) k }{\sqrt{n}} \right) }}
$\quad and \quad
$\displaystyle
\psi(I) = \frac{\partial \log P \big( Y \in [y, y+\delta] \big)}{\partial y} \biggr\rvert_{y=-h}   , \quad  Y \sim N(0, \mu).
$
\end{example}

\begin{proof}
By the property of Poisson random variables, we can decompose $\tilde L_n$ as  $\tilde L_n=\sum_{i=1}^n \xi_i$, where $\{\xi_i,1\leq n\}$ are independent Poisson random variables with mean $\mu$ and variance $\mu$.  We can check that all conditions are satisfies in Theorem \ref{thm:d+d:iid}. Hence Theorem \ref{thm:d+d:iid} can be applied.
\end{proof}

\subsection{Emergence of temperature (conditioned on the scale of large deviations)}
\label{sec:em.temp.ldp}

In this section, we define the parameter $ \frac{1}{\varphi(I)} $ in the exponential function $e^{- \varphi(I) x }$ as the {\em temperature} of the canonical distribution. Consider a sequence of conditional probabilities for a function of a subsystem represented by $X$ in contact with its heat bath represented by $\tilde{Y_n} = \sum_{i=2}^{n} X_i$, where $X_i$ are i.i.d. and $X_i$ has the same distribution as $X$, and $X_i$, $X$ are independent. Suppose that the total energy  $\tilde Z_n = X + \tilde{Y_n}$ is conditioned on the scale of large deviations from its mean, we will show that the {\em temperature}  $ \frac{1}{\varphi(I)} $ is an emergent parameter uniquely determined by the rate function of $\frac{\tilde{Y}_n}{n} $.

\begin{definition}
\label{def:sumiid.ldp}
Let $X$ be a nonnegative and nonconstant continuous random variable with $\mathbb{E}[X^4] < \infty,$ and let $\tilde Y_n := \sum_{i=2}^{n} X_i,$  where all random variables in $\{X_i\}_{i=2}^{n}\cup \{X\}$ are i.i.d.. Denote $\tilde Z_n := X + \tilde Y_n$.  Consider an interval $ I = [d, d+\delta], \ d \in \mathbb{R}, \ \delta>0$ with $\mathbb E [X] \notin I$, and a function $\varphi : I \rightarrow \mathbb{R}$ such that $0<\varphi(I)< \infty$.  Let $\mathbb P_{I}$ be a probability measure with density function $\displaystyle A f_{X}(x) e^{-\varphi(I)x} $, where \[\displaystyle \frac{1}{A}=  \int_{\mathbb{R}^+} f_{X}(x) e^{-\varphi(I)x}dx. \] 
Let $\mathbb Q^{(n)}_I$ be a sequence of probability measures with density functions $f_{X| \tilde Z_n}(x; nI)$. 
\end{definition}

\begin{theorem}
\label{thm:iidsum}
Denote ${Y}_n := \frac{\tilde Y_n}{n},  {X}_n := \frac{X}{n}, $ and ${Z}_n := {X}_n + {Y}_n$, and let $I - \frac{x}{n} = \{ y -  \frac{x}{n} , \  y \in I \}.$  Assume the following conditions hold: 
\begin{enumerate}
\item \label{cond:sumiid1} $\displaystyle \left| \frac{f_{X \mid Z_n}\left(x;  I  \right)  }{ f_{X}(x)} \right| $ is uniformly bounded on $ \mathbb{R}^+$.
\item \label{cond:sumiid2} $\left| \log P_{{Y}_n}(I  )   - \log P_{{Z}_n }(I ) \right|  \ \text{converges to a finite constant as } n\to\infty.$
\item \label{cond:sumiid3} There exists a function $\phi(y) \in C^2(D),$ where $D$ is an open interval containing $I$, with $ -\infty < \phi'(y) < 0, \ \text{for} \ y \in I$, such that 
\begin{align}
\log P_{{Y}_n}\left(I- \frac{x}{n} \right) = - n \phi\left(y^* -  \frac{x}{n}\right) + s_n\left(I-  \frac{x}{n}\right), \quad  \text{for} \ I-\frac{x}{n} \subset D,
\end{align}
where $\displaystyle y^* =  \{ y: \inf_{y \in I } \phi(y) \},$ $\displaystyle \left| \frac{s_{n}( I - \frac{x}{n}) -  s_n(I)}{ s_n(I)}  \right|  = O \left( \frac{x}{n} \right),$ and $ \left| s_n(I')\right|  = o(n)$ for all $I' \subset D$.
\end{enumerate}
Then 
\begin{align}
D_{\textnormal{KL}}\left(\mathbb P_I~ \|~ \mathbb Q^{(n)}_I \right)
  \rightarrow 0 \quad \quad \text{if and only if} \quad \quad  \varphi(I) = -  \phi'(y^*),  \quad  y^* =  \{ y: \inf_{y \in I } \phi(y) \}.
 \label{ifonlyif:lpd}
\end{align}

\end{theorem}

\begin{remark} \label{rmk:int.app.ldp}
 The conditions \eqref{cond:sumiid1} - \eqref{cond:sumiid3} formulated in Theorem \ref{thm:iidsum} are technical, so we would like to characterise and verbally describe the underlying meaning and interpretation of them: The condition \eqref{cond:sumiid1} can be written as 
 $$\left| \frac{f_{X \mid Z_n}\left(x;  I  \right)  }{ f_{X}(x)} \right| = \left| \frac{f_{X_n , Y_n}\left(\frac{x}{n}, I - \frac{x}{n}   \right)  }{ f_{X_n}(\frac{x}{n}) f_{Y_n}(I-\frac{x}{n})} \right| \quad \text{is uniformly bounded on} \ \mathbb R^+, $$
 in which the right hand side is related to the correlation of $X_n$ and $Y_n$, therefore, this condition means that the interaction between $X_n$ and $Y_n$ is regulated; The condition \eqref{cond:sumiid2} is corresponding to the setup that $X_n$ is small relative to $Z_n$ (hence the distributions of $Y_n$ and $Z_n$ have the same asymptotic behavior), specifically, that finite constant can be chosen to be zero (we provide a more general condition in this theorem); The condition \eqref{cond:sumiid3} means that $Y_n$ converges to a constant satisfying the large deviation principle with the rate function $\phi$ and the remainder term $s_n$. 
\end{remark}

\begin{proof}
The proof of Theorem \ref{thm:iidsum} is just the application of Theorem \ref{thm:limit.ldp}, so we will show that all conditions in Theorem \ref{thm:limit.ldp} are satisfied. 
First, Condition \eqref{cond:limit.ldp1} in Theorem \ref{thm:limit.ldp} follows from Condition \eqref{cond:sumiid1}, and $\mathbb{E}[X^4] < \infty$ is assumed in this theorem. Second, Condition \eqref{cond:limit.ldp2} in Theorem \ref{thm:limit.ldp} follows from (i) 
$ Y_n \rightarrow \mathbb E[X]  \ \text{in probability}$ by the law of large numbers, (ii) $\mathbb E [X] \notin I$ by Definition \ref{def:sumiid.ldp}, and (iii) the
Condition \eqref{cond:sumiid3} in this theorem.

Third, since $I$ is closed and contained in an open interval $D$, there exists a constant $d \in \mathbb{R}^+$ such that $\displaystyle  I -\frac{x}{n} \subset D$ for $x \in [0, nd]$. Therefore,   by Condition \eqref{cond:sumiid3},
\begin{align}
\log P_{{Y}_n}\left(I - \frac{x}{n}\right) = - n  \phi\left(y^* - \frac{x}{n}\right) + s_n\left(I - \frac{x}{n} \right), \quad  y^* =  \left\{ y: \inf_{y \in I} \phi(y) \right\}.
\label{ldr:1}
\end{align}
Since $  \left[ y^*,  y^*- \frac{x}{n} \right]  \subseteq  D $ and $\phi \in C^2(D)$, by Taylor's expansion,
\begin{align}
\phi\left( y^*-\frac{x}{n} \right) = \phi(y^*) - \phi'(y^*) \frac{x}{n} + O\left(\frac{x^2}{n^2}\right) \quad \text{for all} \ x \in [0, nd].
\label{ldr:2}
\end{align}
By Condition \eqref{cond:sumiid2} and \eqref{cond:sumiid3}, there exists a sequence $\epsilon_n \rightarrow 0$ and a constant $k$ such that
\begin{align}
\log P_{{Z}_{n}}(I ) = \log P_{{Y}_{n}}(I ) + k + \epsilon_n = - n\phi(y^*) + s_n(I) + k + \epsilon_n.
\label{ldr:3}
\end{align}
By Condition \eqref{cond:sumiid3}, we have 
\begin{align}
\left|    s_n\left(I - \frac{x}{n}\right) -  s_{n}(I)  \right| = \left| s_{n}(I) \right|  O\left(\frac{x}{n}\right) = O(\delta_n x),
\label{ldr:3.5}
\end{align}
in which $\delta_n \rightarrow 0$.
By the results of \eqref{ldr:1}, \eqref{ldr:2}, \eqref{ldr:3}, and \eqref{ldr:3.5}, we obtain
\begin{align}
\log \left(\frac{P_{{Y}_n} \left(  I - \frac{x}{n} \right)  }{P_{{Z}_n}\left( I \right) }  \right)
&= \log \left( \frac{\exp{ \left[- n  \phi\left(y^* - \frac{x}{n}\right) \right]  }}{\exp{ \left[ - n\phi(y^*) \right]  }}   \right) + O\left(\frac{x^2}{n}\right) + O(\delta_n x) + \epsilon_n \quad \text{on} \quad I_n=[0, nd].
\label{ldr:5}
\end{align}
Let  $r_n(x) :=  O\left(\frac{x^2}{n}\right) + O(\delta_n x) + \epsilon_n$, we can check that (i)  $\left| r_n(x)e^{- \xi x} \right|$  uniformly bounded on  $\mathbb{R}^+$   for any   $\xi>0$, and (ii)  $\mathbb{E} \left[ r_n(X)^2 \right] \rightarrow 0$ since $\mathbb{E}\left[ X^4 \right] < \infty$ by Definition \ref{def:sumiid.ldp}. Hence, $ r_n(x), d_n , \phi$ satisfy Condition \eqref{cond:limit.ldp3} in Theorem \ref{thm:limit.ldp}. Therefore, we have checked that all of the conditions in Theorem \ref{thm:limit.ldp} hold, then we can apply it to get
\begin{align}
D_{\textnormal{KL}}\left(\mathbb P_I ~ \|~ \mathbb Q^{(n)}_I \right)
  \rightarrow 0  \quad \text{if and only if} \quad  \varphi(I) =  - \phi'(y^*).
\end{align}
\end{proof}

By Cram\'er's Theorem \ref{thm:cramer}, the existence of the function $\phi(y)$ in Condition \eqref{cond:sumiid3} is from the existence of the rate function of ${Y}_n = \sum_{i=1}^{n-1} {X}_i / n$. Let set $ D_\phi := \{ y \in \mathbb{R} : \phi(y) < \infty    \}$ and we can choose  $D = \text{int} \left( D_\phi \right) $. By the properties of rate functions in Appendix \ref{prop.of.ratefunc}, we have
\begin{align}
\phi(y) \in C^2(D) \quad , \quad \phi(y) \ \text{is convex on} \ D,
\end{align}
and $-\infty < \phi'(y) < 0$ for $y \in I \subset D$ if the interval $I$ is chosen on the left side of the mean of $Y_n$. By Cram\'er's Theorem, the rate function satisfies
\begin{align}
\log P_{{Y}_n}\left(I- \frac{x}{n} \right) = - n \phi\left(y^* -  \frac{x}{n}\right) + o(n), \quad  \text{for} \ I-\frac{x}{n} \subset D.
\label{cond:ratefunction}
\end{align}
Comparing \eqref{cond:ratefunction} with Condition \eqref{cond:sumiid3}, Theorem \ref{thm:iidsum}  requires an explicit form of the remainder:
\begin{align}
\log P_{{Y}_n}\left(I- \frac{x}{n} \right) = - n \phi \left(y^* -  \frac{x}{n}\right) + s_n \left(I-  \frac{x}{n}\right), \quad  \text{for} \ I-\frac{x}{n} \subset D,
\end{align}
where $\displaystyle \left| \frac{s_{n}( I - \frac{x}{n}) -  s_n(I)}{ s_n(I)}  \right|  = O \left( \frac{x}{n} \right),$ and $ \left| s_n(I')\right|  = o(n)$ for all $I' \subset D$. This stronger condition guarantees the ``if and only" if statement \eqref{ifonlyif:lpd}. 

The following is our discussion on the connection between Theorem \ref{thm:iidsum} and Van Campenhout and Cover's Theorem \ref{thm:von.1981}. In Theorem \ref{thm:iidsum}, if the condition is on the scale of large deviations, then the conditional density $$f_{X \mid \tilde{Z}_n}(x; nI), \ n\mu \not\in nI $$  can be approximated by the (normalized) product of its unconditional density $f_X(x)$ and an exponential function $e^{-\lambda x}$. This parameter $\lambda = \phi'(y^*)$ is unique and determined by the first derivative of the rate function evaluated at $y^*= \inf_{y \in I}\phi(y) $. It implies that we are able to find $\lambda$ directly from the rate function without using the maximum entropy principle. Furthermore, by the pair of reciprocal equations \eqref{pre:legendre}:
\begin{align}
 \phi'(y^*) = \lambda \quad \text{if and only if} \quad A'(\lambda) = y^*,
 \label{legendre2}
\end{align}
which means the parameter $\lambda$ we find by the derivative of the rate function (left side of \eqref{legendre2}) is also the solution of the derivative of the free energy function $A$ under the constraint $= y^*$ (right side of \eqref{legendre2}). 

Therefore, using the maximum entropy principle under the first moment constraint to find good approximations of conditional density (Van Campenhout and Cover's approach) is a natural consequence of the emergent behavior of
\begin{align}
\log\left( \frac{f_{X \mid \tilde{Z}_n}(x ; nI) }{ f_X(x)} \right).
\end{align} 
And this emergent behavior gives rise to a large deviation function that uniquely determines the parameter of the exponential weight. As we discussed in the Section \ref{ch:preliminary}, we apply the large deviation principle directly to the distribution of a the heat bath
$$ Y_n = \frac{\tilde{Y}_n}{n} =  \frac{1}{n} \sum_{i=2}^{n} X_i. $$ 
On the other hand, the Gibbs conditioning principle uses the large deviation principle for emprical measures 
$$ L_n = \frac{1}{n} \sum_{i=1}^n \delta_{X_i}.$$
Denote that $X_1 :=X$. Then the limit problem of the sequence of probability measures $\mathbb{Q}^{(n)}_I$ with density functions 
$$     f_{X | Z_n}(x; I),  \quad \text{where} \  Z_n = X + Y_n =  \frac{1}{n} \sum_{i=1}^n X_i,$$
and the limit problem of the sequence of emprical measures 
$$ \mathbb{E} \left[  L_n \mid L_n \in \Gamma \right], \quad \text{where} \  L_n = \frac{1}{n} \sum_{i=1}^n \delta_{X_i} \ \text{and} \  \Gamma = \left\{ \gamma : \int x \gamma(dx) \in I  \right\}$$
are just two sides of the same coin. Eventually, they both give arise to a limit as a canonical distirbution with the density
$$ f_X(x) e^{-\lambda x}.$$
In conclusion, our approach generates $\lambda$ by the large deviation rate function of the heat bath $Y_n$  and the Gibbs conditioning principle solves $\lambda $ by minimizing the relative entropy which is the large deviation rate function of sampling. These two approaches are connected by the reciprocal equations \eqref{legendre2} through the Legendre transform.

\subsection{Emergence of temperature (conditioned on the scale of Gaussian fluctuations)}
\label{sec:em.temp.clt}
Similar to Section \ref{sec:em.temp.ldp}, in this section, we define the parameter $\frac{1}{\beta_n \psi(I)} $ in the exponential function $e^{- \beta_n \psi(I) x }$ as the {\em temperature} of the canonical distribution and consider a sequence of conditional probabilities for a function of a subsystem represented by $X$ in contact with its heat bath represented by $\tilde{Y_n} = \sum_{i=2}^{n} X_i,$  $X_i$ are i.i.d. and $X_i$ has a same distribution as $X$, and $X$, $X_i$ are independent. In comparison with Section \ref{sec:em.temp.ldp}, here we suppose that the total energy $\tilde Z_n := X + \tilde{Y_n}$ is conditioned on the scale of Gaussian fluctuations. We will show that the {\em temperature} $ \frac{1}{\beta_n \psi(I)} $  is an emergent parameter uniquely determined by a normal distribution $N(0, \sigma^2)$, where $ \sigma^2$ is the variance of $X$. 

\begin{definition}
\label{def:seqofevents}
Let $X$ be a nonnegative and nonconstant continuous random variable with $\mathbb{E}[X^4] < \infty,$ and let $\mu = \mathbb{E} \left[X \right], \ $ $\sigma^2$ be the variance of $X$.   Let $\tilde{Y}_n = \sum_{i=2}^{n} X_i,$  where all random variables in $\{X_i\}_{i=2}^{n}\cup \{X\}$ are i.i.d.. Denote $ \tilde Z_n := X + \tilde Y_n$. For an interval $ I = [d, d+\delta], \ d \in \mathbb{R}, \delta>0$ and a function $\psi: I \rightarrow \mathbb{R}$ such that $0 < \psi(I) < \infty$. Let $\mathbb  P^{(n)}_{I}$ be a sequence of probability measures with density functions $ A_n f_{X}(x) e^{- \frac{\psi(I) }{\sqrt{n}} x} $, where \[ \frac{1}{A_n}=  \int_{\mathbb{R}^+} f_{X}(x) e^{- \frac{\psi(I) }{\sqrt{n}} x}dx \] and let $\mathbb Q^{(n)}_{ I}$ be a sequence of probability measures with density functions $f_{X \mid \tilde Z_n}\left(x;  n\mu  + \sqrt{n} I  \right)$. 
\end{definition}

\begin{theorem}
\label{cor1:canonical temperature}
Denote $ {Y}_n = \frac{\tilde Y_n - (n-1) \mu }{\sqrt{n}},  \ {X}_n = \frac{X}{\sqrt{n}}, \ {Z}_n = {X}_n + {Y}_n,$ and  let $I - \frac{x}{\sqrt{n} } = \left\{ y -  \frac{x}{\sqrt{n}} , \  y \in I \right\}.$  Assume the following conditions hold: 
\begin{enumerate}
\item  \label{cond:cor1:temp state 1} 
$\displaystyle \left| \frac{f_{X \mid  Z_n}\left(x;    I \right)  }{ f_{X}(x)} \right| $ is uniformly bounded on $ \mathbb{R}^+$.

\item  \label{cond:cor1:temp state 2} 
$Y_n \rightarrow Y$ in distribution and  $\displaystyle \frac{\partial \log P(Y \in \left[ y, y+\delta \right] )}{\partial y} \biggr\rvert_{y=d} > 0, \ Y \sim N(0, \sigma^2).$ 
\item  \label{cond:cor1:temp state 3} 
There exists a sequence of functions $g_n : \mathbb{R} \rightarrow \mathbb{R}$ with 
$$ \left| g_n(x)e^{-\frac{\xi}{\sqrt{n}} x} \right|  \ \text{uniformly bounded on} \  \mathbb{R}^+,  \ \text{for any} \  \xi>0,  \quad \text{and} \quad  \mathbb{E}\left[ g_n(X)^2 \right]  \rightarrow 0$$
 such that
\begin{align}
\log \left( \frac{P\left( {Y}_n \in I  -  \frac{x}{\sqrt{n}} \right)  }{P \left( {Z}_n \in  I  \right)} \right)=  \log \left( \frac{P\left(   Y \in I -  \frac{x}{\sqrt{n}}  \right)  }{P \left(   Y \in I  \right)} \right) + \frac{g_n(x)}{\sqrt{n}}  \ \text{on} \ I_n,
\label{cond:CLT:small}
\end{align} 
in which $I_n = [0, d_n]$ with $\displaystyle d_n = O(\sqrt{n}).$
\end{enumerate}
Then 
\begin{align} \label{ifonlyif:gauss}
  n{{D_{\textnormal{KL}}\left(\mathbb P^{(n)}_{I} ~ \|~ \mathbb Q^{(n)}_{I} \right)}} 
 \rightarrow 0  \quad \text{if and only if} \quad  \psi(I) =  \frac{\partial \log P(Y \in \left[ y, y+\delta \right] )}{\partial y} \biggr\rvert_{y=d}.  
\end{align}
\end{theorem}

\begin{remark} As Remark \ref{rmk:int.app.ldp}, 
 the conditions \eqref{cond:cor1:temp state 1} - \eqref{cond:cor1:temp state 3} formulated in  Theorem \ref{cor1:canonical temperature} are technical, so we would like to characterise and verbally describe the underlying meaning and interpretation of them: As Theorem \ref{thm:iidsum}, the condition \eqref{cond:cor1:temp state 1} means that the interaction between $X_n$ and $Y_n$ is regulated; The condition \eqref{cond:cor1:temp state 2} follows from the central limit theorem and we need to choose the  interval $I=[d, d+\delta] \subset \mathbb{R}^-$ such that the partial derivative term is positive; The condition \eqref{cond:cor1:temp state 3} combines the setup $X_n \rightarrow 0$ in probability and $Y_n \rightarrow Y$ in distribution, furthermore, the remainder term has a special form $\frac{g_n(x)}{\sqrt{n}}$.  
\end{remark}

The proof of Theorem \ref{cor1:canonical temperature} is just the application of Theorem \ref{thm:limit.smooth}. We can check that all of the conditions in Theorem \ref{thm:limit.smooth} are satisfied. Here we want to further discuss the equation \eqref{cond:CLT:small} in Condition \eqref{cond:cor1:temp state 3}:

As the proof for Theorem \ref{thm:d+d:iid}, by Corollary \ref{cor:berry-esseen-slusky} of Berry-Esseen theorem and Slusky's theorem, we have
 \begin{align}
\log \left( \frac{P\left( {Y}_n \in I  -  \frac{x}{\sqrt{n}} \right)  }{P \left( {Z}_n \in  I  \right)} \right)=  \log \left( \frac{P\left(   Y \in I -  \frac{x}{\sqrt{n}}  \right)  }{P \left(   Y \in I  \right)} \right) + O\left(\frac{1}{\sqrt{n}} \right) \quad \text{on} \ I_n.
\label{cond:CLT:small2}
\end{align} 
However, it only guarantees the convergence of $\mathbb{P}_I^{(n)}$ and $\mathbb{Q}_I^{(n)}$ in $\| \cdot \|_{\infty}$ by Theorem \ref{thm:d+d:iid}.  Compare Equation \eqref{cond:CLT:small2} with Condition \eqref{cond:cor1:temp state 2},
Theorem \ref{cor1:canonical temperature} requires an explicit form of the remainder:
\begin{align}
\log \left( \frac{P\left( {Y}_n \in I  -  \frac{x}{\sqrt{n}} \right)  }{P \left( {Z}_n \in  I  \right)} \right)=  \log \left( \frac{P\left(   Y \in I -  \frac{x}{\sqrt{n}}  \right)  }{P \left(   Y \in I  \right)} \right) + \frac{ g_n(x) }{\sqrt{n}}\quad \text{on} \ I_n,
\end{align} 
and $\mathbb{E}[g_n(X)^2] \rightarrow 0$. This explicit form of remainder guarantees the ``if and only" if statement \eqref{ifonlyif:gauss}.

We now discuss the connection between Theorem \ref{cor1:canonical temperature} and Zabell's Theorem \ref{thm:zabell.1980}. If the condition is on the scale of Gaussian fluctuations, Theorem \ref{thm:zabell.1980} only tells us that the sequence of conditional distributions $F_{X \mid \tilde Z_n } (x; n\mu + \sqrt{n} I)$ should converge to its unconditional distribution $F_X(x)$. By our theorem \ref{cor1:canonical temperature}, we have an explicit formula for the canonical distribution to approximate the conditional distribution well:   $$F_{X \mid \tilde Z_n } (x; n\mu + \sqrt{n} I) \approx \int_{-\infty}^x A_n f_X(s)e^{-\frac{\psi(I)}{\sqrt{n}} s} dx,$$  for a sufficiently large $n$, and it converges to $F_X(x)$ as $n \rightarrow \infty$ which is consistent with Zabell's Theorem \ref{thm:zabell.1980}. In addition, the parameter $ \frac{\psi(I)}{\sqrt{n}} $ of the canonical distribution is uniquely determined if we require that the approximation is ``good" enough, i.e. the KL-divergence of the conditional distribution from the canonical distribution converges to zero in the rate $ o\left(\frac{1}{n}\right)$.

\subsection{Mathematical definitions of the heat bath} \label{sec:heatbath}
In Section \ref{ch:main}, we provided two limit theorems of a sequence of conditional probabilities to derive a unique canonical distribution as an emergent phenomenon. In Theorem \ref{thm:limit.smooth}, the emergent parameter in the exponential weight is uniquely determined by the limiting distribution of the heat bath $Y_n \rightarrow Y$ (note that in Theorem \ref{thm:limit.smooth}, $Y_n$ follows from the appropriate shifting and scaling of the original heat bath $\tilde{Y}_n$) evaluated on the interval $I = [h, h+\delta]$ such as
$$ \psi(I) =  \frac{\partial \log P(Y \in [y, y+\delta])}{\partial y}\biggr\rvert_{h}. $$
Similarly, in Theorem \ref{thm:limit.ldp}, the emergent parameter in the exponential weight is uniquely determined by the large-deviation rate function of the heat bath $Y_n \rightarrow \mu$ (note that in Theorem \ref{thm:limit.ldp}, $Y_n$ follows from the appropriate shifting and scaling of the original heat bath $\tilde{Y}_n$) evaluated on the interval $I = [h, h+\delta]$ such as
$$ \varphi(I) =  - \phi(y^*), $$
where $\phi$ is the rate function of $Y_n$ and $y^* = \{y: \inf_{y\in I} \phi(y) \}.$ 

If we choose an interval $I' \subset I$, the parameter in the exponential weight may depend on $I'$ in both of the limit theorems. However, since $I'$ is just a subinterval of $I$, we expect that a well-defined {\em heat bath} should give rise to an {\em invariant temperature} of the canonical distribution by giving a constant parameter in the exponential weight no matter what subinterval $I'$ we choose for it. In this section, we discuss two cases that follow from Theorem \ref{thm:limit.smooth} and Theorem \ref{thm:limit.ldp}, respectively. Given a finite interval $I$, we first define the {\em subinterval invariant property} of a sequence of conditional distributions, then we provide three equivalent properties: (1) the subinterval invariant property of a sequence of conditional distributions (2) the invariant temperature property of the canonical distribution (3) the heat-bath property.  Based on the equivalence of these three properties, we truly define the concept of ``heat bath" in the language of mathematics.

Recall that $X,$ $\tilde{Y}_n,$ and $\tilde{Z_n} := X +  \tilde{Y}_n, $ are random variables from the definitions in Section \ref{ch:main}. By proper shifting and scaling, let  $X_n := \beta_n X$, $Y_n := \beta_n \left( \tilde{Y}_n - \mu_n \right)$, and $Z_n := X_n + Y_n$, where $\mu_n$, $\beta_n$ are positive sequences and $\beta_n=o(1)$.

For a finite interval $ I = [h, h + \delta], \ h \in \mathbb{R}$ and $\delta>0$, let ${\mathbb Q}^{(n)}_I$ be a sequence of probability measures with density functions $f_{X \mid {Z}_n } \big(x  ;  I \big).$  The sequence of conditional probability measures ${\mathbb Q}^{(n)}_I$ represents our setup for the canonical ensemble, which should have a ``nice" property such that the limiting behaviors of ${\mathbb Q}^{(n)}_{I'}$ and ${\mathbb Q}^{(n)}_I$ are the same for all subintervals  $I' \subset I$.
Hence we define this ``nice" property as follows:
\begin{definition}
Note that $\delta\left(\cdot, \cdot\right)$ represents the total variation distance of two probability measures. For any given interval $I' \subset I$,
 \begin{align} \label{heatbath:condition2}
\frac{ \delta \left(  \mathbb Q^{(n)}_{I'}, \mathbb Q^{(n)}_I \right) }{\alpha_n} \rightarrow  0, 
\end{align}
in which we take $\alpha_n = \beta_n $ for  Theorem \ref{thm:limit.smooth}, and $\alpha_n = 1$ for  Theorem \ref{thm:limit.ldp}. Then we say that the sequence of conditional probability measures $\mathbb Q^{(n)}_{I}$ has the {\em subinterval invariant property} on the interval $I$. 
\end{definition}

We start with our first theorem which follows Theorem \ref{thm:limit.smooth}.  Recall that in Theorem \ref{thm:limit.smooth},  $Y$ is a random variable such that $Y_n \rightarrow Y$ in distribution.
\begin{theorem}\label{thm:heatbath}
For a given interval $ I' = [h', h' + \delta'], \ h'  \in \mathbb{R}$, $\delta'>0$, and $I' \subset I$, and a function $\psi: I' \rightarrow \mathbb{R}$, let $\tilde{\mathbb P}^{(n)}_{I'}$ be a sequence of probability measures with density functions
\begin{align}
\frac{f_X(x) e^{-\beta_n\psi(I')x}}{\displaystyle \int_{\mathbb{R}^+}  f_X(x) e^{-\beta_n \psi(I')x } dx },
\end{align}
where 
\begin{align} \label{heatbath:condition}
\psi(I') =  \frac{\partial \log P(Y \in [y, y+\delta'])}{\partial y}\biggr\rvert_{h'}.
 \end{align}
Assume all of the conditions in Theorem \ref{thm:limit.smooth} hold, then the following three statement are equivalent: 
\begin{enumerate}
    \item \label{state:heatbath1} $\mathbb Q^{(n)}_{I}$ has the subinterval invariant property  on the interval $I$.
    \item \label{state:heatbath2} $\tilde{\mathbb P}^{(n)}_{I'}$ has a unique parameter (the invariant temperature property) such as
$$\psi(I') = \psi(I) \quad \text{for all} \quad  I' \subset I.$$
\item  \label{state:heatbath3} $ Y_n \rightarrow Y$ in distribution and $Y$ is a random variable with a distribution function 
\begin{align}
\label{heatbath.form.gauss}
P(Y \in [h', h' + \delta']) = \alpha(\delta') e^{\psi(I) h'}     \quad \text{for all} \ \left[h',  h' +\delta' \right] \subset I,
\end{align}
where  $\alpha : \mathbb{R}^+ \rightarrow  \mathbb{R}$ is a function. 
\end{enumerate}
\end{theorem}

\begin{proof} \label{proof:heatbath}
Since all of the conditions in Theorem \ref{thm:limit.smooth} hold for all intervals $I'\subset I$ with \eqref{heatbath:condition}, we can obtain that
\begin{align}
\label{heat.bath.kl}
 \lim_{n \rightarrow \infty}\frac{{D_{\textnormal{KL}}\left( \tilde{\mathbb P}^{(n)}_{I'} ~ \|~ {\mathbb Q}^{(n)}_{I'} \right)}}{\beta_n^2}  
 = 0,  \quad \text{for all} \ I' \subset I.
\end{align}

To prove $\left( \eqref{state:heatbath1} \Rightarrow \eqref{state:heatbath2} \Rightarrow \eqref{state:heatbath3}\right)$: Assume the invariant temperature property holds, by applying the triangle inequality and Pinsker's inequality to Equation \eqref{heat.bath.kl} and the assumption of the subinterval invariant property \eqref{heatbath:condition2} with $\alpha_n = \beta_n$, we have that 
\begin{align}
 \frac{ \delta \left(  \tilde{\mathbb P}^{(n)}_{I'} , \tilde{\mathbb P}^{(n)}_I \right) }{\beta_n} \rightarrow  0, \quad \text{for all} \  I' \subset I.
\end{align}
Following every step from \eqref{cor:temp:limintegral} to \eqref{cor:temp:end} in the proof \ref{proof:lemma:unique2} for Lemma \ref{thm:unique2}, we can get 
\begin{align}
\psi(I') = \psi(I), \quad \text{for all} \  I' \subset I.
\label{heatbath:temp.invariant}
\end{align}
By \eqref{heatbath:condition} and \eqref{heatbath:temp.invariant}, we have
\begin{align}
\frac{\partial \log P(Y \in [y, y+\delta'])}{\partial y}\biggr\rvert_{h'} \equiv  \psi(I),  \quad \text{for all} \  [h', h'+\delta'] \subset I, 
\end{align} 
which implies $Y$ has a distribution 
$$  P(Y \in [h', h' + \delta']) = \alpha(\delta') e^{\psi(I) h'},     \quad \text{for all} \ \left[h',  h' +\delta' \right] \subset I, $$
with some function  $\alpha : \mathbb{R}^+ \rightarrow  \mathbb{R}$. \\
To prove $\left( \eqref{state:heatbath3} \Rightarrow \eqref{state:heatbath2} \Rightarrow \eqref{state:heatbath1}\right)$: By the assumption \eqref{state:heatbath3} that
$$  P(Y \in [h', h' + \delta']) = \alpha(\delta') e^{\psi(I) h'},     \quad \text{for all} \ \left[h',  h' +\delta' \right] \subset I, $$
with some function  $\alpha : \mathbb{R}^+ \rightarrow  \mathbb{R}$, and the equation \eqref{heatbath:condition},
we can obtain that
\begin{align}
\psi(I') = \psi(I), \quad \text{for all} \quad  I' \subset I,
\end{align}
therefore, it implies
\begin{align}
\frac{ \delta \left(  \tilde{\mathbb P}^{(n)}_{I'} , \tilde{\mathbb P}^{(n)}_I \right) }{\beta_n} =  0, \quad \text{for all} \  I' \subset I. \label{heatbath:temp.invariant1.5}
\end{align}
By applying the triangle inequality and Pinsker's inequality to \eqref{heatbath:temp.invariant1.5} and \eqref{heat.bath.kl}, we have
 \begin{align}
\frac{ \delta \left(  \mathbb Q^{(n)}_{I'} , \mathbb Q^{(n)}_I \right) }{\beta_n} \rightarrow  0, \quad \text{for all} \  I' \subset I.
\end{align}
\end{proof}

Next, we continue our analysis based on Theorem \ref{thm:limit.ldp}. Recall that in Theorem \ref{thm:limit.ldp}, $Y_n \rightarrow \mu$, for some constant $\mu$, in probability and the sequence of laws of $Y_n$ satisfies a large deviation principle with speed $1/\beta_n$ and rate function $\phi$. The rate function $\phi \in C^2(D),$ where $D$ is an open interval containing $I$, and
\begin{align}
\label{heat.bath.ldp.cond}
    -\infty < \phi'(y) < 0, \ \text{for all} \ y \in I.
\end{align}

\begin{theorem}\label{thm:heatbath2}
For a given interval $ I' = [h', h' + \delta'], \ h',\delta' \in \mathbb{R}$, $\delta'>0$, and $I' \subset I$, and a function $\varphi: I' \rightarrow \mathbb{R}$,  let $\mathbb{P}_{I'}$ be a probability measure with density function
\begin{align}
\frac{f_X(x) e^{-\varphi(I')x}}{\displaystyle \int_{\mathbb{R}^+}  f_X(x) e^{-\varphi(I')x } dx },
\end{align}
where
\begin{align} \label{heatbath:condition3}
\varphi(I') =  -\phi'(\hat{y}^*), \quad \hat{y}^*  = \{y: \inf_{y\in I'} \phi(y) \}.
 \end{align}
 Assume all of the conditions in Theorem \ref{thm:limit.ldp} hold, then the following three statements are equivalent: 
\begin{enumerate}
    \item \label{state:heatbathldp1} $\mathbb Q^{(n)}_{I}$ has subinterval invariant property   on the interval $I$.
    \item \label{state:heatbathldp2} $\mathbb{P}_{I'}$ has a unique parameter (invariant temperature property) such as
$$\varphi(I') = \varphi(I) \quad \text{for all} \quad  I' \subset I.$$
\item  \label{state:heatbathldp3} Let $\phi$ be the large deviation rate function of $Y_n $. $\phi$ is a linear function such as
\begin{align} \label{heatbath:temp.invariant3}
\phi(y) = \phi'(y^*) y + c, \quad \text{for all} \ y \in I,
\end{align}
where ${y}^*  = \{y: \inf_{y\in I} \phi(y) \}$ and $c$ is some constant. 
\end{enumerate}

\end{theorem}

\begin{proof}
Since all of the conditions in Theorem \ref{thm:limit.ldp} hold for all intervals $I'\subset I$ with \eqref{heatbath:condition3},  we can obtain that 
\begin{align}
    \label{heat.bath.kl.2}
 \lim_{n \rightarrow \infty}{D_{\textnormal{KL}}\left( {\mathbb P}_{I'} ~ \|~ {\mathbb Q}^{(n)}_{I'} \right)} 
 = 0, \quad \text{for all} \ I' \subset I. 
\end{align}

We first show $\left( \eqref{state:heatbathldp1} \Rightarrow \eqref{state:heatbathldp2} \Rightarrow \eqref{state:heatbathldp3} \right)$.  The proof  of $\left( \eqref{state:heatbathldp1} \Rightarrow \eqref{state:heatbathldp2}  \right)$ follows from the proof of $\left( \eqref{state:heatbath1} \Rightarrow \eqref{state:heatbath2}  \right)$ in Theorem  \ref{thm:heatbath}, then
we can get 
\begin{align}
\varphi(I') = \varphi(I), \quad \text{for all} \quad  I' \subset I.
\label{heatbath:temp.invariant2}
\end{align}
By \eqref{heatbath:condition3} and \eqref{heatbath:temp.invariant2}, 
$$ \phi'(\hat{y}^*) =  \phi'(y^*), \quad  \text{for all} \ \hat{y}^*  = \{y: \inf_{y\in I'} \phi(y) \} \ \text{with} \ I' \subset I. $$
With the assumption \eqref{heat.bath.ldp.cond}: $ -\infty < \phi'(y) < 0, \ \text{for all} \ y \in I$, and the properties of the rate function $\phi$ in Appendix \ref{prop.of.ratefunc}, we have that 
$$ \phi'(y) \equiv  \phi'(y^*), \quad \text{for all} \ y \in I, $$
which implies 
\begin{align} \label{heatbath:temp.invariant2.5}
\phi(y) = \phi'(y^*) y + c, \quad \text{for all} \ y \in I,
\end{align}
where $c$ is some constant.

Next we prove $\left( \eqref{state:heatbathldp3} \Rightarrow \eqref{state:heatbathldp2} \Rightarrow \eqref{state:heatbathldp1} \right)$.
Equation \eqref{heatbath:temp.invariant3} implies 
$$\phi'(y) \equiv  \phi'(y^*), \quad \text{for all} \ y \in I, $$
then we can obtain 
$$\phi'(\hat{y}^*)  = \phi'(y^*), \quad \text{for all}  \ \hat{y}^*  = \{y: \inf_{y\in I'} \phi(y) \} \ \text{with} \ I'\subset I.  $$
With \eqref{heatbath:condition3}, it implies 
$$  \varphi(I') = \varphi(I) \quad \text{for all} \ I' \subset I. $$ Then the proof of  $\left( \eqref{state:heatbathldp2} \Rightarrow \eqref{state:heatbathldp1}  \right)$ follows from the proof of $\left( \eqref{state:heatbath2} \Rightarrow \eqref{state:heatbath1}  \right)$ in Theorem  \ref{thm:heatbath}.
\end{proof}
\begin{remark}
The formula \eqref{heatbath.form.gauss} for the third property (it is called the heat-bath property) in Theorem \ref{thm:heatbath} provides the precise formulation of what a heat bath is in probabilistic terms when the heat bath $Y_n$ converges to $Y$ on the scale corresponding to Theorem \ref{thm:limit.smooth}; Similarly, the formula \eqref{heatbath:temp.invariant3} for the third property in Theorem \ref{thm:heatbath2} provides the precise formulation of what a heat bath is in probabilistic terms when the heat bath $Y_n$ converges to a constant $\mu$ on the scale corresponding to Theorem \ref{thm:limit.ldp}. Through these formulations and the equivalence of the three properties: (1) the subinterval invariant property (2) the invariant temperature property (3) the heat-bath property, we really define an invariant temperature bath mathematically.  
\end{remark}

\section{Appendix}

\subsection{Properties of the large deviation rate function }    \label{prop.of.ratefunc}
We include the following properties from \cite{dembo1998large}.
Let $\mathcal{L}$ be the law of $X_1$, let $\mu := \mathbb{E}[X_1] $ and $\sigma^2 := \text{Var}(X_1)$ and assume that $\sigma > 0$. Let \[y_- := \inf (\textnormal{supp}(\mathcal{L})), \quad  y_+ := \sup (\textnormal{supp}(\mathcal{L}))\] and $\phi$ be the function defined in Theorem \ref{thm:cramer}. Define
\begin{align}
\mathcal{D}_\phi := \{ y\in \mathbb{R}: \phi(y) < \infty   \} \quad \text{and} \quad \mathcal{U}_\phi := int(\mathcal{D}_\phi).
\end{align} 
Then the following holds:
\begin{enumerate}
\item $\phi(y)$ is convex and lower semi-continious.
\item $0 \leq \phi(y) \leq \infty$ for all $y \in \mathbb{R}$.
\item $\phi(y) = 0$ if and only if $y = \mu$.
\item $\mathcal{U}_\phi = (y_-, y_+)$ and $\phi(y)$ is infinitely differentiable on $\mathcal{U}_\phi$. 
\item  $\phi''(y) > 0$ on $\mathcal{U}_\phi$ and $\phi''(\mu) = 1/\sigma^2$.
\end{enumerate}

\subsection{Proof of Corollary \ref{cor:berry-esseen-slusky}}
\label{appendix:prof.cor.berry}
\begin{proof}
 \eqref{pre:cor:slusky}  follows from Theorem \ref{thm:slutsky} since $Z_n \rightarrow G$ in distribution and $W_n \rightarrow 0$ in probability. \eqref{pre:cor:berry}  basically follows from the proof for Berry-Esseen Theorem (see for example Theorem 2.2.8. in \cite{tao2012topics}).  We include a sketch of the  proof here. 
 
 Let $\phi_Y$ be the charateristic function of a random variable $Y$ and $\epsilon = \mathbb{E} | X |^3/   \sqrt{n} $. To prove \eqref{pre:cor:berry}, following every step in the proof given in \cite{tao2012topics}, it sufficies to show that  
\begin{align}
\int_{ |t| < c/\epsilon} \frac{|\phi_{ \tilde{Z_n}  } (t) - \phi_G(t)    |   }{  1 + |t| } dt = O(\epsilon),
\label{eq:cor:berry.esseen}
\end{align}
for some small constant $c$. We can show that 
\begin{align}
\left|\phi_{ \tilde{Z_n}  } (t) - \phi_G(t)   \right| &= \left| \exp{ \bigg[ \frac{- t^2}{2}\left(  \frac{n+k}{n} \right) + O\left(  \epsilon |t|^3 \left( \frac{n+k}{n} \right) \right) \bigg] }  - \exp( {-t^2/2}) \right|   \notag \\
& = O \left(  \frac{t^2}{n} \exp   ( -t^2/4)       \right) + O\left( \epsilon |t|^3 \exp   ( -t^2/4)    \right).
\label{eq:twoterms.berry.esseen}
\end{align}
Inserting this to \eqref{eq:cor:berry.esseen}, after integration, the first term  in \eqref{eq:twoterms.berry.esseen} has order $O(\frac{1}{n})$ and the second term has order $O( \epsilon )$. It completes the proof. 
\end{proof}

\subsection{Proof of the boundedness of Equation \eqref{normal.dis.bound2}}
\label{appendix:prof.normal}
Denote that
\begin{align} \label{mean.value}
  A(y) := \frac{\partial^{2} \log   P\big( {{Y}} \in [y, y+\delta]\big) }{ \partial y^2} = \frac{f'_Y(y+\delta) - f'_Y(y) }{P\big( {{Y}} \in [y, y+\delta]\big)} - \left( \frac{f_Y(y+\delta) - f_Y(y) }{P\big( {{Y}} \in [y, y+\delta]\big)}\right)^2.  
\end{align}
We can recognize that
$$   A(h - \hat{\alpha}_n x ) =  2 q_n(x),  $$ 
in which  the function $q_n(x)$ is defined in Equation \eqref{taylor.log} for the proof of Theorem \ref{thm:theorem 1}.

In the entire proof of Theorem \ref{thm:theorem 1}, the only place that we use the condition \eqref{eq:logsquare} regarding uniformly bounded $A(y)$  when $y \in \mathbb{R}$ is just for the proof of Equation \eqref{taylor-qn2} to show that  $\exp(-\psi(I)x ) \cdot q_n(x)$ is uniformly bounded on $x\in \mathbb{R}^+$. Therefore, instead of proving uniformly bounded $A(y)$ in the condition \eqref{eq:logsquare}, it suffices to show the uniform boundedness of  $\exp(-\psi(I)) \cdot q_n(x)$: there exists a constant $C$ such that 
\begin{align} \label{mean.value3}
     \left|\exp(-\psi(I)x) \cdot A(h - \hat{\alpha}_n x ) \right| \leq C,  \quad \hat{\alpha}_n \in (0,1), \quad \text{for  all} \ x \in \mathbb{R^+}.
\end{align}
By the mean value theorem and the formula of the density of normal distribution, we can show that there exists $\hat{y}$, $\hat{y} \in (y, y+\delta) $ such that the first term on the right side of \eqref{mean.value} can be written as  
\begin{align} \label{mean.value1}
    \frac{f'_Y(y+\delta) - f'_Y(y) }{P\big( {{Y}} \in [y, y+\delta]\big)} &= (y + \delta) \exp \left[ \frac{ - y^2 + \hat{y}^2}{2\sigma^2}  \right] \left( \exp \left[ \frac{- 2y\delta - \delta^2}{2\sigma^2} \right] - 1 \right) + \delta \exp\left[ \frac{-y^2 + \hat{y}^2}{2\sigma^2} \right] \notag \\
    &=(y+\delta)\exp\left[ \frac{(\hat{y}-y-\delta)(\hat{y}+y+\delta)}{2\sigma^2}\right]-y\exp\left[ \frac{(\hat{y}-y)(\hat{y}+y)}{2\sigma^2} \right].
\end{align}
Recall $y <\hat{y}<y+\delta$. When $\hat{y}+y+\delta \in [0,2h+\delta]$, \eqref{mean.value1} is uniformly bounded. 
When $\hat{y}+y+\delta<0$, we can further have 
\begin{align*}
\left|\frac{f'_Y(y+\delta) - f'_Y(y) }{P\big( {{Y}} \in [y, y+\delta]\big)} \right|\leq  (h+\delta)  \exp\left[ \frac{-\delta (2y+\delta)}{2\sigma^2}\right]+h\exp\left[ \frac{-\delta y}{\sigma^2} \right]. 
\end{align*}
Therefore
\begin{align}\label{eq:682}
 \exp(-\psi(I)x)   \left|\frac{f'_Y(y+\delta) - f'_Y(y) }{P\big( {{Y}} \in [y, y+\delta]\big)} \right|\leq  [(h+\delta)  \exp(-\delta^2/2\sigma^2)+h]\cdot \exp\left[ \frac{-\delta y}{\sigma^2}-\psi(I)x\right]. 
\end{align}
By plugging in $y = h - \hat{\alpha}_n x$, $\hat{\alpha}_n \in (0,1)$ in \eqref{eq:682},  since we have $2\delta / \sigma^2 < \psi(I)$ from the assumptions in Theorem \ref{thm:d+d:iid} , we can check the terms on the right hand side in \eqref{eq:682} is uniformly bounded when $x\in \mathbb R^+$.

The second term on the right side of \eqref{mean.value} can be written as  
\begin{align} \label{mean.value2}
\left( \frac{f_Y(y+\delta) - f_Y(y) }{P\big( {{Y}} \in [y, y+\delta]\big)}\right)^2 &= \exp \left[ \frac{-y^2 + \hat{y}^2}{\sigma^2}   \right] \left(  \exp \left[ \frac{- 2y\delta - \delta^2}{2\sigma^2} \right] - 1\right)^2.
\end{align}
When $y+\hat{y}\in [0,2h+\delta ],$ the right hand side above is uniformly bounded.  When $y+\hat{y}<0$, from \eqref{mean.value2} we have  
\begin{align} 
\exp(-\psi(I)x)  \left( \frac{f_Y(y+\delta) - f_Y(y) }{P\big( {{Y}} \in [y, y+\delta]\big)}\right)^2 &\leq \exp \left[ \frac{(\hat{y}-y)(\hat{y}+y)}{\sigma^2} -\psi(I)x  \right] \left(  \exp \left[ \frac{- 2y\delta - \delta^2}{2\sigma^2} \right] - 1\right)^2 \notag\\
 &\leq \exp(-\psi(I)x)\left(  \exp \left[ \frac{- 2y\delta - \delta^2}{2\sigma^2} \right] - 1\right)^2 \notag\\
 &=\exp(-\psi(I)x)\left( \exp \left[ \frac{- 2y\delta - \delta^2}{\sigma^2} \right]-2\exp \left[ \frac{- 2y\delta - \delta^2}{2\sigma^2} \right]+1\right).\label{eq:69}
\end{align}
By plugging in $y = h - \hat{\alpha}_n x$, $\hat{\alpha}_n \in (0,1)$ in \eqref{eq:69},  since we have $2\delta / \sigma^2 < \psi(I)$,  we can check the terms on the right hand  is uniformly bounded when $x\in \mathbb R^+$. Therefore, combining the estimates in two parts, \eqref{mean.value3} is uniformly bounded for all $x \in \mathbb{R}^+$.

\subsection{Proof of Corollary \ref{cor:corollary 1}, Corollary \ref{cor:limit.smooth}, and Corollary \ref{cor:limit.ldp}} 
\label{proof:cor:corollary 1}
\subsubsection{Proof of Corollary \ref{cor:corollary 1}}
This proof basically follows the proof in Section \ref{proof:theorem 1} for Theorem \ref{thm:theorem 1}, so we only provide the details of the difference here.  For the derivation of Equation \ref{eq:taylorexpansion}, we do Taylor's expansion with respect to $x$ and $y$  for this corollary, so we will get Equations \eqref{ansequence} - \eqref{ansequence3} as following:
\begin{align} \label{cor:ansequence}
	\phi_n(I) &=  \frac{\partial \log  P \big( Y_n \in [y, y+\delta] \mid  X_n = 0 \big)}{\partial y}\biggr\rvert_{y=h} - \frac{\partial \log  P \big( Y_n \in [y, y+\delta] \mid  X_n = 0 \big)}{\partial x}\biggr\rvert_{y=h}, \notag  \\
r_n(x) =& \frac{1}{2}  \frac{\partial^2  P \big( Y_n \in [y, y+\delta] \mid X_n= 
\xi \big)}{\partial y^2}\biggr\rvert_{y = h-\alpha_n x, \ \xi = \alpha_n x}  \notag\\
&+  \frac{1}{2} \frac{\partial^2  P \big( Y_n \in [y, y+\delta] \mid X_n= \xi \big)}{\partial \xi^2}\biggr\rvert_{y = h-\alpha_n x, \ \xi = \alpha_n x} \notag\\
&- \frac{\partial^2  P \big( Y_n \in [y, y+\delta] \mid X_n= \xi \big)}{\partial \xi \partial y}  \biggr\rvert_{y = h-\alpha_n x, \ \xi = \alpha_n x}. 
\end{align}
With the remaining term
\begin{align}
 k_n(x) &=  \frac{r_n(x)}{ P \big( Y_n \in [h, h+\delta] \mid X_n =  0 \big)}   -  \frac{\phi_n(I)^2e^{- \gamma_n \cdot \phi_n(I) x }}{2} \notag,  
\end{align}
for some $\alpha_n, \gamma_n \in (0,1)$. Then we obtain  \begin{align} \label{condprob2}
f_{ X_n | Z_n}(x; I)  = \frac{  f_{X_n} (x)  P(Y_n \in I  \mid X_n = 0) ( e^{-\phi(I)x} + k_n(x)x^2 ) }{  P(Z_n \in I) } , \quad \text{for} \ x \in \mathbb{R}^+.
\end{align} 

Based on these new expressions of Equations \eqref{ansequence} - \eqref{ansequence3}, Equation \eqref{small2} becomes
\begin{align} \label{condprob4}
 & \left| \int_{\mathbb{R}^+}  A_n f_{X_n}(x)e^{-\phi_n(I)}  \log \left( \frac{f_{X_n}(x) P_{Y_n | X_n} \big( I - x; x \big)  }{P_{Z_n}\big(I \big) }\cdot \frac{1}{ A_n f_{X_n}(x) e^{-\phi_n(I) x}}\right)dx \right| \notag\\
 = & \left| \int_{\mathbb{R}^+}  A_n f_{X_n}(x)e^{-\phi_n(I)}  \log \left( \frac{ P_{Y_n | X_n} \big( I - x; x \big)  }{P_{Y_n | X_n} \big( I ; 0 \big) e^{-\phi_n(I) x} }\cdot \frac{P_{Y_n | X_n} \big( I ; 0 \big)}{ P_{Z_n}\big(I \big) A_n }\right)dx \right| \notag\\
\leq &  \left|  \log \left( \frac{ B_n }{ A_n } \right)  \right| +\left| \int_{\mathbb{R}^+}A_n f_{X_n}(x)e^{-\phi_n(I)} \log \left( \frac{ P_{Y_n | X_n} \big( I -x ; x \big)}{  P_{Y_n | X_n} \big( I ; 0 \big) e^{-\phi_n(I) x}  }\right) dx \right|,
\end{align}
where 
\begin{align}
    A_n := \frac{1}{\int_{\mathbb{R}^+} f_{X_n}(x)e^{-\phi(I)x}dx} \quad \text{and} \quad  B_n := \frac{P_{Y_n | X_n} \big( I ; 0 \big) }{P_{Z_n}\left(I \right) }.
\end{align}
From the expression of $f_{X_n|Z_n}\big(x; I  \big)$ in  (\ref{condprob2}), we have the following identity
\begin{align}
\label{cor:sum2.2}
 1&= \int_{\mathbb{R}^+} f_{X_n|Z_n} \big(x; I \big) dx =  \frac{B_n }{A_n }  + B_n \int_{\mathbb{R}^+} f_{X_n}(x) k_n(x) x^2 dx.
\end{align}
Equation \eqref{cor:sum2.2} implies
\begin{align}\label{cor:log}
\log \left( \frac{B_n}{A_n}\right) = \log \left(1- B_n\int_{\mathbb{R}^+} f_{X_n}(x) k_n(x)x^2dx \right).
\end{align}
Now it remains  to show 
\begin{align}\label{Bhn}
	\left|B_n \int_{\mathbb{R}^+} f_{X_n}(x) k_n(x)x^2dx \right|
\end{align}
is small for large $n$.

By the conditions in Corollary \ref{cor:corollary 1}, $P_{Z_n}(I) \geq \delta_2 > 0$, hence there exists a constant $M_1>0$ such that
\begin{align}
B_n = \frac{P_{Y_n|X_n}(I;0)}{P_{Z_n}(I)} \leq  M_1.
\label{cor:cond-1}
\end{align} 
And since $k_n(x)$ is uniformly bounded as proof \ref{proof:theorem 1} for Theorem \ref{thm:theorem 1}, with the assumption $\mathbb{E}[X_n^2] = a_n$, we can derive that 
\begin{align}\label{cor:log-2}
\left|B_n\int_{\mathbb{R}^+} f_{X_n}(x) k_n(x)x^2dx \right| \leq M_1 \cdot \sup|k_n(x)| \cdot \mathbb{E}[X_n^2]  = 	 O(a_n).
\end{align}

Recall \eqref{cor:log}, since $\log(1+x)\leq x$ for all $x>-1$, for sufficiently large $n$, we have
\begin{align}
\log \left( \frac{B_n}{ A_n}\right) &=\log \left(1- B_n\int_{\mathbb{R}^+} f_{X_n}(x) k_n(x)x^2dx \right) \leq B_n\int_{\mathbb{R}^+} f_{X_n}(x) k_n(x)x^2dx =  O( a_n),
\label{cor:small-1}
\end{align}
which gives us that the first term in \eqref{condprob2} is in order $O(a_n)$.

The second term in \eqref{condprob4} is also in order $O(a_n)$ which follows from the steps \eqref{taylor.log} - \eqref{small3} in Section \ref{proof:theorem 1}. Therefore, by the definition of KL-divergence \eqref{KLdivergence} and Bayes' theorem for conditional probability and the inequality  \eqref{condprob4}, we finally obtain
\begin{align}
D_{\text{KL}}\left (\hat{{\mathbb P}}^{(n)}_I ~\|~ {\mathbb Q}^{(n)}_I\right) &= 
\left| \int_{\mathbb{R}^+}  A_n f_{X_n}(x)e^{-\phi_n(I)}  \log \left(  \frac{f_{X_n | Z_n} \big( x; I \big)}{ A_n f_{X_n}(x) e^{-\phi_n(I) x}}\right)dx \right| \notag \\
&= \left| \int_{\mathbb{R}^+}  A_n f_{X_n}(x)e^{-\phi_n(I)}  \log \left( \frac{f_{X_n}(x) P_{Y_n | X_n} \big( I - x; x \big)  }{P_{Z_n}\big(I \big) }\cdot \frac{1}{ A_n f_{X_n}(x) e^{-\phi_n(I) x}}\right)dx \right| \notag \\
&= O(a_n).
\end{align}

\subsubsection{Proof of Corollary \ref{cor:limit.smooth} and Corollary \ref{cor:limit.ldp}} 
For the proof of Corollary \ref{cor:limit.smooth}, since $\log G(I; 0)=0$ and $\log G(I; \xi) \in C(\mathbb{R}^+)$ with respect to $\xi$, we can do Taylor's expansion for it at zero to get
\begin{align}
    \log G( I; \beta_n x) =  \frac{\partial \log G( I; \xi)}{\partial \xi}\biggr\vert_0 \beta_n x + O(\beta_n^2 x^2) \quad \text{for} \ x \in \mathbb{R}^+.
\end{align}
And similarly, for the proof of Corollary \ref{cor:limit.ldp}, since $\log R(I; 0)=0$ and $\log R(I; \xi) \in C(\mathbb{R}^+)$ with respect to $\xi$, we can do Taylor's expansion for it at zero to get
\begin{align}
    \log \left( R( I; \beta_n x)\right)^{\frac{1}{\beta_n}} = \frac{1}{\beta_n} \left(  \frac{\partial \log R(I; \xi)}{\partial \xi}\biggr\vert_0 \beta_n x + O(\beta_n^2 x^2) \right) \quad \text{for} \ x \in \mathbb{R}^+.
\end{align}
Then the proof of Corollary \ref{cor:limit.smooth} follows from the proof given in Section \ref{proof:limit.smooth} with an additional linear term $\frac{\partial \log G( I; 0)}{\partial \xi}\beta_n x$ in \eqref{prof:lemma:unique2}; and the proof of Corollary \ref{cor:limit.ldp} follows from the proof in Section \ref{proof:limit.ldp} with an additional linear term $ \frac{\partial \log R( I; 0)}{\partial \xi} x$ in \eqref{prof:lemma:unique1}.

\begin{acknowledgements}
We thank Krzysztof Burdzy, Amir Dembo, Ken Dill, Hao Ge, Hugo Touchette, Jan Van Campenhout, Yue Wang, and Ying-Jen Yang for many helpful discussions. 
\end{acknowledgements}

%
\section*{Conflict of interest}
The authors declare that they have no conflict of interest.

\bibliographystyle{plain}
\bibliography{ref.bib}

\begin{thebibliography}{10}

\bibitem{ackley1985learning}
D.H. Ackley, G.E. Hinton, and R.J. Sejnowski.
\newblock A learning algorithm for {B}oltzmann machines.
\newblock {\em Cognitive science}, 9(1):147--169, 1985.

\bibitem{anderson1972more}
P.W. Anderson.
\newblock More is different.
\newblock {\em Science}, 177(4047):393--396, 1972.

\bibitem{barlow2016comparison}
M.~Barlow, K.~Burdzy, and {\'A}.~Tim{\'a}r.
\newblock Comparison of quenched and annealed invariance principles for random
  conductance model.
\newblock {\em Probability Theory and Related Fields}, 164(3-4):741--770, 2016.

\bibitem{bialek2012biophysics}
W.~Bialek.
\newblock {\em Biophysics: Searching for Principles}.
\newblock Princeton University Press, 2012.

\bibitem{billingsley2013convergence}
P.~Billingsley.
\newblock {\em Convergence of probability measures}.
\newblock John Wiley \& Sons, 2013.

\bibitem{butterfoss2003boltzmann}
G.L. Butterfoss and J.~Hermans.
\newblock Boltzmann-type distribution of side-chain conformation in proteins.
\newblock {\em Protein Science}, 12(12):2719--2731, 2003.

\bibitem{cramer1976century}
H.~Cram{\'e}r.
\newblock A century with probability theory: Some personal recollections.
\newblock {\em The annals of probability}, 4(4):509--546, 1976.

\bibitem{dembo1998large}
A.~Dembo and O.~Zeitouni.
\newblock {\em Large Deviations Techniques and Applications}.
\newblock Applications of mathematics. Springer, 1988.

\bibitem{dembo1996refinements}
A.~Dembo and O.~Zeitouni.
\newblock Refinements of the {G}ibbs conditioning principle.
\newblock {\em Probability theory and related fields}, 104(1):1--14, 1996.

\bibitem{deuschel1991microcanonical}
J.D. Deuschel, D.W. Stroock, and H.~Zessin.
\newblock Microcanonical distributions for lattice gases.
\newblock {\em Communications in mathematical physics}, 139(1):83--101, 1991.

\bibitem{dill2012molecular}
K.A. Dill and S.~Bromberg.
\newblock {\em Molecular Driving Forces: Statistical Thermodynamics in Biology,
  Chemistry, Physics, and Nanoscience}.
\newblock Garland Science, 2012.

\bibitem{dobrushin1977central}
R.L. Dobrushin and B.~Tirozzi.
\newblock The central limit theorem and the problem of equivalence of
  ensembles.
\newblock {\em Communications in Mathematical Physics}, 54(2):173--192, 1977.

\bibitem{donsker1983asymptotic}
M.D. Donsker and S.R.S. Varadhan.
\newblock Asymptotic evaluation of certain {M}arkov process expectations for
  large time. {IV}.
\newblock {\em Communications on Pure and Applied Mathematics}, 36(2):183--212,
  1983.

\bibitem{friedli2017statistical}
A.~Friedli and Y.~Velenik.
\newblock {\em Statistical mechanics of lattice systems: a concrete
  mathematical introduction}.
\newblock Cambridge University Press, 2017.

\bibitem{georgii1995equivalence}
H.O. Georgii.
\newblock The equivalence of ensembles for classical systems of particles.
\newblock {\em Journal of statistical physics}, 80(5-6):1341--1378, 1995.

\bibitem{georgii2011gibbs}
H.O. Georgii.
\newblock {\em Gibbs Measures and Phase Transitions}.
\newblock De Gruyter studies in mathematics. De Gruyter, 2011.

\bibitem{gibbs1902elementary}
J.W. Gibbs.
\newblock {\em Elementary Principles in Statistical Mechanics: Developed with
  Especial Reference to The Rational Foundation of Thermodynamics}.
\newblock C. Scribner's sons, 1902.

\bibitem{huang1975statistical}
K.~Huang.
\newblock {\em Statistical Mechanics}.
\newblock Wiley Eastern, 1975.

\bibitem{jaynes1957information}
E.T. Jaynes.
\newblock Information theory and statistical mechanics.
\newblock {\em Physical review}, 106(4):620, 1957.

\bibitem{jaynes2003probability}
E.T. Jaynes.
\newblock {\em Probability Theory: The Logic of Science}.
\newblock Cambridge university press, 2003.

\bibitem{aleksandr1949mathematical}
A.I. Khinchin.
\newblock {\em Mathematical Foundations of Statistical Mechanics}.
\newblock Courier Corporation, 1949.

\bibitem{landau1958statistical}
L.D. Landau and E.M. Lifshitz.
\newblock {\em Statistical Physics (Course of Theoretical Physics vol 5)}.
\newblock Pergamon Oxford, 1958.

\bibitem{lanford1973entropy}
O.E. Lanford.
\newblock Entropy and equilibrium states in classical statistical mechanics.
\newblock In {\em Statistical Mechanics and Mathematical Problems}, pages
  1--113. Springer, 1973.

\bibitem{lewis1995entropy}
J.T. Lewis, C.E. Pfister, and W.G. Sullivan.
\newblock Entropy, concentration of probability and conditional limit theorems.
\newblock {\em Markov Process. Related Fields}, 1(3):319--386, 1995.

\bibitem{martin1979equivalence}
A.~Martin-L{\"o}f.
\newblock The equivalence of ensembles and the gibbs phase rule for classical
  lattice systems.
\newblock {\em Journal of Statistical Physics}, 20(5):557--569, 1979.

\bibitem{martin1979statistical}
A.~Martin-L{\"o}f.
\newblock {\em Statistical mechanics and the foundations of thermodynamics}.
\newblock Springer-Verlag Berlin Heidelberg, 1979.

\bibitem{miao2016quantifying}
Z.~Miao and Y.~Cao.
\newblock Quantifying side-chain conformational variations in protein
  structure.
\newblock {\em Scientific reports}, 6(1):1--10, 2016.

\bibitem{pinsker1964information}
M.S. Pinsker.
\newblock {\em Information and Information Stability of Random Variables and
  Processes}.
\newblock Holden-Day series in time series analysis. Holden-Day, 1964.

\bibitem{presse2013principles}
S.~Press{\'e}, K.~Ghosh, J.~Lee, and K.A. Dill.
\newblock Principles of maximum entropy and maximum caliber in statistical
  physics.
\newblock {\em Reviews of Modern Physics}, 85(3):1115, 2013.

\bibitem{sanov1958probability}
I.N. Sanov.
\newblock On the probability of large deviations of random variables.
\newblock Technical report, North Carolina State University. Dept. of
  Statistics, 1958.

\bibitem{boltzmann_1877}
K.~Sharp and F.~Matschinsky.
\newblock Translation of {L}udwig {B}oltzmann’s paper `{O}n the relationship
  between the second fundamental theorem of the mechanical theory of heat and
  probability calculations regarding the conditions for thermal equilibrium'.
\newblock {\em Entropy}, 17(4):1971, 2015.

\bibitem{shorack2000probability}
G.R. Shorack.
\newblock {\em Probability for Statisticians}.
\newblock Springer Texts in Statistics. Springer International Publishing,
  2017.

\bibitem{kesten2012random}
D.W. Stroock and O.~Zeitouni.
\newblock Microcanonical distribution, {G}ibbs states, and the equivalence of
  ensembles.
\newblock In Harry Kesten and Rick Durrett, editors, {\em Random walks,
  Brownian motion, and Interacting Particle Systems: A Festschrift in Honor of
  Frank Spitzer}, volume~28 of {\em Progress in Probability}, pages 399--424.
  Springer Science \& Business Media, New York, 2012.

\bibitem{tao2012topics}
T.~Tao.
\newblock {\em Topics in Random Matrix Theory}.
\newblock Graduate studies in mathematics. American Mathematical Society, 2012.

\bibitem{tasaki2018local}
H.~Tasaki.
\newblock On the local equivalence between the canonical and the microcanonical
  ensembles for quantum spin systems.
\newblock {\em Journal of Statistical Physics}, 172(4):905--926, 2018.

\bibitem{terrell1987statistical}
L.H. Terrell.
\newblock {\em Statistical Mechanics: Principles and Selected Applications}.
\newblock Dover Publications, 1987.

\bibitem{touchette2009large}
H.~Touchette.
\newblock The large deviation approach to statistical mechanics.
\newblock {\em Physics Reports}, 478(1-3):1--69, 2009.

\bibitem{touchette2015equivalence}
H.~Touchette.
\newblock Equivalence and nonequivalence of ensembles: Thermodynamic,
  macrostate, and measure levels.
\newblock {\em Journal of Statistical Physics}, 159(5):987--1016, 2015.

\bibitem{van1981maximum}
J.~Van~Campenhout and T.~Cover.
\newblock Maximum entropy and conditional probability.
\newblock {\em IEEE Transactions on Information Theory}, 27(4):483--489, 1981.

\bibitem{vasicek1980conditional}
O.A. Vasicek.
\newblock A conditional law of large numbers.
\newblock {\em The Annals of Probability}, pages 142--147, 1980.

\bibitem{pittphilsci15757}
D.~Wallace.
\newblock Naturalness and emergence.
\newblock \url{http://philsci-archive.pitt.edu/15757/}, February 2019.

\bibitem{xiang2007prediction}
Z.~Xiang, P.J. Steinbach, M.P. Jacobson, R.A. Friesner, and B.~Honig.
\newblock Prediction of side-chain conformations on protein surfaces.
\newblock {\em Proteins: Structure, Function, and Bioinformatics},
  66(4):814--823, 2007.

\bibitem{zabell1980rates}
S.L. Zabell.
\newblock Rates of convergence for conditional expectations.
\newblock {\em The Annals of Probability}, pages 928--941, 1980.

\end{thebibliography}

\end{document}